\documentclass[journal]{IEEEtran}
\usepackage{amsmath}
\usepackage{amsfonts}
\usepackage{mathrsfs}
\usepackage{amssymb}
\usepackage{multirow}
\usepackage{latexsym}
\usepackage{psfrag}
\usepackage{graphicx}
\usepackage{cite}
\usepackage{color}
\usepackage{times,balance}
\usepackage{algorithm,multirow,xcolor}
\usepackage{algorithmicx}
\usepackage{algpseudocode}
\usepackage{longtable}
\usepackage{booktabs}
\usepackage{colortbl}
\usepackage{dsfont}
\usepackage{etoolbox}
\usepackage{subfigure}
\usepackage{tabularx}
\usepackage{booktabs}
\usepackage{multirow}
\usepackage{siunitx}
\usepackage[table,xcdraw]{xcolor}
\definecolor{revisionpurple}{RGB}{125,70,210}

\usepackage{tabularx}
\usepackage{multirow}
\usepackage{makecell}

\definecolor{mygray}{gray}{.9}
\newtheorem{theorem}{\textbf{Theorem}}
\newtheorem{definition}{\textbf{Definition}}

\newtheorem{example}{Example}
\newtheorem{assumption}{\textbf{Assumption}}

\newtheorem{lemma}{\textbf{Lemma}}
\newtheorem{remark}{\textbf{Remark}}

\newtheorem{procedure}{\textbf{Procedure}}
\allowdisplaybreaks[4]
\begin{document}

\title{\textmd{Learning event-triggered controllers for linear parameter-varying systems 
from data \\ (extended version)}}
\author{Renjie Ma, \IEEEmembership{Member, IEEE},
Su Zhang, Wenjie Liu, Zhijian Hu, \IEEEmembership{Senior Member, IEEE}, \\ and 
Peng Shi, \IEEEmembership{Life Fellow, IEEE} 
\thanks{\hspace*{-0.2em}\hspace*{-0.2em}This work was supported 
in part by the National Natural Science Foundation of China under Grants 62503133, 62188101, and W2341048,  
the Natural Science 
Foundation of Heilongjiang Province under Grant LH2024F026, and the Fundamental Research Funds for the Central Universities under Grant
HIT.NSFJG202445.   
}
\thanks{\hspace*{-0.2em}Renjie Ma and Su Zhang are with the State Key Laboratory of 
Robotics and Systems, Harbin Institute of Technology, 
Harbin, 150001, China (e-mails: renjiema@hit.edu.cn,  
suzhang@buaa.edu.cn).
}
\thanks{\hspace*{-0.2em}Wenjie Liu is with the School of 
Electrical and Electronic Engineering, Nanyang Technological 
University, 639798, Singapore 
(e-mail: wenjie.liu@ntu.edu.sg).}
\thanks{\hspace*{-0.2em}Zhijian Hu is with the Department of Aeronautical and 
Aviation Engineering, The Hong Kong Polytechnic University, Hong Kong, 
999077, China
(e-mail: zhijian.hu@polyu.edu.hk).  }
\thanks{\hspace*{-0.2em}Peng Shi is with the College of Engineering and Information Technology, Adelaide University, 
Adelaide SA5005, Australia (e-mail: peng.shi@adelaide. edu.au).}
}

\maketitle

\begin{abstract}
Nonlinear dynamical behaviors in engineering applications 
can be approximated 
by linear-parameter varying (LPV) representations, 
but obtaining precise 
model knowledge to develop a control algorithm is difficult in 
practice. In this paper, 
we develop data-driven control strategies 
for event-triggered linear parameter-varying (LPV) systems 
with stability verifications. 
First, we  theoretically analyze the $\theta$-persistence of excitation condition for LPV systems, 
which leads to feasible 
data-based representations. 
Then, in terms of the available perturbed data, 
we establish the stability certificates for 
event-triggered LPV systems with 
the aid of Petersen's lemma, in the sense of robust control, 
resulting in computationally tractable semidefinite programs whose feasible solutions yield the gain-scheduling controllers. 
Besides, we generalize the 
data-driven event-triggered LPV stabilization methods to the 
case of trajectory tracking, 
and discuss robust 
tracking stability accordingly. 
Finally, we verify the effectiveness 
of our theoretical derivations through numerical examples. 
\end{abstract}


\markboth{IEEE Transactions on Automatic Control}{Shell
\MakeLowercase{\textit{et al.}}: Bare Demo of IEEEtran.cls for IEEE Journals}



\section{Introduction}
\label{sec:introduction}

Complex dynamical systems in engineering applications, such as 
hypersonic aircraft, 
autonomous vehicles, and turbofan engines, inherently exhibit nonlinear characteristics \cite{Mohammadpour2012Control}. 
There exist difficulties in analysis and 
control of such systems, by virtue of computationally 
tractable tools 
from linear time-invariant framework. 
Within this context, linear 
parameter-varying (LPV) representations, 
formulated by a linear system with 
time-varying scheduling signals, which embed nonlinear and/or 
uncertain impacts, can be utilized to address the issues of 
providing the theoretical 
guarantees on system performance \cite{Martin2023Guarantees}. 
However, obtaining model knowledge of LPV systems via  
first-principle analysis or system identification 
entails a substantial computation burden
and lacks theoretical support 
for control synthesis. 
Therefore, developing control strategies for LPV systems directly from data, thereby bypassing complex modeling procedures, is of considerable practical significance.      

\subsection{Related Literature}

Allowing for the measurable and 
time-varying scheduling signals within 
bounded convex sets, LPV systems approximate the linear 
input-output relationships of nonlinear processes 
\cite{Ping2022Tube}. Note 
that such scheduling signal of a certain LPV system 
 is available only at the current instant and remains unknown at future instants, 
thus, LPV synthesis aims to 
develop a control strategy of the same structure, 
such that the closed-loop 
performance are guaranteed over the entire set of 
permissible scheduling 
signals \cite{Pandey2017A}. The use of quadratic Lyapunov 
functions, 
which are affine in the scheduling signal, 
leads to the robust stability 
certificate 
in terms of linear matrix inequalities (LMIs) \cite{Pandey2017A}. 
Normally, the scheduling signals are typically assumed to satisfy bounded parameter variations \cite{Vargas2022Robust} and 
the bounded rates of variations 
\cite{Cox2018Affine}. 
There exists a compromise between computation 
tractability and representation power for LPV systems 
\cite{Gallegos2024Set}. 
The closed-loop performance guarantees can 
lead to an infinite number of LMIs, 
which  make it difficult to find feasible solutions  
for control synthesis. 
Hence, 
the relaxation techniques are preferred for 
computationally tractable 
control algorithms, for example, 
convex-hull relaxation \cite{Mulagaleti2024Parameter}, 
partial-convexity approach \cite{Cox2018Affine}, grid partition 
argument \cite{Cheng2021Modeling}, Finsler's lemma and 
affine annihilators relaxation \cite{Polcz2019Passivity}, 
and slack variable method \cite{Souza2006Robust}, to name a few. 
Apart from high computational demands induced by a large 
number of decision parameters, which are incorporated in these 
convex programs with quadratic performance certificates, 
reducing conservatism of LPV performance analysis needs to be 
considered \cite{Meijer2024Certificates}, since 
the conservative methods 
can lead to ill-conditioning and fail to find a feasible 
solution for gain 
scheduling. Owing to the developments of aperiodic-sampling 
communication transmissions \cite{Han2025Event}, developing event-triggered LPV 
controller to refine communication efficiency, on the 
premise of closed-loop performance guarantees, has 
attracted considerable research attention in recent 
years \cite{Cai2025Dynamic}. 
Existing results rely on LPV model knowledge. How to develop an event-triggered LPV controller directly from data when such model knowledge is unavailable remains an open problem.

Based on the Willems fundamental lemma of 
persistence of excitation (PoE) 
\cite{Willems2005A,He2025From,Liu2024Learning}, 
the parameterization of linear feedback systems from  
data pushes forward the developments of data-dependent control 
algorithms, in which explicit system 
identification is not required 
\cite{Persis2020Formulas}. The critical problem lies in 
involving the 
data-based performance verifications in the control 
algorithms, 
for example, linear quadratic regulator \cite{Dorfler2023On} and 
model predictive control \cite{Berberich2021Data}. 
It is noteworthy that the inevitable presence of 
external perturbations in data and/or physical plant, such as 
noise and disturbance \cite{Ren2024Event,Wang2024Periodic,
Ma2021Sparse,An2024Dynamic,Sen2024On}, leads to the nonequivalent representation 
of actual systems. For the specific bound on perturbation, 
a set of  
systems 
cannot be distinguished from one another in terms of the  
available data. The data batches are obtained within the 
context of 
open-loop experiments and then used for control synthesis 
\cite{Rotulo2022Online}, that targets all systems described 
by the 
same data-based representation. 
This issue motivates one of the main 
research lines of data-driven control in recent years. 
By virtue of linear 
fractional transformation \cite{Berberich2023Combining}, 
full-block matrix S-lemma \cite{Waarde2022From}, 
Petersen's lemma \cite{Bisoffi2022Data}, 
Young's inequality relaxation \cite{Persis2023Learning}, 
and 
Farkas's lemma \cite{Dai2021A}, data-dependent 
performance certificates of closed-loop systems can be 
obtained, leading to either semidefinite programs or 
sum-of-squares programs, and the feasibility of which 
implies an available data-driven control strategy. 
This methodology has been widely utilized to handle 
dissipativity analysis \cite{Koch2022Provably}, reachability 
assessment \cite{Bisoffi2022Learning}, online switching 
mechanism specification \cite{Rotulo2022Online}, 
transfer stabilization 
\cite{Li2024Data}, Koopman-based feedback design 
\cite{Strasser2024Koopman}, integral quadratic 
constraint-based argument \cite{Koch2021Determining}, 
performance assurance \cite{Bianchi2025Data}, and observer 
development \cite{Disaro2024On}. 
Moreover, aperiodic sampling-based 
communication transmissions can be incorporated into 
data-driven control 
algorithms, in which the time span between adjacent sampling 
instants is maximized by a user-defined triggering logic while 
preserving the closed-loop performance, leading to the 
innovative data-based event-triggered controller 
\cite{Persis2024Event} and self-triggered controller 
\cite{Wildhagen2022Data}, holding for linear time-invariant 
systems. Albeit data-driven approaches can successfully     
be applied to LPV systems, in terms of 
input-scheduling-state/output 
data, to cope with the dissipativity analysis based on 
Finsler's lemma \cite{Verhoek2023Data} and the LPV 
control gain scheduling \cite{Verhoek2024Direct}, fewer 
results have investigated the event-triggered control 
design for perturbed LPV systems from data and have extracted the 
$\theta$-PoE prerequisite for 
data-driven control of LPV systems, leaving room for further investigation. 

\subsection{Technical Contributions}
In this paper, we establish the event-triggered control strategies 
for perturbed LPV systems from data, in order to fill the 
gap between 
aperiodic-sampling communication transmissions and 
data-driven LPV control. For the perturbed LPV 
systems, the existing PoE 
argument for the linear time-invariant framework 
\cite{Clouson2022Robust}
is not applicable, this issue has not been thoroughly investigated 
in the innovative result \cite{Verhoek2023Data}. 
Although the recent work \cite{Verhoek2024The} 
has generalized the Willems fundamental lemma 
to the LPV framework, the kernel-based representation cannot 
be directly 
applied to data-driven control synthesis, and only the 
shifted-affine scheduling signal is considered therein.  
Besides, 
the data-based stability analysis and control synthesis 
herein are 
different from the existing work \cite{Verhoek2024Direct}, 
allowing for 
the deployed event-triggered communication mechanisms with 
the channel from controller to actuator and 
the perturbed LPV systems. 
Petersen's lemma and the full-block S-procedure 
are utilized to 
derive the semidefinite programs solvable 
at the vertices of the polytopic scheduling set as adopted in the model-based approaches  
\cite{Abbas2016A}. To some extent, we extend 
the methodology of learning event-triggered control of 
linear time-invariant systems \cite{Persis2024Event,Zhou2024Fuzzy,Liu2024Decentralized} to 
LPV systems. Moreover, the existing literature 
on data-driven LPV control cannot cope with the 
robust output tracking issue directly, and 
in terms of the auxiliary integral compensator technique 
\cite{Golabi2016Event}, we accordingly develop a 
learning event-triggered controller for LPV 
output tracking. The main contributions of this paper 
can be briefly summarized by the following keypoints. 

\noindent $\bullet$ The 
$\theta$-PoE condition for LPV systems is explored to reveal the 
{minimum length of the collected data, which serves as a prerequisite for data-driven analysis and control of perturbed LPV systems.  

\noindent $\bullet$ The data-dependent tracking stability certificates are derived via Petersen's lemma and full-block 
S-procedure, yielding  computationally tractable semidefinite programs that are feasible at the vertices of the polytopic scheduling set. 

\noindent $\bullet$ We develop event-triggered communication mechanisms in the controller-to-actuator channel and explore the relationships between 
the stability certificates and the triggering parameters, leading to the potentials of ETC of LPV systems from data. 

\noindent $\bullet$ By introducing an auxiliary integral compensator, we obtain the augmented LPV systems to cope with trajectory tracking, and 
derive the corresponding data-based stability certificates, yielding solvable semidefinite programs for event-based  
LPV tracking controller synthesis.

\subsection{Outline of This Paper}
The problem formulation and preliminaries for data-driven control of perturbed LPV systems are stated in Section \ref{section2}. The data-based stability certificates 
of event-triggered LPV systems and the event-triggered LPV controller syntheses are discussed in Section \ref{section3}. We perform a group of  
numerical examples to verify the effectiveness of control algorithms 
in Section \ref{section4}, and conclude this paper in Section \ref{section5}.

\subsection{Applicatory Notations}
Most of the applicatory notations  are standard, therefore, 
we only introduce the specific ones to be utilized herein. 
$\mathds{N}_{[\ell_{1},\ell_{2}]}$ represents the integer 
set $\{\ell_{1},\ell_{1}\hspace*{-0.2em}+\hspace*{-0.2em}1,
\cdots\hspace*{-0.2em},\ell_{2}\}$. $\mathds{R}_{+}$ and $\mathds{N}_{+}$ 
imply the positive real number and integer sets, respectively. 
Besides, $\mathbb{S}^{n} (\mathbb{S}_{++}^{n})$ is the set of (positive) 
symmetric square matrices. 
The Kronecker product of $M\hspace*{-0.3em}\in\hspace*{-0.3em}\mathds{R}^{m\times n}$ 
and 
$N\hspace*{-0.3em}\in\hspace*{-0.3em}\mathds{R}^{p\times q}$ is denoted 
by $M\hspace*{-0.2em}\otimes\hspace*{-0.2em}N
\hspace*{-0.3em}\in\hspace*{-0.3em}\mathds{R}^{mp\times nq}$. 
The operator $\mathrm{Col}(X_{1},\cdots\hspace*{-0.2em},X_{\ell})$ with 
$\ell\hspace*{-0.2em}\in\mathds{N}_{+}$ denotes an augmented column vector/matrix, 
in which $X_{\imath}$, 
$\imath\hspace*{-0.2em}\in\hspace*{-0.2em}\mathds{N}_{[1,\ell]}$ has an appropriate  
dimension. 
$\mathrm{Blkdiag}(x)$ represents a diagonal matrix, 
each diagonal element of which is the vector $x$, whereas  
$\hspace*{-0.1em}\mathrm{Blkdiag}(X,\hspace*{-0.1em}Y)$ represents the 
block-diagonal matrix, the diagonal elements of which are 
$X$ and $Y$ in sequence. For a matrix $G$, 
we represent its Moore-Penrose inverse by $G^{\dagger}$ if 
available. 
Given a vector $x_{k}$ and time instants 
$t_{1},t_{2}\hspace*{-0.3em}\in\hspace*{-0.3em}\mathds{N}_{+}$, we 
define 
$\hspace*{-0.1em}\|x_{k}\|_{[t_{1},t_{2}]}$ $=\hspace*{-0.2em}\sup_{k\in[t_{1},t_{2}]}\hspace*{-0.3em}\|x_{k}\|$.  
\hspace*{-0.1em}For a sequence $\hspace*{-0.1em}\{x_{k}\}_{k=0}^{\mathcal{T}-1}$, \hspace*{-0.1em}the Hankel matrix 
$\boldsymbol{\mathcal{H}}_{L}(x_{[0,\mathcal{T}-1]})$ is denoted by 
\begin{eqnarray*}
\boldsymbol{\mathcal{H}}_{L}(x_{[0,\mathcal{T}-1]})\hspace*{-0.2em}=\hspace*{-0.2em}
\left[\hspace*{-0.2em}
\begin{array}{cccc}
  x_{0}  \hspace*{-0.2em}&\hspace*{-0.2em}x_{1} \hspace*{-0.2em}&\hspace*{-0.2em}\cdots \hspace*{-0.2em}&\hspace*{-0.2em}x_{\mathcal{T}-L}\\
  x_{1} \hspace*{-0.2em}&\hspace*{-0.2em} x_{2}  \hspace*{-0.2em}&\hspace*{-0.2em}\cdots \hspace*{-0.2em}&\hspace*{-0.2em}
  x_{\mathcal{T}-L+1}\\
 \vdots \hspace*{-0.2em}&\hspace*{-0.2em} \vdots \hspace*{-0.2em}&\hspace*{-0.2em}\ddots \hspace*{-0.2em}&\hspace*{-0.2em}\vdots\\
  x_{L-1} \hspace*{-0.2em}&\hspace*{-0.2em} x_{L} \hspace*{-0.2em}&\hspace*{-0.2em}\cdots \hspace*{-0.2em}&\hspace*{-0.2em}x_{\mathcal{T}-1}
\end{array}
\hspace*{-0.2em}\right].
\end{eqnarray*}

\section{Problem Formulation}
\label{section2}

\subsection{System Description}

At the sampling instant $k\hspace*{-0.2em}\in\hspace*{-0.2em}\mathds{N}$, 
a discrete-time LPV system can be described by the following state-space representation
\begin{eqnarray}
\hspace*{-0.6em} \boldsymbol{x}_{k+1}\hspace*{-0.3em}& \hspace*{-0.6em}=\hspace*{-0.6em} &\hspace*{-0.3em}
A_{d}(\boldsymbol{p}_{k})\boldsymbol{x}_{k}\hspace*{-0.2em}+\hspace*{-0.2em}B_{d}(\boldsymbol{p}_{k})\boldsymbol{u}_{k}
\hspace*{-0.2em}+\hspace*{-0.2em}\boldsymbol{\omega_{k}}, \notag \\
\hspace*{-0.6em} \boldsymbol{y}_{k}\hspace*{-0.3em}& \hspace*{-0.6em}=\hspace*{-0.6em} &\hspace*{-0.3em}
C_{d}(\boldsymbol{p}_{k})\boldsymbol{x}_{k}
\hspace*{-0.2em}+\hspace*{-0.2em}D_{d}
(\boldsymbol{p}_{k})\boldsymbol{u}_{k}, \label{e1}
\end{eqnarray}
where $\hspace*{-0.1em}\boldsymbol{x}_{k}\hspace*{-0.2em}\in\hspace*{-0.2em}\mathds{R}^{n}$,
$\hspace*{-0.2em}\boldsymbol{u}_{k}\hspace*{-0.2em}\in\hspace*{-0.2em}\mathds{R}^{m}$,
$\hspace*{-0.1em}\boldsymbol{y}_{k}\hspace*{-0.2em}\in\hspace*{-0.2em}\mathds{R}^{r}$, $\hspace*{-0.1em}\boldsymbol{p}_{k}\hspace*{-0.2em}\in\hspace*{-0.2em}\mathds{P}$, and
$\hspace*{-0.1em}\boldsymbol{\omega}_{k}\hspace*{-0.2em}\in\hspace*{-0.2em}\mathds{R}^{n}\hspace*{-0.1em}$ denote the system state, control input, measured output, scheduling signal, and external perturbation, respectively. Accordingly, the mappings
$\hspace*{-0.2em}A_{d}\hspace*{-0.2em}:\hspace*{-0.2em}\mathds{P}\hspace*{-0.3em}\mapsto\hspace*{-0.3em}\mathds{R}^{n\times n}$,
$\hspace*{-0.2em}B_{d}\hspace*{-0.2em}:\hspace*{-0.2em}\mathds{P}\hspace*{-0.3em}\mapsto\hspace*{-0.3em}\mathds{R}^{n\times m}$,
$\hspace*{-0.2em}C_{d}\hspace*{-0.2em}:\hspace*{-0.2em}\mathds{P}\hspace*{-0.3em}\mapsto\hspace*{-0.3em}\mathds{R}^{r\times n}$, and
$D_{d}\hspace*{-0.2em}:\hspace*{-0.2em}\mathds{P}\hspace*{-0.3em}\mapsto\hspace*{-0.3em}\mathds{R}^{r\times m}$ depend affinely on $p_{k}$, namely,
\begin{eqnarray}
\hspace*{-1.0em} A_{d}(\boldsymbol{p}_{k})\hspace*{-0.3em}& \hspace*{-0.6em}=\hspace*{-0.6em} &\hspace*{-0.4em}
A_{d0}\hspace*{-0.2em}+\hspace*{-0.3em}\sum\limits_{\imath=1}^{\ell}\boldsymbol{p}_{k}^{[\imath]}A_{d\imath},
B_{d}(\boldsymbol{p}_{k})\hspace*{-0.3em}=\hspace*{-0.3em}
B_{d0}\hspace*{-0.2em}+\hspace*{-0.3em}\sum\limits_{\imath=1}^{\ell}\boldsymbol{p}_{k}^{[\imath]}B_{d\imath}, \notag\\
\hspace*{-1.0em} C_{d}(\boldsymbol{p}_{k})\hspace*{-0.3em}& \hspace*{-0.6em}=\hspace*{-0.6em} &\hspace*{-0.4em}
C_{d0}\hspace*{-0.2em}+\hspace*{-0.3em}\sum\limits_{\imath=1}^{\ell}\boldsymbol{p}_{k}^{[\imath]}C_{d\imath},
D_{d}(\boldsymbol{p}_{k})\hspace*{-0.2em}=\hspace*{-0.2em}
D_{d0}\hspace*{-0.2em}+\hspace*{-0.3em}\sum\limits_{\imath=1}^{\ell}
\boldsymbol{p}_{k}^{[\imath]}D_{d\imath}, \label{e4}
\end{eqnarray}
where the matrices $\hspace*{-0.1em}A_{d\imath}$, $B_{d\imath}$, $C_{d\imath}$, and $\hspace*{-0.1em}D_{d\imath}$ for $\hspace*{-0.1em}\imath\hspace*{-0.2em}\in\hspace*{-0.2em}\mathds{N}_{[0,\ell]}$
have appropriate dimensions, and the superscript $[\imath]$ denotes the $\imath$-th element of the scheduling signal $\boldsymbol{p}_{k}$. 
By separating the coefficient matrices in \eqref{e4}, we further rewrite the LPV system \eqref{e1} 
in the following form 
\begin{eqnarray}
\hspace*{-0.6em} \boldsymbol{x}_{k+1}\hspace*{-0.3em}& \hspace*{-0.6em}=\hspace*{-0.6em} &\hspace*{-0.3em}
\mathcal{A}_{d}\boldsymbol{z}_{k}\hspace*{-0.2em}+\hspace*{-0.2em}\mathcal{B}_{d}\boldsymbol{\bar{u}}_{k}
\hspace*{-0.2em}+\hspace*{-0.2em}\boldsymbol{\omega}_{k},
\notag \\
\hspace*{-0.6em} \boldsymbol{y}_{k}\hspace*{-0.3em}& \hspace*{-0.6em}=\hspace*{-0.6em} &\hspace*{-0.3em}
\mathcal{C}_{d}\boldsymbol{z}_{k}\hspace*{-0.2em}+\hspace*{-0.2em}\mathcal{D}_{d}\boldsymbol{\bar{u}}_{k}, \label{e6}
\end{eqnarray}
where $\boldsymbol{z}_{k}\hspace*{-0.2em}=\hspace*{-0.2em}\mathrm{Col}(\boldsymbol{x}_{k},\boldsymbol{p}_{k}\hspace*{-0.2em}\otimes\hspace*{-0.2em} \boldsymbol{x}_{k})$,
$\boldsymbol{\bar{u}}_{k}\hspace*{-0.2em}=\hspace*{-0.2em}\mathrm{Col}(\boldsymbol{u_{k}},\boldsymbol{p}_{k}\hspace*{-0.2em}\otimes\hspace*{-0.2em} \boldsymbol{u}_{k})$, and
\begin{eqnarray*}
\hspace*{-0.6em}\mathcal{A}_{d}\hspace*{-0.3em}& \hspace*{-0.6em}=\hspace*{-0.6em} &\hspace*{-0.3em}
[A_{d0},A_{d1},\cdots,A_{d\ell}], \mathcal{B}_{d}\hspace*{-0.2em}=\hspace*{-0.3em}
[B_{d0},B_{d1},\cdots,B_{d\ell}], \\
\hspace*{-0.6em}\mathcal{C}_{d}\hspace*{-0.3em}& \hspace*{-0.6em}=\hspace*{-0.6em} &\hspace*{-0.3em}
[C_{d0},C_{d1},\cdots,C_{d\ell}],\mathcal{D}_{d}\hspace*{-0.2em}=\hspace*{-0.3em}
[D_{d0},D_{d1},\cdots,D_{d\ell}]. 
\end{eqnarray*}
We collect the data from the LPV system  \eqref{e6}, yielding the data set 
$\mathbb{D}_{s}\hspace*{-0.3em}=\hspace*{-0.3em}\{\boldsymbol{u}_{k},\boldsymbol{p}_{k},\boldsymbol{x}_{k},
\boldsymbol{\omega}_{k}\}_{k=0}^{\mathcal{T}}$. Therefore, we can 
construct the following data matrices for 
a specific time window 
$\mathcal{T}\hspace*{-0.2em}\in\hspace*{-0.2em}\mathds{N}_{+}$, 
\begin{subequations}
\label{ee4}
\begin{align}
\hspace*{-0.6em}\boldsymbol{U}& \hspace*{-0.3em}=\hspace*{-0.3em} 
[\boldsymbol{u}_{0},\cdots,\boldsymbol{u}_{\mathcal{T}-1}]\hspace*{-0.2em}\in\hspace*{-0.2em}\mathds{R}^{m\times\mathcal{T}}\hspace*{-0.2em}, \\
\hspace*{-0.6em}\boldsymbol{U}_{\hspace*{-0.2em}P}& \hspace*{-0.3em}=\hspace*{-0.3em} 
[\boldsymbol{p}_{0}\hspace*{-0.2em}\otimes\hspace*{-0.2em}\boldsymbol{u}_{0},\cdots,
\boldsymbol{p}_{\mathcal{T}-1}\hspace*{-0.2em}\otimes\hspace*{-0.2em}\boldsymbol{u}_{\mathcal{T}-1}]\hspace*{-0.2em}\in\hspace*{-0.2em}\mathds{R}^{\ell m\times\mathcal{T}}\hspace*{-0.2em}, \\
\hspace*{-0.6em}\boldsymbol{X}& \hspace*{-0.3em}=\hspace*{-0.3em} 
[\boldsymbol{x}_{0},\cdots,\boldsymbol{x}_{\mathcal{T}-1}]\hspace*{-0.2em}\in\hspace*{-0.2em}\mathds{R}^{n\times\mathcal{T}}\hspace*{-0.2em}, \\
\hspace*{-0.6em}\boldsymbol{X}_{\hspace*{-0.2em}P}& \hspace*{-0.3em}=\hspace*{-0.3em} 
[\boldsymbol{p}_{0}\hspace*{-0.2em}\otimes\hspace*{-0.2em}\boldsymbol{x}_{0},\cdots,
\boldsymbol{p}_{\mathcal{T}-1}\hspace*{-0.2em}\otimes\hspace*{-0.2em}\boldsymbol{x}_{\mathcal{T}-1}]\hspace*{-0.2em}\in\hspace*{-0.2em}\mathds{R}^{\ell n\times\mathcal{T}}\hspace*{-0.2em}, \\
\hspace*{-0.6em}\boldsymbol{X}_{+}& \hspace*{-0.3em}=\hspace*{-0.3em} 
[\boldsymbol{x}_{1},\cdots,\boldsymbol{x}_{\mathcal{T}}]\hspace*{-0.2em}\in\hspace*{-0.2em}\mathds{R}^{n\times\mathcal{T}}\hspace*{-0.2em},\\
\hspace*{-0.6em}\boldsymbol{W}& \hspace*{-0.3em}=\hspace*{-0.3em} 
[\boldsymbol{\omega}_{0},\cdots,\boldsymbol{\omega}_{\mathcal{T}-1}]\hspace*{-0.2em}\in\hspace*{-0.2em}\mathds{R}^{n\times\mathcal{T}}\hspace*{-0.2em}, \\
\hspace*{-0.6em}\boldsymbol{W}_{\hspace*{-0.2em}P}& \hspace*{-0.3em}=\hspace*{-0.3em} 
[\boldsymbol{p}_{0}\hspace*{-0.2em}\otimes\hspace*{-0.2em}\boldsymbol{\omega}_{0},\cdots,
\boldsymbol{p}_{\mathcal{T}-1}\hspace*{-0.2em}\otimes\hspace*{-0.2em}\boldsymbol{\omega}_{\mathcal{T}-1}]\hspace*{-0.2em}\in\hspace*{-0.2em}\mathds{R}^{\ell n\times\mathcal{T}}\hspace*{-0.2em}, \\
\hspace*{-0.6em}\boldsymbol{Y}& \hspace*{-0.3em}=\hspace*{-0.3em} 
[\boldsymbol{y}_{0},\cdots,\boldsymbol{y}_{\mathcal{T}-1}]\hspace*{-0.2em}\in\hspace*{-0.2em}\mathds{R}^{r\times\mathcal{T}}\hspace*{-0.2em},
\end{align}
\end{subequations}
based on which the LPV system \eqref{e6} can then be described by the following data-based representation 
\begin{eqnarray}
\hspace*{-0.6em} \boldsymbol{X}_{+}\hspace*{-0.3em}& \hspace*{-0.6em}=\hspace*{-0.6em} &\hspace*{-0.3em}
\mathcal{A}_{d}\boldsymbol{Z}\hspace*{-0.1em}+\hspace*{-0.1em}\mathcal{B}_{d}\boldsymbol{\overline{U}}\hspace*{-0.1em}+\hspace*{-0.1em}
\boldsymbol{W}, \notag \\
\hspace*{-0.6em} \boldsymbol{Y}\hspace*{-0.3em}& \hspace*{-0.6em}=\hspace*{-0.6em} &\hspace*{-0.3em}
\mathcal{C}_{d}\boldsymbol{Z}\hspace*{-0.1em}+\hspace*{-0.1em}\mathcal{D}_{d}\boldsymbol{\overline{U}}, \label{e8}
\end{eqnarray}
where $\hspace*{-0.1em}\boldsymbol{Z}\hspace*{-0.3em}=\hspace*{-0.3em}\mathrm{Col}(\boldsymbol{X}\hspace*{-0.1em},\hspace*{-0.1em}\boldsymbol{X}_{\hspace*{-0.2em}P})\hspace*{-0.1em}$ and
$\hspace*{-0.1em}\boldsymbol{\overline{U}}\hspace*{-0.3em}=\hspace*{-0.3em}\mathrm{Col}(\boldsymbol{U},\boldsymbol{U}_{\hspace*{-0.2em}P})$. The perturbation $\boldsymbol{W}$ is embedded in the data set $\mathbb{D}_{s}$, and is assumed to satisfy a 
bounded-energy constraint, leading to the 
following standard assumption.

\begin{assumption}
\label{ass1}
For the external perturbation $\boldsymbol{\omega}_{k}$ with $k\hspace*{-0.2em}\in\hspace*{-0.2em}\mathds{N}$, it holds that
$\boldsymbol{\omega}_{k}\hspace*{-0.3em}\in\hspace*{-0.3em}\mathds{B}_{\delta}\hspace*{-0.3em}=\hspace*{-0.3em}\{\boldsymbol{\omega}|\|\boldsymbol{\omega}\|
\hspace*{-0.3em}\leq\hspace*{-0.3em}\delta\}$, where $\delta\hspace*{-0.2em}\in\hspace*{-0.2em}\mathds{R}_{+}$. This further implies 
$\boldsymbol{W}\hspace*{-0.1em}\boldsymbol{W}^{\top}\hspace*{-0.4em}\preceq\hspace*{-0.2em}\Delta\Delta^{\hspace*{-0.1em}\top}\hspace*{-0.1em}$, where 
$\Delta\hspace*{-0.3em}=\hspace*{-0.3em}\sqrt{\mathcal{T}}\delta I_{n}$.
\end{assumption}

\subsection{Heuristic Inspiration}

\begin{definition}
The LPV system \eqref{e8} is said to be 
exponentially input-to-state practically stable, if the 
condition 
\begin{eqnarray}
\hspace*{-0.6em} \|\boldsymbol{x}_{k}\|\hspace*{-0.3em}& \hspace*{-0.6em}\leq\hspace*{-0.6em} &\hspace*{-0.3em}
\alpha_{1}e^{-\alpha_{2}}\|\boldsymbol{x}_{0}\|+\gamma(\|\boldsymbol{\omega}_{s}\|_{[0,k-1]})+\alpha_{3}\nu \label{m1}
\end{eqnarray}
holds for some scalars $\nu, \alpha_{\jmath}\hspace*{-0.2em}\in\hspace*{-0.2em}\mathds{R}_{+},\jmath\hspace*{-0.2em}\in\hspace*{-0.2em}\mathds{N}_{[1,3]}$ and a  $\mathcal{K}$-function $\gamma$. 
Besides, the system \eqref{e8} is exponentially input-to-state stable (ISS) if the inequality \eqref{m1} holds with $\nu\hspace*{-0.2em}=\hspace*{-0.2em}0$. 
\end{definition}

\begin{definition}
\cite{Clouson2022Robust} The sequence 
$\hspace*{-0.1em}\{\boldsymbol{x}_{k}\}_{k=0}^{\mathcal{T}-1}
\hspace*{-0.1em}$ is $\theta$-persistence of excitation of 
order $L\hspace*{-0.2em}\in\hspace*{-0.2em}\mathbb{N}_{+}$, if the minimum singular value of the Hankel matrix 
$\boldsymbol{\mathcal{H}}_{L}
(\boldsymbol{x}_{[0,\mathcal{T}-1]})$ is no less 
than $\theta\hspace*{-0.2em}\in\hspace*{-0.2em}\mathbb{R}_{+}$. 
\end{definition}

The existing PoE condition for LPV systems \cite{Verhoek2024The} is derived for the noise-free scenario, whereas robust data-driven stabilization of perturbed LPV systems requires the effects of perturbations to be explicitly considered. This motivates the following technical lemma.
\begin{lemma}
\label{lemm1}
Suppose that the sequence $\{\mathcal{U}_{k}\}_{k=0}^{\mathcal{T}-1}$ of the LPV system \eqref{e6} is 
$\theta$-PoE of order $L_{o}\hspace*{-0.2em}=\hspace*{-0.2em}(1\hspace*{-0.2em}+\hspace*{-0.2em}\ell)n\hspace*{-0.2em}+\hspace*{-0.2em}1$, 
with $\mathcal{T}\hspace*{-0.2em}\geq\hspace*{-0.2em}\bar{\mathcal{T}}\hspace*{-0.2em}=n(1\hspace*{-0.2em}+\hspace*{-0.2em}\ell)(1\hspace*{-0.2em}+\hspace*{-0.2em}
m(1\hspace*{-0.2em}+\hspace*{-0.2em}\ell))\hspace*{-0.2em}-\hspace*{-0.2em}1$. If the excitation level satisfies 
$\theta\hspace*{-0.2em}>\hspace*{-0.3em}\sqrt{L_{o}}
c_{\omega}\delta/\rho$, where $c_{\omega}$ is defined in \eqref{c4}, $\delta$ is the perturbation bound, and 
$\rho\hspace*{-0.2em}\in\hspace*{-0.2em}\mathbb{R}_{+}$ denotes the system-dependent constant introduced in \cite[Assumption 3.1]{Clouson2022Robust}, 
then   
\begin{eqnarray}
\hspace*{-0.6em} \mathrm{Rank}(\Theta)\hspace*{-0.3em}& \hspace*{-0.6em}=\hspace*{-0.6em} &\hspace*{-0.3em}
(1\hspace*{-0.2em}+\hspace*{-0.2em}\ell)(n\hspace*{-0.2em}+\hspace*{-0.2em}m), \label{e9}
\end{eqnarray}
where $\Theta\hspace*{-0.2em}=\hspace*{-0.2em}\mathrm{Col}(\boldsymbol{U},\boldsymbol{U}_{\hspace*{-0.2em}P},\boldsymbol{X},\boldsymbol{X}_{\hspace*{-0.2em}P})$ 
and $\mathcal{U}_{k}=\mathrm{Col}(\boldsymbol{u}_{k},\boldsymbol{p}_{k}\hspace*{-0.2em}\otimes\hspace*{-0.2em}\boldsymbol{u}_{k})$.
\end{lemma}
\begin{IEEEproof}
\hspace*{-0.2em}We partition the matrix $\Theta\hspace*{-0.3em}=\hspace*{-0.2em}\Theta_{\bar{x}}
\hspace*{-0.2em}+\hspace*{-0.2em}\Theta_{\bar{\omega}}$, where $\Theta_{\bar{x}}\hspace*{-0.3em}=$
$\mathrm{Col}(\boldsymbol{U}\hspace*{-0.1em},\boldsymbol{U}_{\hspace*{-0.2em}P}\hspace*{-0.1em},\overline{\boldsymbol{X}}\hspace*{-0.1em},\overline{\boldsymbol{X}}_{\hspace*{-0.2em}P})$ and
$\Theta_{\bar{\omega}}\hspace*{-0.3em}=\hspace*{-0.2em}\mathrm{Col}(\boldsymbol{0},\boldsymbol{0},\overline{\boldsymbol{W}}\hspace*{-0.1em},\overline{\boldsymbol{W}}_{\hspace*{-0.2em}P})$ holding with
$\overline{\boldsymbol{X}}\hspace*{-0.2em}=\hspace*{-0.2em}[
\bar{\boldsymbol{x}}_{0},\cdots,\bar{\boldsymbol{x}}_{\mathcal{T}-1}]$ and $\bar{\boldsymbol{x}}_{0}\hspace*{-0.3em}=\hspace*{-0.2em}\boldsymbol{x}_{0}$. 
Hence, the perturbation-free LPV system can be represented as 
\begin{eqnarray}
\hspace*{-0.6em} \bar{\boldsymbol{x}}_{k+1}\hspace*{-0.3em}& \hspace*{-0.6em}=\hspace*{-0.6em} &\hspace*{-0.3em}
\mathcal{A}_{d}\bar{\boldsymbol{z}}_{k}\hspace*{-0.2em}+\hspace*{-0.2em}\mathcal{B}_{d}\bar{\boldsymbol{u}}_{k},
\label{c1}
\end{eqnarray}
where $\bar{\boldsymbol{z}}_{k}\hspace*{-0.3em}=\hspace*{-0.3em}\mathrm{Col}(\bar{\boldsymbol{x}}_{k},\hspace*{-0.1em}\boldsymbol{p}_{k}\hspace*{-0.1em}\otimes\hspace*{-0.1em} \bar{\boldsymbol{x}}_{k}\hspace*{-0.1em})\hspace*{-0.1em}$. We describe 
$\overline{\boldsymbol{W}}\hspace*{-0.3em}=\hspace*{-0.3em}[\bar{\boldsymbol{\omega}}_{0},\cdots\hspace*{-0.1em},\bar{\boldsymbol{\omega}}_{\mathcal{T}-1}]$
with $\bar{\boldsymbol{\omega}}_{0}\hspace*{-0.2em}=\hspace*{-0.2em}0$ and 
$\bar{\boldsymbol{\omega}}_{k+1}\hspace*{-0.2em}=\hspace*{-0.2em}
\boldsymbol{x}_{k+1}\hspace*{-0.2em}-\hspace*{-0.2em}\bar{\boldsymbol{x}}_{k+1}$, then 
$\widetilde{\boldsymbol{W}}\hspace*{-0.2em}=\hspace*{-0.2em}\mathrm{Col}(\overline{\boldsymbol{W}}\hspace*{-0.1em},\overline{\boldsymbol{W}}_{\hspace*{-0.2em}P})$ denotes the cumulative effect induced by the external perturbation. 
The data matrices $\hspace*{-0.1em}\overline{\boldsymbol{X}}_{\hspace*{-0.2em}P}\hspace*{-0.1em}$ and $\hspace*{-0.1em}\overline{\boldsymbol{W}}_{\hspace*{-0.2em}P}\hspace*{-0.1em}$ have the same structures as $\boldsymbol{X}_{\hspace*{-0.2em}P}$ and $\hspace*{-0.1em}\boldsymbol{W}_{\hspace*{-0.2em}P}\hspace*{-0.1em}$, 
respectively, with $(\boldsymbol{x}_{k},\boldsymbol{\omega}_{k})$ replaced by $(\bar{\boldsymbol{x}}_{k},\bar{\boldsymbol{\omega}}_{k})$. 
Therefore, the relationship between the sampled actual states and the nominal ones can be expressed as follows  
\begin{eqnarray*}
\hspace*{-0.6em} \boldsymbol{x}_{0}\hspace*{-0.3em}& \hspace*{-0.6em}=\hspace*{-0.6em} &\hspace*{-0.3em}\bar{\boldsymbol{x}}_{0}, \notag \\
\hspace*{-0.6em} \boldsymbol{x}_{1}\hspace*{-0.3em}& \hspace*{-0.6em}=\hspace*{-0.6em} &\hspace*{-0.3em}\bar{\boldsymbol{x}}_{1}\hspace*{-0.2em}+\hspace*{-0.2em}\boldsymbol{\omega}_{0}, \notag \\
\hspace*{-0.6em} \boldsymbol{x}_{2}\hspace*{-0.3em}& \hspace*{-0.6em}=\hspace*{-0.6em} &\hspace*{-0.3em}\bar{\boldsymbol{x}}_{2}
\hspace*{-0.2em}+\hspace*{-0.2em}\overline{\mathcal{A}}_{1}\boldsymbol{\omega}_{0}\hspace*{-0.2em}+\hspace*{-0.2em}\boldsymbol{\omega}_{1}, \notag \\
\hspace*{-0.6em} \boldsymbol{x}_{3}\hspace*{-0.3em}& \hspace*{-0.6em}=\hspace*{-0.6em} &\hspace*{-0.3em}\bar{\boldsymbol{x}}_{3}
\hspace*{-0.2em}+\hspace*{-0.2em}\overline{\mathcal{A}}_{2}\overline{\mathcal{A}}_{1}\boldsymbol{\omega}_{0}\hspace*{-0.2em}+\hspace*{-0.2em}
\overline{\mathcal{A}}_{2}\boldsymbol{\omega}_{1}\hspace*{-0.2em}+\hspace*{-0.2em}\boldsymbol{\omega}_{2}, \notag\\
\hspace*{-0.6em} \cdots \hspace*{-0.3em}& \hspace*{-0.6em}\hspace*{-0.6em} &\hspace*{-0.3em}\cdots \notag \\
\hspace*{-0.6em} \boldsymbol{x}_{\mathcal{T}-1}\hspace*{-0.3em}& \hspace*{-0.6em}=\hspace*{-0.6em} &\hspace*{-0.3em}\bar{\boldsymbol{x}}_{\mathcal{T}-1}
\hspace*{-0.2em}+\hspace*{-0.2em}\sum\nolimits_{\imath=1}^{\mathcal{T}-2}\left(\prod\nolimits_{\jmath=i}^{\mathcal{T}-2}\overline{\mathcal{A}}_{\jmath}\right)
\boldsymbol{\omega}_{\imath-1}
\hspace*{-0.2em}+\hspace*{-0.2em}\boldsymbol{\omega}_{\mathcal{T}-2},
\label{c2}
\end{eqnarray*}
where $\overline{\mathcal{A}}_{\jmath}\hspace*{-0.2em}=\hspace*{-0.2em}
A_{d0}\hspace*{-0.1em}+\hspace*{-0.3em}\sum_{\imath=1}^{\ell}p_{\jmath}^{[\imath]}A_{d\imath}$. Accordingly, we further obtain 
\begin{eqnarray}
\hspace*{-0.6em} \bar{\boldsymbol{\omega}}_{k+1}\hspace*{-0.3em}& \hspace*{-0.6em}=\hspace*{-0.6em} &\hspace*{-0.3em}\sum
\nolimits_{\imath=1}^{k}\left(\prod\nolimits_{\jmath=i}^{k}\overline{\mathcal{A}}_{\jmath}\right)
\boldsymbol{\omega}_{\imath-1}
\hspace*{-0.2em}+\hspace*{-0.2em}\boldsymbol{\omega}_{k},
\label{c3}
\end{eqnarray}
holding for $\hspace*{-0.1em}k\hspace*{-0.3em}\in\hspace*{-0.3em}\mathds{N}_{[0,\mathcal{T}-2]}$. \hspace*{-0.1em}Based on \hspace*{-0.1em}Assumption\hspace*{-0.1em} \ref{ass1}, we aim to bound the norm of $\widetilde{\boldsymbol{W}}$ in terms of the constant 
$\delta\hspace*{-0.2em}\in\hspace*{-0.2em}\mathds{R}_{+}$. Using the norm property of the Kronecker product, we obtain  
\begin{eqnarray}
\hspace*{-0.6em} \|\widetilde{\boldsymbol{W}}\|\hspace*{-0.3em}& \hspace*{-0.6em}\leq\hspace*{-0.6em} &\hspace*{-0.3em} \sum\nolimits_{\imath=0}^{\mathcal{T}-1}\|\bar{\boldsymbol{\omega}}_{\imath}\|\hspace*{-0.2em}+\hspace*{-0.2em}
\sum\nolimits_{\imath=0}^{\mathcal{T}-1}\|\boldsymbol{p}_{\imath}\otimes\bar{\boldsymbol{\omega}}_{\imath}\|\notag \\
\hspace*{-0.6em} \hspace*{-0.3em}& \hspace*{-0.6em}=\hspace*{-0.6em} &\hspace*{-0.3em} 
\sum\nolimits_{\imath=0}^{\mathcal{T}-1}\|\bar{\boldsymbol{\omega}}_{\imath}\|\hspace*{-0.2em}+\hspace*{-0.2em}
\sum\nolimits_{\imath=0}^{\mathcal{T}-1}\|\boldsymbol{p}_{\imath}\|\|\bar{\boldsymbol{\omega}}_{\imath}\| \notag\\
\hspace*{-0.6em} \hspace*{-0.3em}& \hspace*{-0.6em}=\hspace*{-0.6em} &\hspace*{-0.3em} 
\sum\nolimits_{\imath=0}^{\mathcal{T}-1}(1+\|\boldsymbol{p}_{\imath}\|)\|\bar{\boldsymbol{\omega}}_{\imath}\|
\label{c4}
\end{eqnarray}
Then, combining \eqref{c3} with \eqref{c4} yields  
\begin{eqnarray*}
\hspace*{-0.6em} \|\widetilde{\boldsymbol{W}}\|\hspace*{-0.3em}& \hspace*{-0.6em}=\hspace*{-0.6em} &\hspace*{-0.3em}  
\sum\nolimits_{\imath=0}^{\mathcal{T}-1}(1\hspace*{-0.2em}+\hspace*{-0.2em}\|\boldsymbol{p}_{\imath}\|)
\left\|\sum_{\jmath=1}^{\imath-1}\left(\prod_{h=\jmath}^{\imath-1}\overline{\mathcal{A}}_{h}\right)\boldsymbol{\omega}_{\jmath-1}
\hspace*{-0.2em}+\hspace*{-0.2em}\boldsymbol{\omega}_{\imath-1}\right\| \notag \\
\hspace*{-0.6em} \hspace*{-0.3em}& \hspace*{-0.6em}\leq\hspace*{-0.6em} &\hspace*{-0.3em} 
\sum\nolimits_{\imath=0}^{\mathcal{T}-1}(1\hspace*{-0.2em}+\hspace*{-0.2em}\|\boldsymbol{p}_{\imath}\|)
\left[\left(\sum_{\jmath=1}^{\imath-1}\left(\prod_{h=\jmath}^{\imath-1}\|\overline{\mathcal{A}}_{h}\|\right)\hspace*{-0.2em} \right)
\hspace*{-0.2em}+\hspace*{-0.2em}1 \right]\hspace*{-0.2em}\delta, 
\label{c5}
\end{eqnarray*}
where the right-hand side is denoted by $ c_{\omega}\delta$, and $c_\omega$ is only 
used for theoretical analysis. According to \cite[Theorem 3.1]{Clouson2022Robust}, 
the minimum singular value of $\Theta_{\bar{x}}$ satisfies 
$\lambda_{\mathrm{min}}(\Theta_{\bar{x}})\hspace*{-0.2em}\geq\hspace*{-0.2em}\frac{\theta\rho}
{\sqrt{(1\hspace*{-0.1em}+\hspace*{-0.1em}\ell)n\hspace*{-0.1em}+\hspace*{-0.1em}1}}$ with 
$\rho\hspace*{-0.2em}\in\hspace*{-0.2em}\mathbb{R}_{+}$. 
If $\theta\hspace*{-0.2em}>\hspace*{-0.3em}\sqrt{(1\hspace*{-0.2em}+\hspace*{-0.2em}\ell)n\hspace*{-0.2em}+\hspace*{-0.2em}1}
c_{\omega}\delta/\rho$, then $\lambda_{\mathrm{min}}(\Theta_{\bar{x}})\hspace*{-0.2em}>\hspace*{-0.2em}
c_{\omega}\delta\hspace*{-0.2em}\geq\hspace*{-0.2em}\|\widetilde{\boldsymbol{W}}\|\hspace*{-0.3em}=\hspace*{-0.3em}\|\Theta_{\bar{\omega}}\|$. 
Note that $\lambda_{\mathrm{min}}(\boldsymbol{M}\hspace*{-0.3em}+\hspace*{-0.3em}\boldsymbol{N})\hspace*{-0.2em}\geq\hspace*{-0.2em}
\lambda_{\mathrm{min}}(\boldsymbol{M})$ $-\|\boldsymbol{N}\|$ holds for matrices $\boldsymbol{M}$ and $\boldsymbol{N}$ with the same dimension.
Hence, we obtain $\lambda_{\min}(\Theta)\hspace*{-0.2em}\geq\hspace*{-0.2em}
\lambda_{\mathrm{min}}(\Theta_{\bar{x}})\hspace*{-0.2em}-\hspace*{-0.2em}\|\Theta_{\bar{\omega}}\|>0$, which implies that 
the matrix $\Theta$ has full-row rank, and hence \eqref{e9} holds.
\end{IEEEproof}

\begin{remark}
Lemma 1, inspired by the PoE condition for perturbed LTI systems \cite{Clouson2022Robust} and nominal LPV systems \cite{Verhoek2024Direct}, introduces the $\theta$-PoE condition for data-driven control of LPV systems with bounded disturbances. The parameter $\theta$ measures the informativity of the collected data. In practice, the $\theta$-PoE condition ensures that the input and scheduling signals exhibit sufficient variability during data collection. 
Failing to meet this condition can result in collected data that cannot adequately capture the unknown LPV system, leading to difficulties in synthesizing a controller directly from data.
\end{remark}

Identifying $\mathcal{A}_{d}$ and $\mathcal{B}_{d}$ explicitly may incur a high computational 
burden.  We aim to develop a robust controller 
directly from the available data $\mathbb{D}_{s}$. 
Within this context, the LPV system \eqref{e1} 
can be expressed by 
\begin{eqnarray}
\hspace*{-0.6em} \boldsymbol{x}_{k+1}\hspace*{-0.3em}& \hspace*{-0.6em}=\hspace*{-0.6em} &\hspace*{-0.3em}
(\boldsymbol{X}_{\hspace*{-0.1em}+}\hspace*{-0.2em}-\hspace*{-0.2em}\boldsymbol{W})\mathcal{\boldsymbol{G}}^{\dagger}\boldsymbol{f}(\boldsymbol{x}_{k},\boldsymbol{u}_{k},\boldsymbol{p}_{k})\hspace*{-0.2em}+\hspace*{-0.2em}\boldsymbol{\omega}_{k},
\label{d1}
\end{eqnarray} 
where $\boldsymbol{f}(\boldsymbol{x}_k,\boldsymbol{u}_k,\boldsymbol{p}_k)\hspace*{-0.2em}=\hspace*{-0.2em}\mathrm{Col}(\boldsymbol{z}_{k},\bar{\boldsymbol{u}}_{k})$ and $\boldsymbol{G}=\mathrm{Col}(\boldsymbol{X}\hspace*{-0.1em},\boldsymbol{X}_{\hspace*{-0.2em}P},\boldsymbol{U}\hspace*{-0.1em},\boldsymbol{U}_{\hspace*{-0.2em}P})$. 
The system parameters $\mathcal{A}_{d}$ and $\mathcal{B}_{d}$ can be effectively identified by minimizing the loss 
$\|(\boldsymbol{X}_{\hspace*{-0.1em}+}\hspace*{-0.2em}-\hspace*{-0.2em}\boldsymbol{W})
\hspace*{-0.2em}-\hspace*{-0.2em}[\mathcal{A}_{d},\mathcal{B}_{d}]\boldsymbol{G}\|$. 
To stabilize the LPV system \eqref{e6}, one can design the feedback control strategy 
$\boldsymbol{u}_{k}\hspace*{-0.3em}=\hspace*{-0.3em}K_{\hspace*{-0.1em}d}(\boldsymbol{p}_{k})\boldsymbol{x}_{k}$, in which 
$K_{\hspace*{-0.1em}d}\hspace*{-0.2em}:\hspace*{-0.2em}\mathds{P}\hspace*{-0.2em}\mapsto\hspace*{-0.2em}\mathds{R}^{m\times n}$
depends affinely on $\boldsymbol{p}_{k}$ as
\begin{eqnarray}
\label{f1}
\hspace*{-1.0em} K_{d}(\boldsymbol{p}_{k})\hspace*{-0.3em}& \hspace*{-0.6em}=\hspace*{-0.6em} &\hspace*{-0.4em}
K_{d0}\hspace*{-0.2em}+\hspace*{-0.3em}\sum\limits_{\imath=1}^{\ell}\boldsymbol{p}_{k}^{[\imath]}K_{d\imath}. 
\end{eqnarray}

\begin{figure}[htbp]
\centerline{
\includegraphics[width=8.5cm]{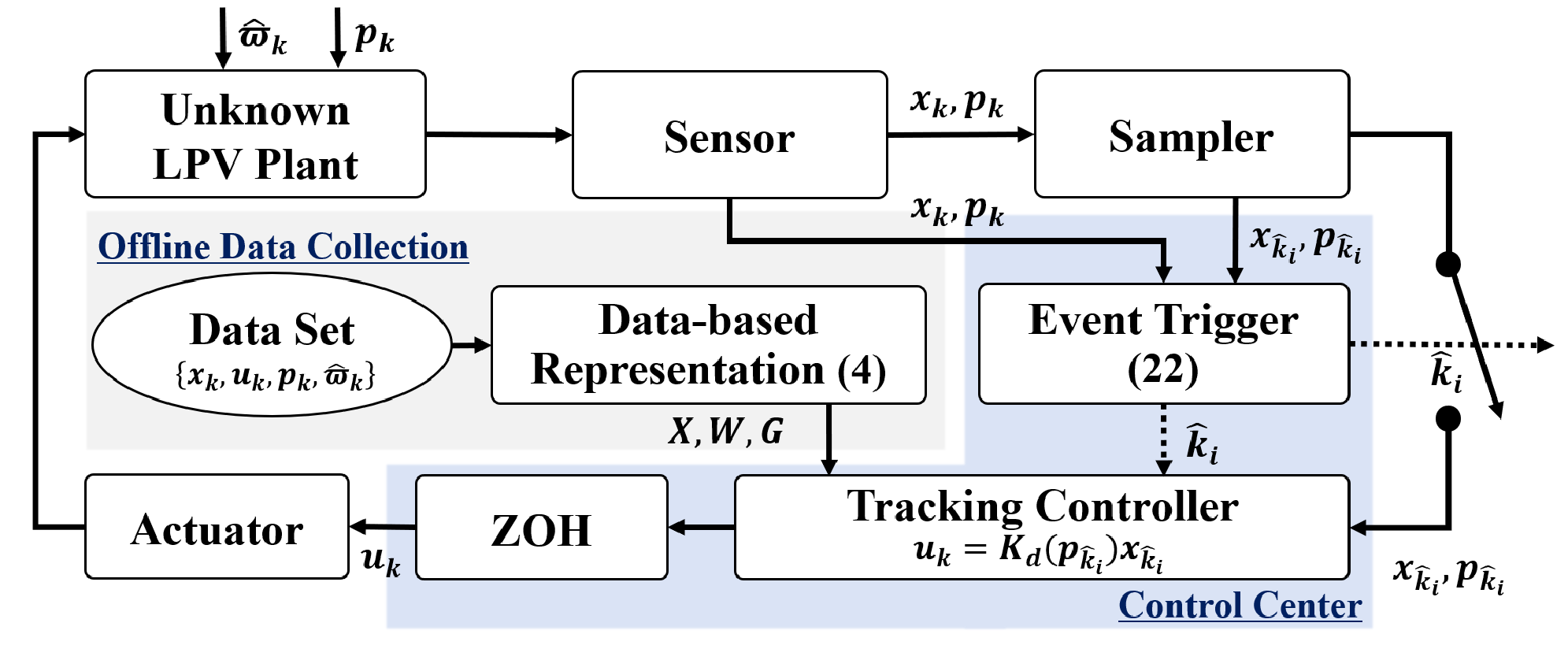} \vspace{-1ex}}
\caption{The design framework of event-triggered LPV control 
from data.}
\label{Fig01}
\end{figure}

As illustrated in Fig. \ref{Fig01}, the framework consists of offline data collection and control center modules. 
The offline data set $\mathbb{D}_{s}$ is first used to construct the data-based LPV representation under the 
$\theta$-PoE condition in Lemma 1. Based on this representation, the parameter-dependent feedback gain $\boldsymbol{K}_{d}(\boldsymbol{p}_{k})$ 
and the triggering parameters are synthesized through the semidefinite {\color{revisionpurple}programs} developed in Section III. 
During online implementation, the event trigger evaluates the triggering logic to determine the transmission instants. 
The latest transmitted state and scheduling signal are retained by the ZOH and used to generate the control input. 

In what follows, we develop data-driven control strategies for perturbed LPV systems with closed-loop stability guarantees, covering both robust stabilization and reference tracking. 
Based on the proposed $\theta$-PoE condition for perturbed LPV systems 
in Lemma \ref{lemm1}, the problems addressed in this paper are summarized as follows. \newline
\noindent \textbf{Problem 1}. 
Given a data-based closed-loop LPV representation with bounded external perturbations, how can data-based stability certificates be derived and formulated as convex programs that are feasible at the vertices of the polytopic scheduling set, thereby enabling the synthesis of a data-driven robust LPV stabilization controller? \newline   
\noindent \textbf{Problem 2}. 
With an event-triggered communication mechanism, how can the triggering parameters be incorporated into the data-based stability certificates so that a data-driven event-triggered LPV controller can be synthesized via a convex program involving only finitely many matrix inequalities? \newline
\noindent \textbf{Problem 3}. 
To extend data-driven event-triggered LPV stabilization to robust reference tracking, how can an integral compensator be constructed to obtain an augmented LPV system and corresponding stability certificates for the synthesis of a data-driven event-triggered LPV tracking controller? 

\section{Learning Event-Triggered LPV Controllers}
\label{section3}

\subsection{Closed-Loop Stability Verification}
\label{sub303}
Recall the specific feedback control strategy $\boldsymbol{u}_{k}\hspace*{-0.3em}=\hspace*{-0.3em}K_{\hspace*{-0.1em}d}(\boldsymbol{p}_{k})\boldsymbol{x}_{k}$ in \eqref{f1}, we further have $\boldsymbol{u}_{k}
\hspace*{-0.3em}=\hspace*{-0.3em}\mathcal{K}\boldsymbol{z}_{k}$ with $\mathcal{K}\hspace*{-0.3em}=\hspace*{-0.3em}[K_{d0},\bar{K}_{d}]$, where $\bar{K}_{d}\hspace*{-0.3em}=\hspace*{-0.3em}[K_{d1},\cdots\hspace*{-0.1em},K_{d\ell}]$. 
Therefore, the closed-loop dynamics of the  
LPV system \eqref{e6} can be written as 
\begin{eqnarray}
\hspace*{-0.6em} \boldsymbol{x}_{k+1}\hspace*{-0.3em}& \hspace*{-0.6em}=\hspace*{-0.6em} &\hspace*{-0.3em}
\mathcal{N}\mathrm{Col}(\boldsymbol{z}_{k},\boldsymbol{p}_{k}\hspace*{-0.2em}\otimes\hspace*{-0.2em}\boldsymbol{p}_{k}
\hspace*{-0.2em}\otimes\hspace*{-0.2em}\boldsymbol{x}_{k})\hspace*{-0.2em}+\hspace*{-0.2em}\boldsymbol{\omega}_{k},
\label{d2}
\end{eqnarray} 
where $\mathcal{N}\hspace*{-0.2em}=\hspace*{-0.2em}[A_{d0}\hspace*{-0.2em}+\hspace*{-0.2em}B_{d0}K_{d0},\bar{A}_{d0}\hspace*{-0.2em}+\hspace*{-0.2em}B_{d0}\bar{K}_{d}\hspace*{-0.2em}+\hspace*{-0.2em}\bar{B}_{d}(I_{\ell}\hspace*{-0.1em}\otimes\hspace*{-0.1em} K_{d0}),\bar{B}_{d}$ $(I_{\ell}\hspace*{-0.1em}\otimes\hspace*{-0.2em}\bar{K}_{d})]$ with $\hspace*{-0.1em}\bar{A}_{d0}$ and $\hspace*{-0.1em}\bar{B}_{d}$ having the similar structures as $\bar{K}_{d}$. According to the fundamental lemma \cite[Lemma 2]{Persis2020Formulas}, 
we have the data-based closed-loop representation 
\begin{eqnarray}
\hspace*{-0.6em} \boldsymbol{x}_{k+1}\hspace*{-0.3em}& \hspace*{-0.6em}=\hspace*{-0.6em} &\hspace*{-0.3em}
(\boldsymbol{X}_{\hspace*{-0.1em}+}\hspace*{-0.2em}-\hspace*{-0.2em}\boldsymbol{W})\mathcal{V}\mathrm{Col}(\boldsymbol{z}_{k},\boldsymbol{p}_{k}\hspace*{-0.2em}\otimes\hspace*{-0.2em}\boldsymbol{p}_{k}
\hspace*{-0.2em}\otimes\hspace*{-0.2em}\boldsymbol{x}_{k})\hspace*{-0.2em}+\hspace*{-0.2em}\boldsymbol{\omega}_{k},
\label{d3}
\end{eqnarray}
where any $\mathcal{V}\hspace*{-0.2em}\in\hspace*{-0.2em}\mathds{R}^{\mathcal{T}\times n(1+\ell+\ell^{2})}$ satisfies 
\begin{eqnarray}
\mathcal{N}_{cl}\hspace*{-0.2em}=\hspace*{-0.2em}\left[\hspace*{-0.2em}
\begin{array}{ccc}
  I_{n}  \hspace*{-0.2em} & \hspace*{-0.2em}\boldsymbol{0} \hspace*{-0.2em}&\hspace*{-0.2em}\boldsymbol{0} \\
    \boldsymbol{0} \hspace*{-0.2em}&\hspace*{-0.2em} I_{\ell}\otimes I_{n}  \hspace*{-0.2em}&\hspace*{-0.2em} \boldsymbol{0} \\
    K_{d0}\hspace*{-0.2em}&\hspace*{-0.2em} \bar{K}_{d}\hspace*{-0.2em}&\hspace*{-0.2em} \boldsymbol{0} \\
    \boldsymbol{0} \hspace*{-0.2em}&\hspace*{-0.2em} I_{\ell}\otimes K_{d0} \hspace*{-0.2em}&\hspace*{-0.2em} I_{\ell}\otimes \bar{K}_{d}
\end{array}
\hspace*{-0.2em}\right]\hspace*{-0.4em}=\hspace*{-0.2em}\boldsymbol{G}\mathcal{V} \label{r1}
\end{eqnarray}
and accordingly $\mathcal{N}\hspace*{-0.2em}=\hspace*{-0.2em}[\mathcal{A}_{d},\mathcal{B}_{d}]\mathcal{N}_{cl}$. This result 
enables the design of a data-driven controller that stabilizes 
all LPV systems represented by \eqref{d3} for 
any $\boldsymbol{\omega}_{k}\hspace*{-0.3em}\in\hspace*{-0.3em}\mathds{B}_{\delta}\hspace*{-0.3em}=\hspace*{-0.3em}\{\boldsymbol{\omega}|\|\boldsymbol{\omega}\|
\hspace*{-0.3em}\leq\hspace*{-0.3em}\delta\}$. 
\begin{theorem}
\label{theo1}
For the data matrices generated from $\mathbb{D}_{s}$ in \eqref{ee4}, suppose that there exist a matrix 
$P\hspace*{-0.2em}\in\hspace*{-0.2em}\mathbb{S}^{n}_{++}$ and scalars 
$\varepsilon_{1}\hspace*{-0.2em}\in\hspace*{-0.2em}\mathds{R}_{>0}$, $\sigma\hspace*{-0.2em}\in\hspace*{-0.2em}\mathds{R}_{>1}$ and 
$\beta_{1}\hspace*{-0.2em}\in\hspace*{-0.2em}\mathds{R}_{(0,1)}$, such that  
\begin{eqnarray}
\left[\hspace*{-0.2em}
\begin{array}{ccccc}
\hspace*{-0.2em}P  \hspace*{-0.2em} & \hspace*{-0.2em}\boldsymbol{0} \hspace*{-0.2em}&\hspace*{-0.2em}F^{\top}(\boldsymbol{p}_{k})\boldsymbol{X}^{\top}_{+} \hspace*{-0.2em}&\hspace*{-0.2em} P \hspace*{-0.2em}&\hspace*{-0.2em}F^{\top}(\boldsymbol{p}_{k})\hspace*{-0.2em}\\
\hspace*{-0.2em}\boldsymbol{0} \hspace*{-0.2em}&\hspace*{-0.2em} \sigma P  \hspace*{-0.2em}&\hspace*{-0.2em} P \hspace*{-0.2em}&\hspace*{-0.2em} \boldsymbol{0} \hspace*{-0.2em}&\hspace*{-0.2em}\boldsymbol{0}\hspace*{-0.2em}\\
\hspace*{-0.2em}\boldsymbol{X}_{+}F(\boldsymbol{p}_{k})\hspace*{-0.2em}&\hspace*{-0.2em} P\hspace*{-0.2em}&\hspace*{-0.2em} P-\varepsilon_{1}\Delta \Delta^{\top} \hspace*{-0.2em}&\hspace*{-0.2em} \boldsymbol{0}\hspace*{-0.2em}&\hspace*{-0.2em}\boldsymbol{0}\hspace*{-0.2em}\\
\hspace*{-0.2em}P \hspace*{-0.2em}&\hspace*{-0.2em} \boldsymbol{0} \hspace*{-0.2em}&\hspace*{-0.2em} \boldsymbol{0} \hspace*{-0.2em}&\hspace*{-0.2em} \beta_{1}^{-1}P\hspace*{-0.2em}&\hspace*{-0.2em}\boldsymbol{0} \hspace*{-0.2em}\\
 \hspace*{-0.2em}F(\boldsymbol{p}_{k}) \hspace*{-0.2em}&\hspace*{-0.2em} \boldsymbol{0} \hspace*{-0.2em}&\hspace*{-0.2em} \boldsymbol{0} \hspace*{-0.2em}&\hspace*{-0.2em} \boldsymbol{0}\hspace*{-0.2em}&\hspace*{-0.2em} \varepsilon_{1}I_{\mathcal{T}}\hspace*{-0.2em}
\end{array}
\hspace*{-0.2em}\right]\hspace*{-0.4em}\succ\hspace*{-0.2em}\boldsymbol{0}, \label{t1}
\end{eqnarray}
where $F(\boldsymbol{p}_{k})\hspace*{-0.2em}=\hspace*{-0.2em}\mathcal{V}\mathrm{Col}(I_{n},
\boldsymbol{p}_{k}\hspace*{-0.2em}\otimes\hspace*{-0.2em} I_{n},
\boldsymbol{p}_{k}\hspace*{-0.2em}\otimes\hspace*{-0.2em} \boldsymbol{p}_{k}\hspace*{-0.2em}\otimes\hspace*{-0.2em} I_{n})P$, 
holds for the known external perturbation bound $\Delta$. 
Then, the closed-loop LPV system \eqref{d2}
is exponentially ISS. 
\end{theorem}
\begin{IEEEproof}
By applying the Schur complement to \eqref{t1}, we obtain 
\begin{eqnarray}
\hspace*{-1.0em}&&\hspace*{-1.0em}\underbrace{\left[\hspace*{-0.2em}
\begin{array}{cccc}
\hspace*{-0.2em}P  \hspace*{-0.2em} & \hspace*{-0.2em}\boldsymbol{0} \hspace*{-0.2em}&\hspace*{-0.2em}F^{\top}(\boldsymbol{p}_{k})\boldsymbol{X}^{\top}_{+} \hspace*{-0.2em}&\hspace*{-0.2em} P \hspace*{-0.2em}\\
\hspace*{-0.2em}\boldsymbol{0}  \hspace*{-0.2em} & \hspace*{-0.2em}\sigma P \hspace*{-0.2em}&\hspace*{-0.2em}P \hspace*{-0.2em}&\hspace*{-0.2em} \boldsymbol{0} \hspace*{-0.2em}\\
\hspace*{-0.2em}\boldsymbol{X}_{+}F(\boldsymbol{p}_{k})  \hspace*{-0.2em} & \hspace*{-0.2em}P \hspace*{-0.2em}&\hspace*{-0.2em}P \hspace*{-0.2em}&\hspace*{-0.2em} \boldsymbol{0} \hspace*{-0.2em}\\
\hspace*{-0.2em}P  \hspace*{-0.2em} & \hspace*{-0.2em}\boldsymbol{0} \hspace*{-0.2em}&\hspace*{-0.2em}\boldsymbol{0} \hspace*{-0.2em}&\hspace*{-0.2em} \beta_{1}^{-1}P\hspace*{-0.2em}
\end{array} 
\hspace*{-0.2em}\right]}_{\amalg_{1}}\hspace*{-0.2em}-\varepsilon_{1}^{-1}\hspace*{-0.3em}
\underbrace{\left[\hspace*{-0.2em}
\begin{array}{c}
\hspace*{-0.2em}F^{\top}\hspace*{-0.1em}(\boldsymbol{p}_{k})\hspace*{-0.2em}  \\
\hspace*{-0.2em}\boldsymbol{0} \hspace*{-0.2em} \\
\hspace*{-0.2em}\boldsymbol{0}\hspace*{-0.2em}  \\
\hspace*{-0.2em}\boldsymbol{0} \hspace*{-0.2em} 
\end{array} 
\hspace*{-0.2em}\right]}_{\bar{F}^{\top}(\boldsymbol{p}_{k})}\hspace*{-0.2em}\times \notag \\
\hspace*{-1.0em}&&\hspace*{-1.0em}
\bar{F}(\boldsymbol{p}_{k})\hspace*{-0.2em}-\hspace*{-0.2em}\varepsilon_{1}\mathds{I}_{1}^{\top}\Delta\Delta^{\top}
\underbrace{[\boldsymbol{0},\boldsymbol{0},I_{n},\boldsymbol{0}]}_{\mathds{I}_{1}}\hspace*{-0.2em}\succ\hspace*{-0.2em}\boldsymbol{0}. 
\end{eqnarray}
By Petersen's lemma \cite{Bisoffi2022Data}, we further obtain  
\begin{eqnarray*}
\label{mm1}
\hspace*{-1.0em} \amalg_{1}\hspace*{-0.2em}-\hspace*{-0.2em}\bar{F}^{\hspace*{-0.1em}\top}\hspace*{-0.2em}
(\boldsymbol{p}_{k})\boldsymbol{W}^{\hspace*{-0.1em}\top}\hspace*{-0.2em}\mathds{I}_{1}\hspace*{-0.2em}-\hspace*{-0.2em}\mathds{I}_{1}^{\top}
\hspace*{-0.2em}\boldsymbol{W}\bar{F}(\boldsymbol{p}_{k})\hspace*{-0.2em}\succ\hspace*{-0.2em}\boldsymbol{0},
\end{eqnarray*}
which is equivalent to 
\begin{eqnarray}
\hspace*{-1.4em}\left[\hspace*{-0.2em}
\begin{array}{cccc}
\hspace*{-0.3em}P  \hspace*{-0.2em} & \hspace*{-0.2em}\boldsymbol{0} \hspace*{-0.3em}&\hspace*{-0.3em}F^{\hspace*{-0.1em}\top}\hspace*{-0.2em}(\boldsymbol{p}_{k})(\boldsymbol{X}^{\hspace*{-0.1em}\top}_{+}\hspace*{-0.2em}-\hspace*{-0.2em}\boldsymbol{W}^{\top}\hspace*{-0.2em}) \hspace*{-0.3em}&\hspace*{-0.3em} P \hspace*{-0.3em}\\
\hspace*{-0.3em}\boldsymbol{0}  \hspace*{-0.2em} & \hspace*{-0.2em}\sigma P \hspace*{-0.3em}&\hspace*{-0.3em}P \hspace*{-0.3em}&\hspace*{-0.3em} \boldsymbol{0} \hspace*{-0.3em}\\
\hspace*{-0.3em}(\boldsymbol{X}_{\hspace*{-0.1em}+}\hspace*{-0.2em}-\hspace*{-0.2em}\boldsymbol{W})F(\boldsymbol{p}_{k})  \hspace*{-0.2em} & \hspace*{-0.2em}P \hspace*{-0.3em}&\hspace*{-0.3em}P \hspace*{-0.3em}&\hspace*{-0.3em} \boldsymbol{0} \hspace*{-0.3em}\\
\hspace*{-0.3em}P  \hspace*{-0.2em} & \hspace*{-0.2em}\boldsymbol{0} \hspace*{-0.3em}&\hspace*{-0.3em}\boldsymbol{0} \hspace*{-0.3em}&\hspace*{-0.3em} \beta_{1}^{-1}P\hspace*{-0.3em}
\end{array} 
\hspace*{-0.2em}\right]\hspace*{-0.3em}&\hspace*{-0.8em}\succ\hspace*{-0.8em}&\boldsymbol{0}. \label{m2}
\end{eqnarray}
We then apply the Schur complement to \eqref{m2} to obtain  
\begin{eqnarray}
\hspace*{-1.0em}&&\hspace*{-1.0em}
\left[\hspace*{-0.2em}
\begin{array}{cc}
P  \hspace*{-0.2em} & \hspace*{-0.2em}\boldsymbol{0} \\
\boldsymbol{0}\hspace*{-0.2em} & \hspace*{-0.2em}\sigma P
\end{array} 
\hspace*{-0.2em}\right]\hspace*{-0.2em}-\hspace*{-0.2em} 
\left[\hspace*{-0.2em}
\begin{array}{cc}
F^{\top}(\boldsymbol{p}_{k})(\boldsymbol{X}_{+}\hspace*{-0.2em}-\hspace*{-0.2em}\boldsymbol{W})^{\top}  \hspace*{-0.2em} & \hspace*{-0.2em}P \\
P\hspace*{-0.2em} & \hspace*{-0.2em}\boldsymbol{0}
\end{array} 
\hspace*{-0.2em}\right]\hspace*{-0.3em}\left[\hspace*{-0.2em}
\begin{array}{cc}
P^{-1}  \hspace*{-0.2em} & \hspace*{-0.2em}\boldsymbol{0} \\
\boldsymbol{0}\hspace*{-0.2em} & \hspace*{-0.2em}\beta_{1} P^{-1}
\end{array} 
\hspace*{-0.2em}\right]\hspace*{-0.2em}\times \notag \\
\hspace*{-1.0em}&&\hspace*{-1.0em}
\left[\hspace*{-0.2em}
\begin{array}{cc}
(\boldsymbol{X}_{+}\hspace*{-0.2em}-\hspace*{-0.2em}\boldsymbol{W})F(\boldsymbol{p}_{k})  \hspace*{-0.2em} & \hspace*{-0.2em}P \\
P\hspace*{-0.2em} & \hspace*{-0.2em}\boldsymbol{0}
\end{array} 
\hspace*{-0.2em}\right]\hspace*{-0.2em}\succ\hspace*{-0.2em}\boldsymbol{0}. \label{m3}
\end{eqnarray}
To establish the exponential ISS of the LPV system \eqref{d2}, 
we choose $V(\boldsymbol{x}_{k})\hspace*{-0.2em}=\hspace*{-0.2em}\boldsymbol{x}_{k}^{\top}\hspace*{-0.1em}P^{-1}\hspace*{-0.1em}\boldsymbol{x}_{k}$ 
as a Lyapunov function, whose forward difference satisfies 
\begin{eqnarray}
\hspace*{-1.0em} V(\boldsymbol{x}_{k+1})\hspace*{-0.2em}-\hspace*{-0.2em}V(\boldsymbol{x}_{k})
&\hspace*{-0.8em}\leq\hspace*{-0.8em}&-\hspace*{-0.2em}\beta_{1}V(\boldsymbol{x}_{k})
\hspace*{-0.2em}+\hspace*{-0.2em}\sigma\boldsymbol{\omega}_{k}^{\top}P^{-1}\boldsymbol{\omega_{k}}. \label{m4}
\end{eqnarray}
By pre- and post-multiplying both sides of \eqref{m3} by $\mathrm{Blkdiag}(P^{-1}\hspace*{-0.1em},\hspace*{-0.1em}P^{-1}\hspace*{-0.1em})$, 
the equivalence between  
\eqref{m3} and \eqref{m4} is established.  In addition, \eqref{m4} further yields 
\begin{eqnarray}
\hspace*{-1.0em} V(\boldsymbol{x}_{k+1})\hspace*{-0.2em}-\hspace*{-0.2em}V(\boldsymbol{x}_{k})
&\hspace*{-0.8em}\leq\hspace*{-0.8em}&-\hspace*{-0.2em}\beta_{1}V(\boldsymbol{x}_{k})
\hspace*{-0.2em}+\hspace*{-0.2em}\sigma\mu_{\max}(P^{-1})\|\omega_{k}\|^{2} \notag \\
&\hspace*{-0.8em}=\hspace*{-0.8em}&-\hspace*{-0.2em}\beta_{1}V(\boldsymbol{x}_{k})+\hbar(\|\omega_{k}\|). 
\label{m5}
\end{eqnarray}
Here, $\hbar(\|\omega_{k}\|)$ is a $\mathcal{K}$-function of $\|\omega_{k}\|$, 
and $\mu_{\max}(P^{-1})$ denotes the maximum eigenvalue of $P^{-1}$. By recursively applying 
\eqref{m5}, we obtain
\begin{eqnarray*}
\hspace*{-1.0em} V(\boldsymbol{x}_{k})
&\hspace*{-0.8em}\leq\hspace*{-0.8em}&(1\hspace*{-0.2em}-\hspace*{-0.2em}\beta_{1})^{k}V(\boldsymbol{x}_{0})
\hspace*{-0.2em}+\hspace*{-0.2em}\sum\nolimits_{i=0}^{k-1}\hspace*{-0.2em}(1\hspace*{-0.2em}-\hspace*{-0.2em}\beta_{1})^{k-1-i}\hbar(\|\omega_{i}\|),
\label{m6}
\end{eqnarray*}
which leads to 
\begin{eqnarray}
\hspace*{-1.0em} \|\boldsymbol{x}_{k}\|
&\hspace*{-0.8em}\leq\hspace*{-0.8em}&e^{-\beta_{1}k/2}\mathcal{R}\|\boldsymbol{x}_{0}\|+\gamma_{1}(\|\boldsymbol{\omega_{s}}\|_{[0,k-1]}),
\label{m7}
\end{eqnarray}
with the overshoot factor $\mathcal{R}\hspace*{-0.2em}=\hspace*{-0.2em}\mu_{\max}^{1/2}(\hspace*{-0.1em}P^{-1}\hspace*{-0.1em})/\mu_{\min}^{1/2}(\hspace*{-0.1em}P^{-1}\hspace*{-0.1em})$ and 
\begin{eqnarray*}
\hspace*{-1.0em} \gamma_{1}(\|\boldsymbol{\omega_{s}}\|_{[0,k-1]})
&\hspace*{-0.9em}=\hspace*{-0.9em}&c_{1}\hspace*{-0.2em}\sum\nolimits_{i=1}^{k-1}\hspace*{-0.2em}e^{-\beta_{1}(k-1-i)/2}\hbar(\|\boldsymbol{\omega_{s}}\|_{[0,k-1]}),
\label{mm7}
\end{eqnarray*}
holding with $c_{1}\hspace*{-0.2em}=\hspace*{-0.2em}\mu_{\min}^{-1/2}(\hspace*{-0.1em}P^{-1}\hspace*{-0.1em})$. Therefore, 
the closed-loop LPV system \eqref{d2} is exponential ISS, which completes the proof.
\end{IEEEproof}

\begin{remark}
The scalars $\varepsilon_1 \hspace{-0.2em}>\hspace{-0.2em} 0$, $\sigma \hspace{-0.2em}>\hspace{-0.2em} 1$, and $\beta_1 \hspace{-0.2em}\in\hspace{-0.2em} (0,1)$ in \eqref{t1} are design parameters arising from the robust stability analysis.
Specifically, $\varepsilon_1$ is introduced via Petersen’s lemma to handle the bounded
data uncertainty, while $\sigma$ and $\beta_1$ determine the disturbance
weighting and the exponential decay rate in the ISS Lyapunov inequality, respectively. 
\end{remark}

\subsection{Data-Based Control Synthesis}
\label{3b}
The quadratic dependence of  $F(\boldsymbol{p}_{k})$ on $\boldsymbol{p}_{k}$ makes it difficult to reduce 
\eqref{t1} into a finite number of convex constraints. 
We consider a polytopic 
scheduling set $\mathds{P}$ and employ the full-block $\mathrm{S}$-procedure \cite{Scherer2001LPV} to construct tractable 
convex relaxations based on its vertices. We introduce matrices 
$\mathcal{F}\hspace*{-0.2em}\in\hspace*{-0.2em}\mathds{R}^{\mathcal{T}\times n(1+\ell+\ell^{2})}$ and 
$F_{Q}\hspace*{-0.2em}\in\hspace*{-0.2em}\mathds{R}^{\mathcal{T}(1+\ell)\times n(1+\ell)}$, such that 
\begin{eqnarray} \label{m8} \hspace*{-0.6em}F(\boldsymbol{p}_{k})\hspace*{-0.2em}=\hspace*{-0.2em}\mathcal{F}\hspace*{-0.2em} \left[\hspace*{-0.2em} \begin{array}{c} \hspace*{-0.2em}I_{n}\hspace*{-0.2em} \\ \hspace*{-0.2em}\boldsymbol{p}_{k}\hspace*{-0.2em}\otimes\hspace*{-0.2em} I_{n} \hspace*{-0.2em} \\ \hspace*{-0.2em}\boldsymbol{p}_{k}\hspace*{-0.2em}\otimes\hspace*{-0.2em} \boldsymbol{p}_{k} \hspace*{-0.2em}\otimes\hspace*{-0.2em} I_{n}\hspace*{-0.2em} \end{array} \hspace*{-0.2em}\right]\hspace*{-0.3em}=\hspace*{-0.3em} \left[\hspace*{-0.2em} \begin{array}{c} \hspace*{-0.2em}I_{\mathcal{T}}\hspace*{-0.2em} \\ \hspace*{-0.2em}\boldsymbol{p}_{k}\hspace*{-0.2em}\otimes\hspace*{-0.2em} I_{\mathcal{T}} \hspace*{-0.2em} \end{array} \hspace*{-0.2em}\right]^{\hspace*{-0.1em}\top}\hspace*{-0.3em}F_{Q}\hspace*{-0.2em} \left[\hspace*{-0.2em} \begin{array}{c} \hspace*{-0.2em}I_{n}\hspace*{-0.2em} \\ \hspace*{-0.2em}\boldsymbol{p}_{k}\hspace*{-0.2em}\otimes\hspace*{-0.2em} I_{n} \hspace*{-0.2em} \end{array} \hspace*{-0.2em}\right].\end{eqnarray}
This reformulation eliminates the explicit quadratic dependence on $\boldsymbol{p}_k$ and facilitates the derivation of vertex-wise LMI conditions, leading to the stability certificates established in the following theorem.
\begin{theorem}
\label{theo2}
For the data matrices generated from $\mathbb{D}_{s}$ in \eqref{ee4}, we suppose that there exist the matrices 
$P\hspace*{-0.2em}\in\hspace*{-0.2em}\mathbb{S}^{n}_{++}$, 
$Z_{0}\hspace*{-0.2em}\in\hspace*{-0.2em}\mathds{R}^{m\times n}$, 
$\bar{Z}\hspace*{-0.2em}\in\hspace*{-0.2em}\mathds{R}^{m\times \ell n}$, and 
$\Xi\hspace*{-0.2em}\in\hspace*{-0.2em}\mathbb{S}^{2\ell(4n+\mathcal{T})}$, such that
\begin{eqnarray}
\label{m9}
\hspace*{-0.6em}\mathcal{N}_{cl}\left[\hspace*{-0.2em}
\begin{array}{ccc}
\hspace*{-0.2em}P\hspace*{-0.2em}  &
\hspace*{-0.2em}\boldsymbol{0}\hspace*{-0.2em} &
\hspace*{-0.2em}\boldsymbol{0}\hspace*{-0.2em} \\
\hspace*{-0.2em}\boldsymbol{0}\hspace*{-0.2em}  &
\hspace*{-0.2em}I_{\ell}\hspace*{-0.2em}\otimes\hspace*{-0.2em} P\hspace*{-0.2em} &
\hspace*{-0.2em}\boldsymbol{0}\hspace*{-0.2em} \\
\hspace*{-0.2em}\boldsymbol{0}\hspace*{-0.2em}  &
\hspace*{-0.2em}\boldsymbol{0}\hspace*{-0.2em} &
\hspace*{-0.2em}I_{\ell}\hspace*{-0.2em}\otimes\hspace*{-0.2em} P\hspace*{-0.2em}\otimes\hspace*{-0.2em} P\hspace*{-0.2em}
\end{array} 
\hspace*{-0.2em}\right]\hspace*{-0.2em}=\hspace*{-0.2em}\boldsymbol{G}\mathcal{F},
\end{eqnarray}
and the following conditions 
\begin{subequations}
\label{m10}
\begin{align}
\hspace*{-0.6em}\left[\hspace*{-0.2em}
\begin{array}{cc}
\hspace*{-0.2em}M_{11}\hspace*{-0.2em}  &
\hspace*{-0.2em}M_{12}\hspace*{-0.2em} \\
\hspace*{-0.2em}I\hspace*{-0.2em}  &
\hspace*{-0.2em}\boldsymbol{0}\hspace*{-0.2em} \\
\hspace*{-0.2em}M_{21}\hspace*{-0.2em}  &
\hspace*{-0.2em}M_{22}\hspace*{-0.2em} 
\end{array} 
\hspace*{-0.2em}\right]^{\hspace*{-0.2em}\top}\hspace*{-0.3em}\left[\hspace*{-0.2em}
\begin{array}{cc}
\hspace*{-0.2em}\Xi\hspace*{-0.2em}  &
\hspace*{-0.2em}\boldsymbol{0}\hspace*{-0.2em} \\
\hspace*{-0.2em}\boldsymbol{0}\hspace*{-0.2em}  &
\hspace*{-0.2em}-\Phi\hspace*{-0.2em} 
\end{array} 
\hspace*{-0.2em}\right]\hspace*{-0.2em}
\left[\hspace*{-0.2em}
\begin{array}{cc}
\hspace*{-0.2em}M_{11}\hspace*{-0.2em}  &
\hspace*{-0.2em}M_{12}\hspace*{-0.2em} \\
\hspace*{-0.2em}I\hspace*{-0.2em}  &
\hspace*{-0.2em}\boldsymbol{0}\hspace*{-0.2em} \\
\hspace*{-0.2em}M_{21}\hspace*{-0.2em}  &
\hspace*{-0.2em}M_{22}\hspace*{-0.2em} 
\end{array} 
\hspace*{-0.2em}\right]
& \hspace*{-0.3em}\prec\hspace*{-0.2em} 
\boldsymbol{0}, \label{m10:1} \\
\hspace*{-0.6em}\left[\hspace*{-0.2em}
\begin{array}{c}
\hspace*{-0.2em}I\hspace*{-0.2em}  \\
\hspace*{-0.2em}\Upsilon(\boldsymbol{p}_{k})\hspace*{-0.2em}  
\end{array} 
\hspace*{-0.2em}\right]^{\hspace*{-0.2em}\top} \hspace*{-0.2em}
\left[\hspace*{-0.2em}
\begin{array}{cc}
\hspace*{-0.2em}\Xi_{11}\hspace*{-0.2em}  &
\hspace*{-0.2em}\Xi_{12}\hspace*{-0.2em}   \\
\hspace*{-0.2em}\Xi_{12}^{\top}\hspace*{-0.2em}  &
\hspace*{-0.2em}\Xi_{22}\hspace*{-0.2em} 
\end{array} 
\hspace*{-0.2em}\right]\hspace*{-0.2em}
\left[\hspace*{-0.2em}
\begin{array}{c}
\hspace*{-0.2em}I\hspace*{-0.2em}  \\
\hspace*{-0.2em}\Upsilon(\boldsymbol{p}_{k})\hspace*{-0.2em}  
\end{array} 
\hspace*{-0.2em}\right]
& \hspace*{-0.3em}\succeq\hspace*{-0.2em} 
\boldsymbol{0}, \label{m10:2} \\
\hspace*{-0.6em}\Xi_{22}
& \hspace*{-0.3em}\prec\hspace*{-0.2em} 
\boldsymbol{0}, \label{m10:3} \\
\hspace*{-0.6em}
\left[\hspace*{-0.2em}
\begin{array}{ccc}
\hspace*{-0.2em}P\hspace*{-0.2em}  &
\hspace*{-0.2em}\boldsymbol{0}\hspace*{-0.2em} &
\hspace*{-0.2em}\boldsymbol{0}\hspace*{-0.2em} \\
\hspace*{-0.2em}\boldsymbol{0}\hspace*{-0.2em}  &
\hspace*{-0.2em}I_{\ell}\hspace*{-0.2em}\otimes\hspace*{-0.2em} P\hspace*{-0.2em} &
\hspace*{-0.2em}\boldsymbol{0}\hspace*{-0.2em} \\
\hspace*{-0.2em}Z_{0}\hspace*{-0.2em}  &
\hspace*{-0.2em}\bar{Z}\hspace*{-0.2em} &
\hspace*{-0.2em}\boldsymbol{0}\hspace*{-0.2em}\\
\hspace*{-0.2em}\boldsymbol{0}\hspace*{-0.2em}  &
\hspace*{-0.2em}I_{\ell}\hspace*{-0.2em}\otimes\hspace*{-0.2em} Z_{0}\hspace*{-0.2em} &
\hspace*{-0.2em}I_{\ell}\hspace*{-0.2em}\otimes\hspace*{-0.2em} \bar{Z}\hspace*{-0.2em}
\end{array} 
\hspace*{-0.2em}\right]
& \hspace*{-0.3em}=\hspace*{-0.2em} 
\boldsymbol{G}\mathcal{F}, \label{m10:4}
\end{align}
\end{subequations}
hold for all $\boldsymbol{p}_{k}\hspace*{-0.2em}\in\hspace*{-0.2em}\mathbb{P}$, where 
\begin{subequations}
\label{m11}
\begin{align}
\hspace*{-0.6em}\Upsilon(\boldsymbol{p}_{k})\hspace*{-0.2em}=
&  
\mathrm{Blkdiag}(\mathrm{Diag}(\boldsymbol{p}_{k}\hspace*{-0.2em} \otimes\hspace*{-0.2em} I_{n}),
\mathrm{Diag}(\boldsymbol{p}_{k}\hspace*{-0.2em} \otimes\hspace*{-0.2em} I_{n}), \notag\\
\hspace*{-0.6em}
&  
\mathrm{Diag}(\boldsymbol{p}_{k}\hspace*{-0.2em} \otimes\hspace*{-0.2em} I_{n}),
\mathrm{Diag}(\boldsymbol{p}_{k}\hspace*{-0.2em} \otimes\hspace*{-0.2em} I_{n}),
\mathrm{Diag}(\boldsymbol{p}_{k}\hspace*{-0.2em} \otimes\hspace*{-0.2em} I_{\mathcal{T}})), \\
\hspace*{-0.6em}M_{11}\hspace*{-0.2em}=&  
\boldsymbol{0}_{\ell(4n+\mathcal{T})\times \ell(4n+\mathcal{T})}, \label{m11:2}\\
\hspace*{-0.6em}M_{12}\hspace*{-0.2em}=&\mathrm{Blkdiag}(\boldsymbol{1}_{\ell}\hspace*{-0.3em}\otimes\hspace*{-0.2em}I_{n},
\hspace*{-0.1em}\boldsymbol{1}_{\ell}\hspace*{-0.3em}\otimes\hspace*{-0.2em}I_{n},
\hspace*{-0.1em}\boldsymbol{1}_{\ell}\hspace*{-0.3em}\otimes\hspace*{-0.2em}I_{n},
\hspace*{-0.1em}\boldsymbol{1}_{\ell}\hspace*{-0.3em}\otimes\hspace*{-0.2em}I_{n},
\hspace*{-0.1em}\boldsymbol{1}_{\ell}\hspace*{-0.3em}\otimes\hspace*{-0.2em}I_{\mathcal{T}}\hspace*{-0.1em}), \label{m11:3}\\
\hspace*{-0.6em}M_{21}\hspace*{-0.2em}=&
\mathrm{Blkdiag}\left( \left[\hspace*{-0.2em}
\begin{array}{c}
\hspace*{-0.2em}\boldsymbol{0}_{n\times\ell n}\hspace*{-0.2em}  \\
\hspace*{-0.2em}I_{\ell n}\hspace*{-0.2em}  
\end{array} 
\hspace*{-0.2em}\right],\left[\hspace*{-0.2em}
\begin{array}{c}
\hspace*{-0.2em}\boldsymbol{0}_{n\times\ell n}\hspace*{-0.2em}  \\
\hspace*{-0.2em}I_{\ell n}\hspace*{-0.2em}  
\end{array} 
\hspace*{-0.2em}\right],\left[\hspace*{-0.2em}
\begin{array}{c}
\hspace*{-0.2em}\boldsymbol{0}_{n\times\ell n}\hspace*{-0.2em}  \\
\hspace*{-0.2em}I_{\ell n}\hspace*{-0.2em}  
\end{array} 
\hspace*{-0.2em}\right],\right. \notag\\
\hspace*{-0.6em}\hspace*{-0.2em}&
\left.\left[\hspace*{-0.2em}
\begin{array}{c}
\hspace*{-0.2em}\boldsymbol{0}_{n\times\ell n}\hspace*{-0.2em}  \\
\hspace*{-0.2em}I_{\ell n}\hspace*{-0.2em}  
\end{array} 
\hspace*{-0.2em}\right],\left[\hspace*{-0.2em}
\begin{array}{c}
\hspace*{-0.2em}\boldsymbol{0}_{\mathcal{T}\times\ell \mathcal{T}}\hspace*{-0.2em}  \\
\hspace*{-0.2em}I_{\ell \mathcal{T}}\hspace*{-0.2em}  
\end{array} 
\hspace*{-0.2em}\right] \right), \label{m11:4}\\
\hspace*{-0.6em}M_{22}\hspace*{-0.2em}=&
\mathrm{Blkdiag}\left( \left[\hspace*{-0.2em}
\begin{array}{c}
\hspace*{-0.2em}I_{n}\hspace*{-0.2em}  \\
\hspace*{-0.2em}\boldsymbol{0}_{\ell n\times n}\hspace*{-0.2em}  
\end{array} 
\hspace*{-0.2em}\right],\left[\hspace*{-0.2em}
\begin{array}{c}
\hspace*{-0.2em}I_{n}\hspace*{-0.2em}  \\
\hspace*{-0.2em}\boldsymbol{0}_{\ell n\times n}\hspace*{-0.2em}  
\end{array} 
\hspace*{-0.2em}\right],\left[\hspace*{-0.2em}
\begin{array}{c}
\hspace*{-0.2em}I_{n}\hspace*{-0.2em}  \\
\hspace*{-0.2em}\boldsymbol{0}_{\ell n\times n}\hspace*{-0.2em}  
\end{array} 
\hspace*{-0.2em}\right],\right. \notag\\
\hspace*{-0.6em}\hspace*{-0.2em}&
\left.\left[\hspace*{-0.2em}
\begin{array}{c}
\hspace*{-0.2em}I_{n}\hspace*{-0.2em}  \\
\hspace*{-0.2em}\boldsymbol{0}_{\ell n\times n}\hspace*{-0.2em}  
\end{array} 
\hspace*{-0.2em}\right],\left[\hspace*{-0.2em}
\begin{array}{c}
\hspace*{-0.2em}I_{\mathcal{T}}\hspace*{-0.2em}  \\
\hspace*{-0.2em}\boldsymbol{0}_{\ell \mathcal{T}\times \mathcal{T}}\hspace*{-0.2em}  
\end{array} 
\hspace*{-0.2em}\right] \right), \label{m11:5} \\
\hspace*{-0.6em}\Phi\hspace*{-0.2em}=&  
\left[\hspace*{-0.2em}
\begin{array}{ccccc}
\hspace*{-0.2em}\mathcal{Y}  \hspace*{-0.2em} & \hspace*{-0.2em}\boldsymbol{0} \hspace*{-0.2em}&\hspace*{-0.2em}(\mathcal{X}_{+}F_{Q})^{\top} \hspace*{-0.2em}&\hspace*{-0.2em} \mathcal{Y} \hspace*{-0.2em}&\hspace*{-0.2em}F^{\top}_{Q}\hspace*{-0.2em} \\
\hspace*{-0.2em}\boldsymbol{0}  \hspace*{-0.2em} & \hspace*{-0.2em}\sigma\mathcal{Y} \hspace*{-0.2em}&\hspace*{-0.2em}\mathcal{Y} \hspace*{-0.2em}&\hspace*{-0.2em} \boldsymbol{0} \hspace*{-0.2em}&\hspace*{-0.2em}\boldsymbol{0}\hspace*{-0.2em}\\
\hspace*{-0.2em}\mathcal{X}_{+}F_{Q}  \hspace*{-0.2em} & \hspace*{-0.2em}\mathcal{Y} \hspace*{-0.2em}&\hspace*{-0.2em}\bar{\mathcal{Y}} \hspace*{-0.2em}&\hspace*{-0.2em} \boldsymbol{0} \hspace*{-0.2em}&\hspace*{-0.2em}\boldsymbol{0}\hspace*{-0.2em}\\
\hspace*{-0.2em}\mathcal{Y}  \hspace*{-0.2em} & \hspace*{-0.2em}\boldsymbol{0} \hspace*{-0.2em}&\hspace*{-0.2em}\boldsymbol{0} \hspace*{-0.2em}&\hspace*{-0.2em} \beta_{1}^{-1}\mathcal{Y} \hspace*{-0.2em}&\hspace*{-0.2em}\boldsymbol{0}\hspace*{-0.2em}\\
\hspace*{-0.2em}F_{Q}  \hspace*{-0.2em} & \hspace*{-0.2em}\boldsymbol{0} \hspace*{-0.2em}&\hspace*{-0.2em}\boldsymbol{0} \hspace*{-0.2em}&\hspace*{-0.2em} \boldsymbol{0} \hspace*{-0.2em}&\hspace*{-0.2em}\mathcal{E}\hspace*{-0.2em}
\end{array}
\hspace*{-0.2em}\right], \label{m11:f}
\end{align}
\end{subequations}
with $\boldsymbol{1}_{n}\hspace*{-0.2em}=\hspace*{-0.2em}\mathrm{Col}(1,\cdots,1)\hspace*{-0.2em}\in\hspace*{-0.2em}\mathds{R}^{n}$ and 
\begin{subequations} 
\begin{align}
\hspace*{-0.6em}\mathcal{Y}\hspace*{-0.2em}=&  
\mathrm{Blkdiag}(P,\boldsymbol{0}_{\ell n \times \ell n}),  \notag \\ 
\hspace*{-0.6em}\bar{\mathcal{Y}}\hspace*{-0.2em}=&  
\mathrm{Blkdiag}(P\hspace*{-0.2em}-\hspace*{-0.2em}\varepsilon_{1}\Delta\Delta^{\top},\boldsymbol{0}_{\ell n \times \ell n}), \notag \\
\hspace*{-0.6em}\mathcal{E}\hspace*{-0.2em}=&  
\mathrm{Blkdiag}(\varepsilon_{1}I_{\mathcal{T}},\boldsymbol{0}_{\ell\mathcal{T}\times \ell\mathcal{T}}), \notag \\
\hspace*{-0.6em}\mathcal{X}_{+}\hspace*{-0.2em}=&  
\mathrm{Blkdiag}(\boldsymbol{X}_{+},I_{\ell}\otimes\boldsymbol{X}_{+}). \notag
\end{align}
\end{subequations}
Then, the state-feedback controller 
$\boldsymbol{u}_{k}\hspace*{-0.3em}=\hspace*{-0.3em}K_{\hspace*{-0.1em}d}(\boldsymbol{p}_{k})\boldsymbol{x}_{k}$ in \eqref{f1} 
can be constructed as 
\begin{eqnarray}
\label{e10}
    K_{d0}\hspace*{-0.2em}=\hspace*{-0.2em}Z_{0}P^{-1}, \bar{K}_{d}\hspace*{-0.2em}=\hspace*{-0.2em}\bar{Z}(I_{\ell}\otimes P)^{-1}, 
\end{eqnarray}
which ensures the exponential ISS of the closed-loop system \eqref{d2}.
\end{theorem}
\begin{IEEEproof}
We recall the definition of $F_{Q}$ in \eqref{m8} and apply the Kronecker-product property in \cite[Eq.(1)]{Verhoek2024Direct} twice to $\boldsymbol{X}_{+}F(\boldsymbol{p}_{k})$, yielding 
\begin{eqnarray}
\label{m12}
\hspace*{-0.6em}\boldsymbol{X}_{+}F(\boldsymbol{p}_{k})\hspace*{-0.2em}=\hspace*{-0.2em}\hspace*{-0.2em} \left[\hspace*{-0.2em}
\begin{array}{c}
\hspace*{-0.2em}I_{n}\hspace*{-0.2em}  \\
\hspace*{-0.2em}\boldsymbol{p}_{k}\hspace*{-0.2em}\otimes\hspace*{-0.2em} I_{n} \hspace*{-0.2em} 
\end{array} 
\hspace*{-0.2em}\right]^{\hspace*{-0.1em}\top}\hspace*{-0.2em}\mathcal{X}_{+}F_{Q}\left[\hspace*{-0.2em}
\begin{array}{c}
\hspace*{-0.2em}I_{n}\hspace*{-0.2em}  \\
\hspace*{-0.2em}\boldsymbol{p}_{k}\hspace*{-0.2em}\otimes\hspace*{-0.2em} I_{n} \hspace*{-0.2em} 
\end{array} 
\hspace*{-0.2em}\right]. 
\end{eqnarray}
Combining \eqref{t1} with \eqref{m12} yields 
\begin{eqnarray}
\label{m13}
\hspace*{-0.6em}M^{\top}(\boldsymbol{p}_{k})\Phi M(\boldsymbol{p}_{k})\hspace*{-0.2em}\succ\hspace*{-0.2em} \boldsymbol{0},
\end{eqnarray}
where $\Phi$ is defined in \eqref{m11:f} and $M(\boldsymbol{p}_{k})$ is given by 
\begin{eqnarray}
\label{m14}
\hspace*{-0.6em}M(\boldsymbol{p}_{k})\hspace*{-0.2em}=\hspace*{-0.2em} M_{22}\hspace*{-0.2em}+\hspace*{-0.2em}
M_{21}\hspace*{-0.2em}\Upsilon(\boldsymbol{p}_{k})(I\hspace*{-0.2em}-\hspace*{-0.2em}M_{11}
\hspace*{-0.2em}\Upsilon(\boldsymbol{p}_{k}))^{-1}\hspace*{-0.2em}M_{12},
\end{eqnarray}
where $M_{ij}$ with $i,j\hspace*{-0.2em}\in\hspace*{-0.2em}\mathds{N}_{[1,2]}$ are defined in \eqref{m11:2}-\eqref{m11:5}, 
respectively. By applying the full-block S-procedure 
\cite[Theorem 8]{Scherer2001LPV} with $\Xi_{22}\hspace*{-0.2em}\prec\hspace*{-0.2em}\mathbf{0}$, the inequality \eqref{m10:2} 
is no longer multiconvex in $\boldsymbol{p}_{k}$. 
Therefore, it suffices to verify a finite set of matrix inequalities at the vertices $\hat{\boldsymbol{p}}_{k}^{i}$ of the polytopic scheduling set 
$\mathbb{P}$, that is 
\begin{eqnarray}
\hspace*{-0.6em}\left[\hspace*{-0.2em}
\begin{array}{c}
\hspace*{-0.2em}I\hspace*{-0.2em}  \\
\hspace*{-0.2em}\Upsilon(\hat{\boldsymbol{p}}_{k}^{i})\hspace*{-0.2em}  
\end{array} 
\hspace*{-0.2em}\right]^{\hspace*{-0.2em}\top} \hspace*{-0.2em}
\left[\hspace*{-0.2em}
\begin{array}{cc}
\hspace*{-0.2em}\Xi_{11}\hspace*{-0.2em}  &
\hspace*{-0.2em}\Xi_{12}\hspace*{-0.2em}   \\
\hspace*{-0.2em}\Xi_{12}^{\top}\hspace*{-0.2em}  &
\hspace*{-0.2em}\Xi_{22}\hspace*{-0.2em} 
\end{array} 
\hspace*{-0.2em}\right]\hspace*{-0.2em}
\left[\hspace*{-0.2em}
\begin{array}{c}
\hspace*{-0.2em}I\hspace*{-0.2em}  \\
\hspace*{-0.2em}\Upsilon(\hat{\boldsymbol{p}}_{k}^{i})\hspace*{-0.2em}  
\end{array} 
\hspace*{-0.2em}\right]
& \hspace*{-0.3em}\succeq\hspace*{-0.2em} 
\boldsymbol{0}, i\in\mathds{N}_{[1,n]}.
\end{eqnarray}
Substituting \eqref{r1} into \eqref{m9} yields 
\begin{eqnarray}
\label{m15}
\hspace*{-0.6em}\left[\hspace*{-0.2em}
\begin{array}{ccc}
\hspace*{-0.2em}P\hspace*{-0.2em}  &
\hspace*{-0.2em}\boldsymbol{0}\hspace*{-0.2em} &
\hspace*{-0.2em}\boldsymbol{0}\hspace*{-0.2em} \\
\hspace*{-0.2em}\boldsymbol{0}\hspace*{-0.2em}  &
\hspace*{-0.2em}I_{\ell}\hspace*{-0.2em}\otimes\hspace*{-0.2em} P\hspace*{-0.2em} &
\hspace*{-0.2em}\boldsymbol{0}\hspace*{-0.2em} \\
\hspace*{-0.2em}K_{d0}\hspace*{-0.2em}  &
\hspace*{-0.2em}\bar{K}_{d}(I_{\ell}\otimes P)\hspace*{-0.2em} &
\hspace*{-0.2em}\boldsymbol{0} \hspace*{-0.2em} \\
\hspace*{-0.2em}\boldsymbol{0}\hspace*{-0.2em}  &
\hspace*{-0.2em}I_{\ell}\otimes K_{d0}\hspace*{-0.2em} &
\hspace*{-0.2em}I_{\ell}\otimes\bar{K}_{d}(I_{\ell}\otimes P) \hspace*{-0.2em}
\end{array} 
\hspace*{-0.2em}\right]\hspace*{-0.2em}=\hspace*{-0.2em}\boldsymbol{G}\mathcal{F}.
\end{eqnarray}
Let $Z_{0}\hspace*{-0.2em}=\hspace*{-0.2em}K_{d0}P$ and $\bar{Z}\hspace*{-0.2em}=\hspace*{-0.2em}\bar{K}_{d}(I_{\ell}\hspace*{-0.2em}\otimes \hspace*{-0.2em}P)$, then, \eqref{m10:4} holds. Note that if there exist $Z_{0}$ and $\bar{Z}$ satisfying 
\eqref{m10:4}, then the controller constructed according to \eqref{e10} yields $\mathcal{N}_{cl}$ satisfying \eqref{m9}.
This completes the proof. 
\end{IEEEproof}

The stability certificates in Theorem \ref{theo2} depend on several vertex-wise LMIs over a polytopic LPV scheduling signal set. 
In practical applications, such as a wheeled vehicle with an LPV representation, 
the scheduling parameter $\boldsymbol{p}_k$ may represent the vehicle velocity and typically has a low dimension, thereby 
avoiding a 
significant computational burden.
\begin{remark}  
To handle the dependency of 
$F(\boldsymbol{p}_{k})$ on $\boldsymbol{p}_{k}$, we 
express the matrix inequality \eqref{t1} win the equivalent data-dependent form \eqref{m10:1}-\eqref{m10:3},   
using the full-block S-procedure \cite[Section IV]{Abbas2016A}, 
such that the decision variables $P$ and $F_{Q}$ can be determined by solving 
a finite number of inequality constraints at the 
vertices of the polytopic scheduling set $\mathds{P}$. 
Therefore, the controller gains can be scheduled once  
constraint \eqref{m10:4} is satisfied. Our theoretical results apply to perturbed LPV systems, where the perturbations may stem from unmodeled dynamics or external disturbances. 
Unlike the existing work \cite{Verhoek2024Direct}, we do not explicitly model the perturbations as process noise. Nevertheless, the proposed framework can potentially be extended to account for measurement noise during data collection. 
In this case, one may consider $\hat{\boldsymbol{x}}\hspace*{-0.2em}=\hspace*{-0.2em}
\boldsymbol{x}\hspace*{-0.2em}+\hspace*{-0.2em}\hat{\boldsymbol{d}}$, where $\hat{\boldsymbol{d}}$ denotes the measurement noise, and 
combine the above theorems with \cite[Remark 3]{Bisoffi2022Data} to address this issue. 
\end{remark}

\begin{remark}
\label{remarkx}
Note that the states of practical LPV control systems are typically subject to physical constraints \cite{Zeilinger2016Auto}, such as the steering-angle and velocity limits of a wheeled vehicle. 
In this context, we introduce a state constraint set \(\mathcal{X}\hspace*{-0.2em}=\hspace*{-0.2em}
\{\boldsymbol{x}_k \hspace*{-0.1em}|\hspace*{-0.1em} H_i \boldsymbol{x}_k \hspace*{-0.2em}\le\hspace*{-0.2em} h_i, \; i \hspace*{-0.2em}=\hspace*{-0.2em} 1, \dots, n\} \hspace*{-0.2em}\subset\hspace*{-0.2em} \mathbb{R}^n\), which is  typically represented by 
a polytope. Here, $H_i$ is the $i$-th row of $H$, representing the half-space normal vector of the $i$-th constraint, and $h_i$ 
denotes the corresponding upper bound \cite{Zeilinger2016Auto,He2021CSL}. 
We consider the Lyapunov level set $\mathcal{E}(P) \hspace*{-0.2em}=\hspace*{-0.2em} \{\boldsymbol{x}_k \mid \boldsymbol{x}_k^\top P \boldsymbol{x}_k \hspace*{-0.2em}\le\hspace*{-0.2em} 1\}$ as an invariant set of the LPV system. 
 To ensure that $\mathcal{E}(P)$ is always contained within the state constraint set $\mathcal{X}$, we impose the following 
additional LMI condition. 
\begin{eqnarray*}
     \max_{\boldsymbol{x}^\top P \boldsymbol{x}_k \hspace*{-0.2em}\le\hspace*{-0.2em} 1} (H_i \boldsymbol{x}_k)^2 \hspace*{-0.2em}\le\hspace*{-0.2em} h_i^2 \quad \text{implies} \quad H_i P^{-1} H_i^\top \hspace*{-0.2em}\le\hspace*{-0.2em} h_i^2. 
\end{eqnarray*}
By applying Schur's complement, we obtain 
\begin{align}
\label{LMIadd}
\begin{bmatrix}
    h_i^2 & H_i \\
    H_i^\top & P
    \end{bmatrix} \succeq 0,\quad i = 1, \dots, n,
\end{align}
which ensures that the system state remains within the convex working region, that is, the system state is bounded.
\end{remark}

\begin{remark}
The conservatism of the proposed data-driven LPV framework originates from two main aspects. The first aspect is data uncertainty due to bounded disturbances and the finite length of the collected data. The second aspect arises from the convexification using the full-block $\mathrm{S}$-procedure and Petersen’s lemma. These tools introduce auxiliary variables to obtain LMI conditions and lead to a relaxation in LPV analysis.
\end{remark}

\subsection{Event-Triggered Communication Scheme}
\label{subsec303}
In addition to ensuring the closed-loop stability of the LPV system, we aim to improve communication efficiency by incorporating an event-triggered communication mechanism into the data-driven robust control framework. The sequence of transmission instants is denoted by 
$\{\hat{k}_{i}\},i\hspace*{-0.2em}\in\hspace*{-0.2em}\mathds{N}$. Under the zero-order hold (ZOH), the control input is given by 
\begin{eqnarray}
\label{et111}
\boldsymbol{u}_{k}\hspace*{-0.2em}=\hspace*{-0.2em}K_{d}(\boldsymbol{p}_{\hat{k}_{i}})\boldsymbol{x}_{\hat{k}_{i}}, 
k\hspace*{-0.2em}\in\hspace*{-0.2em}\mathds{N}_{[\hat{k}_{i},\hat{k}_{i+1})}. 
\end{eqnarray}
Accordingly, the closed-loop 
LPV system can be described as 
\begin{eqnarray}
\hspace*{-0.6em} \boldsymbol{x}_{k+1}\hspace*{-0.3em}& \hspace*{-0.6em}=\hspace*{-0.6em} &\hspace*{-0.3em}
A_{d}(\boldsymbol{p}_{k})\boldsymbol{x}_{k}\hspace*{-0.2em}+\hspace*{-0.2em}B_{d}(\boldsymbol{p}_{k})\boldsymbol{u}_{\hat{k}_{i}}
\hspace*{-0.2em}+\hspace*{-0.2em}\boldsymbol{\omega_{k}}, \notag \\
\hspace*{-0.6em} \hspace*{-0.3em}& \hspace*{-0.6em}=\hspace*{-0.6em} &\hspace*{-0.3em}
(A_{d}(\boldsymbol{p}_{k})+B_{d}K_{d}(\boldsymbol{p}_{k}))\boldsymbol{x}_{k}\hspace*{-0.2em}+\hspace*{-0.2em}\boldsymbol{\nu}_{k}
\hspace*{-0.2em}+\hspace*{-0.2em}\boldsymbol{\widetilde{\omega}}_{k}, \label{e51}
\end{eqnarray}
where $\boldsymbol{\nu}_{k}$ is given by 
\begin{eqnarray}
\hspace*{-0.6em} \boldsymbol{\nu}_{k}\hspace*{-0.3em}& \hspace*{-0.6em}=\hspace*{-0.6em} &\hspace*{-0.3em}
\widetilde{B}_{d}(\boldsymbol{p}_{k})K_{d}(\boldsymbol{p}_{k})\boldsymbol{e}(k)+ \notag \\
\hspace*{-0.6em} \hspace*{-0.3em}& \hspace*{-0.6em}\hspace*{-0.6em} &\hspace*{-0.3em}
\widetilde{B}_{d}(\boldsymbol{p}_{k})(K_{d}(\boldsymbol{p}_{\hat{k}_{i}})-K_{d}(\boldsymbol{p}_{k}))\boldsymbol{x}_{\hat{k}_{i}}, \label{e52}
\end{eqnarray}
where $\boldsymbol{e}(k)\hspace{-0.2em}=\hspace{-0.2em}\boldsymbol{x}_{\hat{k}_{i}}\hspace{-0.2em}-\hspace{-0.2em}\boldsymbol{x}_k$ denotes the sampling-induced error capturing the discrepancy between the most recently transmitted state and the current state. 
Since the system state $\boldsymbol{x}_k$ is accessible and the feedback gain $K_d(\cdot)$ is available from 
Theorem \ref{theo2}, the signal
\begin{equation}
\bar{\boldsymbol{\xi}}_k \triangleq K_d(\boldsymbol{p}_{k})\boldsymbol{e}(k)+\left(K_d(\boldsymbol{p}_{\hat k_i})-K_d(\boldsymbol{p}_k)\right)\boldsymbol{x}_{\hat{k}_{i}}
\end{equation}
can be computed online. Define $\Delta B_d(\boldsymbol{p}_{k})\triangleq B_d(\boldsymbol{p}_{k})-\widetilde{B}_d(\boldsymbol{p}_{k})$. Then, $\boldsymbol{\widetilde\omega}_k$ is given by $\boldsymbol{\widetilde\omega}_k \triangleq \boldsymbol{\omega}_k+\Delta B_d(\boldsymbol{p}_k)\bar{\boldsymbol{\xi}}_k$. An approximation $\widetilde B_d(\boldsymbol{p}_k)$ is obtained offline from the collected data. Based on data-based representation $[\widetilde{\mathcal{A}}_{d},\widetilde{\mathcal{B}}_{d}]\hspace*{-0.2em}=\hspace*{-0.2em}\hspace*{-0.1em}\boldsymbol{X}_{\hspace*{-0.2em}+}\mathcal{\boldsymbol{G}}^{\dagger}\hspace*{-0.1em}$,  $\widetilde{\mathcal{B}}_{d}$ 
can be extracted from columns $n(1\hspace*{-0.2em}+\hspace*{-0.2em}\ell)\hspace*{-0.2em}+\hspace*{-0.2em}1$ through 
$(m\hspace*{-0.2em}+\hspace*{-0.2em}n)(1\hspace*{-0.2em}+\hspace*{-0.2em}\ell)$ of 
$\boldsymbol{X}_{\hspace*{-0.1em}+}\mathcal{\boldsymbol{G}}^{\dagger}\hspace*{-0.1em}$. 
 Hence, the signal $\hspace*{-0.1em}\boldsymbol{\nu}_{k}$ is available online. 
 The data-based closed-loop LPV system \eqref{d3} for $k\hspace*{-0.2em}\in\hspace*{-0.2em}\mathds{N}_{[\hat{k}_{i},\hat{k}_{i+1})}$ is given by 
\begin{eqnarray}
\hspace*{-0.6em} \boldsymbol{x}_{k+1}\hspace*{-0.3em}& \hspace*{-0.6em}=\hspace*{-0.6em} &\hspace*{-0.3em}
(\boldsymbol{X}_{\hspace*{-0.1em}+}\hspace*{-0.2em}-\hspace*{-0.2em}\boldsymbol{W})\mathcal{V}\mathrm{Col}(\boldsymbol{z}_{k},\boldsymbol{p}_{k}\hspace*{-0.2em}\otimes\hspace*{-0.2em}\boldsymbol{p}_{k}
\hspace*{-0.2em}\otimes\hspace*{-0.2em}\boldsymbol{x}_{k})\hspace*{-0.2em}+\hspace*{-0.2em}\boldsymbol{\nu}_{k}
\hspace*{-0.2em}+\hspace*{-0.2em}\boldsymbol{\widetilde{\omega}}_{k}. 
\label{dd1}
\end{eqnarray}

\begin{remark}
To analyze the boundedness of the equivalent disturbance $\widetilde\omega_k$, we recall 
$\boldsymbol{\widetilde\omega}_k \triangleq \boldsymbol{\omega}_k+\Delta B_d(\boldsymbol{p}_k)\boldsymbol{\xi}_k$. 
The term $\boldsymbol{\omega}_k$ is bounded according to Assumption \ref{ass1}. $\Delta B_d(\boldsymbol{p}_k)\hspace{-0.2em}=\hspace{-0.2em}\Delta B_{d0}\hspace*{-0.2em}+\hspace*{-0.3em}\sum_{\imath=1}^{\ell}\boldsymbol{p}_{k}^{[\imath]}\Delta B_{d\imath}$ and $\Delta \mathcal{B}_{d}$ 
can be extracted from the $n(1\hspace*{-0.2em}+\hspace*{-0.2em}\ell)\hspace*{-0.2em}+\hspace*{-0.2em}1$-th column to the 
$(m\hspace*{-0.2em}+\hspace*{-0.2em}n)(1\hspace*{-0.2em}+\hspace*{-0.2em}\ell)$-th column of  
$-\boldsymbol{W}\mathcal{\boldsymbol{G}}^{\dagger}\hspace*{-0.1em}$. 
Since $\boldsymbol{p}_k$ and $\boldsymbol{W}$ are bounded, there exists $\widetilde\Delta_B\hspace{-0.2em}\in\hspace{-0.2em}\mathbb{R}_+$ such that 
$\|\Delta B_d(\boldsymbol{p}_k)\|\hspace{-0.2em}\le \hspace{-0.2em}\widetilde\Delta_B$,  
$\boldsymbol{p}_k\hspace{-0.2em}\in\hspace{-0.2em}\mathbb{P}$. The feedback 
gain $K_d(\cdot)$ is bounded over the scheduling set, hence, there exists $\bar{K}\hspace{-0.2em}\in\hspace{-0.2em}\mathbb{R}_+$ such that 
$\|K_d(\boldsymbol{p}_k)\|\hspace{-0.2em}\le\hspace{-0.2em} \bar K$. 
The system state $\boldsymbol{x}_k$ is typically subject to physical constraints described by the convex set $\mathcal{X}$ defined in Remark \ref{remarkx}, 
within which $\|\boldsymbol{x}\|\hspace{-0.2em}\le \hspace{-0.2em}\boldsymbol{\bar x}$ holds for $k$ and $\hat k_i$. 
By the triangle inequality, 
we obtain $\|\boldsymbol{e}(k)\|\hspace{-0.2em}=\hspace{-0.2em}\|\boldsymbol{x}_{\hat k_i}\hspace{-0.2em}-\hspace{-0.2em}\boldsymbol{x}_k\|\hspace{-0.2em}\le \hspace{-0.2em}\|\boldsymbol{x}_{\hat k_i}\|\hspace{-0.2em}+\hspace{-0.2em}\|\boldsymbol{x}_k\|\hspace{-0.2em}\le \hspace{-0.2em}2\boldsymbol{\bar x}$, leading 
to 
\begin{align}
\|\bar{\boldsymbol\xi}_k\|
&\le \|K_d(\boldsymbol{p}_k)\|\,\|\boldsymbol{e}(k)\|
+ \|K_d(\boldsymbol{p}_{\hat k_i})-K_d(\boldsymbol{p}_k)\|\,\|\boldsymbol{x}_{\hat k_i}\| \nonumber\\
&\le 2 \bar K  \boldsymbol{\bar x} + \big(\|K_d(\boldsymbol{p}_{\hat k_i})\|+\|K_d(\boldsymbol{p}_k)\|\big)\boldsymbol{\bar x} \nonumber\\
&\le 2\bar K\boldsymbol{\bar x} + 2\bar K\boldsymbol{\bar x}
=4\bar K\boldsymbol{\bar x} .
\end{align}
Hence, we have 
\begin{equation}
\|\boldsymbol{\widetilde\omega_k}\|
\le \|\boldsymbol{\omega}_k\|+\|\Delta B_d(\boldsymbol{p}_k)\|\,\|\bar{\boldsymbol{\xi}}_k\|
\le \delta + 4\widetilde\Delta_B\,\bar K\,\boldsymbol{\bar x},
\end{equation}
which shows that $\widetilde\omega_k$ remains bounded.
\end{remark}

Recall the sequence of triggering instants $\hspace*{-0.1em}\{\hat{k}_{i}\},i\hspace*{-0.2em}\in\hspace*{-0.2em}\mathds{N}\hspace*{-0.1em}$ with $\hat{k}_{0}\hspace*{-0.2em}=\hspace*{-0.2em}0$, the triggering logic is specified as  
\begin{eqnarray}
\hspace*{-0.6em} \hat{k}_{i+1}\hspace*{-0.3em}& \hspace*{-0.6em}=\hspace*{-0.6em} &\hspace*{-0.3em}
\inf\left\{ k>\hat{k}_{i}| \boldsymbol{\nu}^{\top}_{k}\Psi_{1}\boldsymbol{\nu}_{k}\geq 
\boldsymbol{x}_{k}^{\top}\Psi_{2}\boldsymbol{x}_{k}+v \right\},  \label{e53}
\end{eqnarray}
where $\hspace*{-0.1em}v\hspace*{-0.2em}\in\hspace*{-0.2em}\mathds{R}_{>0}\hspace*{-0.1em}$ 
is an arbitrary constant, and $\hspace*{-0.1em}\Psi_{1}\hspace*{-0.1em}$ and $\hspace*{-0.1em}\Psi_{2}$ are positive 
definite matrices to be determined later. 

Compared with the event-triggered mechanisms for data-driven LTI systems \cite{Persis2024Event,Digge2022Data}, 
$\boldsymbol{\nu}_{k}$ depends not only on the sampling-induced error $\boldsymbol{e}_{k}$ but also on 
the scheduling signal   
$\boldsymbol{p}_{k}$.  Based on the above analysis, $\boldsymbol{\nu}_{k}$ can be treated as a whole in developing 
the event-triggered mechanism. 
\begin{theorem}
\label{theo3}
For the event-triggered LPV system \eqref{dd1} and the data set $\mathbb{D}_{s}$, 
suppose that the conditions in Lemma \ref{lemm1} and \eqref{m10} hold, yielding the feasible controller in \eqref{e10}. Then, system   
\eqref{dd1} is practically exponentially ISS, if there exist $\mu,\varepsilon_{2}\in\mathds{R}_{+}$ and 
$\Psi_{1}, \Psi_{2}\hspace*{-0.2em}\in\hspace*{-0.2em}\mathbb{S}_{++}^{n}$, such that 
\begin{eqnarray}
\hspace*{-1.4em}\left[\hspace*{-0.2em}
\begin{array}{ccc}
\hspace*{-0.3em}\Omega_{11}  
\hspace*{-0.2em} & \hspace*{-0.2em} 
\Omega_{12} \hspace*{-0.3em}&\hspace*{-0.3em}\Omega_{13} \\
\hspace*{-0.3em}\Omega_{12}^{\top}  
\hspace*{-0.2em} & \hspace*{-0.2em} 
\Omega_{22} \hspace*{-0.3em}&\hspace*{-0.3em}\Omega_{23} \\
\hspace*{-0.3em}\Omega_{13}^{\top}  
\hspace*{-0.2em} & \hspace*{-0.2em} \Omega_{23}^{\top} 
\hspace*{-0.3em}&\hspace*{-0.3em}\Omega_{33}
\end{array} 
\hspace*{-0.2em}\right]\hspace*{-0.3em}&\hspace*{-0.8em}\prec\hspace*{-0.8em}&\boldsymbol{0} \label{m21}
\end{eqnarray}
holding for 
\begin{eqnarray}
\label{m22}
\hspace*{-0.6em}\Omega_{11}&\hspace*{-0.8em}=\hspace*{-0.8em}& 
-\mu\beta_{2}P\hspace*{-0.2em}+\hspace*{-0.2em}P\Psi_{2}P, 
\Omega_{12}\hspace*{-0.2em}=\hspace*{-0.2em}\mu 
F^{\hspace*{-0.1em}\top}\hspace*{-0.2em}
(\boldsymbol{p}_{k})\boldsymbol{X}_{+}, 
\Omega_{13}\hspace*{-0.2em}=\hspace*{-0.2em}
F^{\hspace*{-0.1em}\top}\hspace*{-0.2em}(\boldsymbol{p}_{k}), \notag \\
\hspace*{-0.6em}
\Omega_{22}&\hspace*{-0.8em}=\hspace*{-0.8em}& 
-P\Psi_{1}P+\varepsilon_{2}\mu^{2}\Delta\Delta^{\top}, 
\Omega_{23}\hspace*{-0.2em}=\hspace*{-0.2em}
\boldsymbol{0}_{n\times \mathcal{T}}, 
\Omega_{33}\hspace*{-0.2em}=\hspace*{-0.2em}
-\varepsilon_{2}I_{\mathcal{T}}, \notag
\end{eqnarray}
where $\beta_{2}=\beta_{1}/2$, and $P$ is determined from Theorem \ref{theo2}.
\end{theorem}
\begin{IEEEproof}
By applying the Schur complement to \eqref{m21}, we obtain 
\begin{eqnarray}
&&\hspace*{-1.2em}\mu\left[\hspace*{-0.2em}
\begin{array}{cc}
\hspace*{-0.3em}-\beta_{2}P  
\hspace*{-0.2em} & \hspace*{-0.2em} F^{\top}(\boldsymbol{p}_{k})\boldsymbol{X}^{\top}_{+}  \\
\hspace*{-0.3em}\boldsymbol{X}_{+}F(\boldsymbol{p}_{k})
\hspace*{-0.2em} & \hspace*{-0.2em} \boldsymbol{0}  
\end{array} 
\hspace*{-0.2em}\right]\hspace*{-0.2em}-\hspace*{-0.2em}
\left[\hspace*{-0.2em}
\begin{array}{cc}
\hspace*{-0.3em}-P\Psi_{2}P  
\hspace*{-0.2em} & \boldsymbol{0}  \\
\hspace*{-0.3em}\boldsymbol{0}
\hspace*{-0.2em} & \hspace*{-0.2em} P\Psi_{1}P
\end{array} 
\hspace*{-0.2em}\right]\hspace*{-0.2em}+\hspace*{-0.2em} \varepsilon_{2}^{-1}\hspace*{-0.2em}\times \notag \\
&&\hspace*{-1.2em} 
\left[\hspace*{-0.2em}
\begin{array}{c}
\hspace*{-0.3em}F^{\top}(\boldsymbol{p}_{k})   \\
\hspace*{-0.3em}\boldsymbol{0}
\end{array} 
\hspace*{-0.2em}\right]\hspace*{-0.4em}\left[\hspace*{-0.2em}
\begin{array}{cc}
\hspace*{-0.3em}F(\boldsymbol{p}_{k}) &
\hspace*{-0.3em}\boldsymbol{0}
\end{array} 
\hspace*{-0.2em}\right]\hspace*{-0.2em}+\hspace*{-0.2em}\varepsilon_{2}
\left[\hspace*{-0.2em}
\begin{array}{c}
\hspace*{-0.3em}\boldsymbol{0}   \\
\hspace*{-0.3em}-\mu I
\end{array} 
\hspace*{-0.2em}\right]\hspace*{-0.2em}\Delta\Delta^{\hspace*{-0.1em}\top}\hspace*{-0.2em}
\left[\hspace*{-0.2em}
\begin{array}{cc}
\hspace*{-0.3em}\boldsymbol{0} &
\hspace*{-0.3em}-\mu I
\end{array} 
\hspace*{-0.2em}\right]\hspace*{-0.2em}\prec\hspace*{-0.2em} \boldsymbol{0}. \label{m23}
\end{eqnarray}
With the aid of Petersen's lemma \cite[Lemma 1]{Persis2023Learning}, \eqref{m23} further leads to 
\begin{eqnarray}
&&\hspace*{-2.0em}\mu\left[\hspace*{-0.2em}
\begin{array}{cc}
\hspace*{-0.3em}-\beta_{2}P  
\hspace*{-0.2em} & \hspace*{-0.2em} F^{\top}(\boldsymbol{p}_{k})\boldsymbol{X}^{\top}_{+}  \\
\hspace*{-0.3em}\boldsymbol{X}_{+}F(\boldsymbol{p}_{k})
\hspace*{-0.2em} & \hspace*{-0.2em} \boldsymbol{0}  
\end{array} 
\hspace*{-0.2em}\right]\hspace*{-0.2em}-\hspace*{-0.2em}
\left[\hspace*{-0.2em}
\begin{array}{cc}
\hspace*{-0.3em}P  
\hspace*{-0.2em} & \boldsymbol{0}  \\
\hspace*{-0.3em}\boldsymbol{0}
\hspace*{-0.2em} & \hspace*{-0.2em} P
\end{array} 
\hspace*{-0.2em}\right]\hspace*{-0.2em}\Psi\hspace*{-0.2em} \left[\hspace*{-0.2em}
\begin{array}{cc}
\hspace*{-0.3em}P  
\hspace*{-0.2em} & \boldsymbol{0}  \\
\hspace*{-0.3em}\boldsymbol{0}
\hspace*{-0.2em} & \hspace*{-0.2em} P
\end{array} 
\hspace*{-0.2em}\right]+\notag \\
&&\hspace*{-2.0em}
\left[\hspace*{-0.2em}
\begin{array}{c}
\hspace*{-0.3em}F^{\top}(\boldsymbol{p}_{k})   \\
\hspace*{-0.3em}\boldsymbol{0}
\end{array} 
\hspace*{-0.2em}\right]\hspace*{-0.2em}\boldsymbol{W}^{\hspace*{-0.1em}\top}
\hspace*{-0.3em}\left[\hspace*{-0.2em}
\begin{array}{cc}
\hspace*{-0.3em}\boldsymbol{0} \hspace*{-0.3em}&
\hspace*{-0.3em}-\mu I
\end{array} 
\hspace*{-0.2em}\right]\hspace*{-0.2em}+\hspace*{-0.2em}\left[\hspace*{-0.2em}
\begin{array}{c}
\hspace*{-0.3em}\boldsymbol{0}   \\
\hspace*{-0.3em}-\mu I
\end{array} 
\hspace*{-0.2em}\right]\hspace*{-0.2em}\boldsymbol{W}\hspace*{-0.2em}
\left[\hspace*{-0.2em}
\begin{array}{cc}
\hspace*{-0.3em}F(\boldsymbol{p}_{k})\hspace*{-0.3em} &
\hspace*{-0.3em}\boldsymbol{0}
\end{array} 
\hspace*{-0.2em}\right]\hspace*{-0.2em}\prec\hspace*{-0.2em} \boldsymbol{0} \label{m24}
\end{eqnarray}
where $\Psi\hspace*{-0.2em}=\hspace*{-0.2em}
\mathrm{Blkdiag}(-\Psi_{2},\Psi_{1})$. We then pre- and post-multiply 
both sides of \eqref{m24} with $\mathrm{Blkdiag}(P^{-1},P^{-1})$, and obtain 
\begin{eqnarray}
\hspace*{-0.6em} \mu N(\boldsymbol{W})-\Psi\hspace*{-0.3em}& \hspace*{-0.6em}\prec\hspace*{-0.6em} &\hspace*{-0.3em}\boldsymbol{0},
\label{m25}
\end{eqnarray}
where 
\begin{eqnarray*}
N(\boldsymbol{W})\hspace*{-0.2em}=\hspace*{-0.2em}
\left[\hspace*{-0.2em}
\begin{array}{cc}
\hspace*{-0.3em}-\beta_{2}P^{-1}\hspace*{-0.3em} & \amalg_{12}
\hspace*{-0.3em} \\
\hspace*{-0.3em}\amalg_{12}^{\top}\hspace*{-0.3em} &
\hspace*{-0.3em}\boldsymbol{0} 
\end{array} 
\hspace*{-0.2em}\right]
\end{eqnarray*}
with $\amalg_{12}\hspace*{-0.2em}=\hspace*{-0.2em}P^{-1}F^{\top}\hspace*{-0.2em}(\boldsymbol{p}_{k})(\boldsymbol{X}_{+}
\hspace*{-0.2em}-\hspace*{-0.2em}\boldsymbol{W})^{\hspace*{-0.1em}\top}\hspace*{-0.2em}P^{-1}$. Let
$\boldsymbol{\xi}_{k}\hspace*{-0.2em}=\hspace*{-0.2em}\mathrm{Col}(\boldsymbol{x}_{k},\boldsymbol{\nu}_{k})$, and 
the inequality \eqref{m25} yields
\begin{eqnarray}
\hspace*{-0.6em} \boldsymbol{\xi}^{\top}_{k}N(\boldsymbol{W})\boldsymbol{\xi}_{k}\hspace*{-0.3em}& \hspace*{-0.6em}\leq\hspace*{-0.6em} &\hspace*{-0.3em}v/\mu,
\label{m26}
\end{eqnarray}
which holds under the triggering condition $\boldsymbol{\nu}^{\top}_{k}\Psi_{1}\boldsymbol{\nu}_{k}\hspace*{-0.2em}<\hspace*{-0.2em}\boldsymbol{x}^{\top}_{k}\Psi_{2}\boldsymbol{x}_{k}\hspace*{-0.2em}+\hspace*{-0.2em}v$ over the interval 
$\mathds{N}_{[\hat{k}_{i},\hat{k}_{i+1})}$. Recall the Lyapunov function $V(\boldsymbol{x}_{k})=\hspace*{-0.2em}\boldsymbol{x}_{k}^{\top}P^{-1}\boldsymbol{x}_{k}$, it follows from \eqref{m10} that 
\begin{eqnarray}
& \hspace*{-0.6em}\hspace*{-0.6em} & V(\boldsymbol{x}_{k+1})\hspace*{-0.2em}-\hspace*{-0.2em}V(\boldsymbol{x}_{k}) \notag \\
& \hspace*{-0.6em}\leq\hspace*{-0.6em} &-\beta_{2}V(\boldsymbol{x}_{k})\hspace*{-0.2em}+\hspace*{-0.2em}\sigma\boldsymbol{\omega}_{k}^{\top}P^{-1}\boldsymbol{\omega}_{k}
\hspace*{-0.2em}+\hspace*{-0.2em}\boldsymbol{\xi}_{k}^{\top}N(\boldsymbol{W})\boldsymbol{\xi}_{k} \notag \\
& \hspace*{-0.6em}\hspace*{-0.6em} & +\boldsymbol{\nu}_{k}^{\top}P^{-1}\boldsymbol{\omega}_{k}\hspace*{-0.2em}+\hspace*{-0.2em}
\boldsymbol{\omega}_{k}^{\top}P^{-1}\boldsymbol{\nu}_{k}\hspace*{-0.2em}+\hspace*{-0.2em}
\boldsymbol{\nu}_{k}^{\top}P^{-1}\boldsymbol{\nu}_{k} \notag \\
& \hspace*{-0.6em}\leq\hspace*{-0.6em} & -\beta_{2}V(\boldsymbol{x}_{k})\hspace*{-0.2em}+\hspace*{-0.2em}\sigma\boldsymbol{\zeta}_{k}^{\top}P^{-1}
\boldsymbol{\zeta}_{k}\hspace*{-0.2em}+\hspace*{-0.2em}\boldsymbol{\xi}_{k}^{\top}N(W)\boldsymbol{\xi}_{k} \notag \\
& \hspace*{-0.6em}\leq\hspace*{-0.6em} & -\beta_{2}V(\boldsymbol{x}_{k})\hspace*{-0.2em}+\hspace*{-0.2em}\sigma \mu_{\max}(P^{-1})\|\boldsymbol{\zeta}_{k}\|^{2}
\hspace*{-0.2em}+\hspace*{-0.2em}v/\mu \notag \\
& \hspace*{-0.6em}=\hspace*{-0.6em} & -\beta_{2}V(\boldsymbol{x}_{k})\hspace*{-0.2em}+\hspace*{-0.2em}h(\|\boldsymbol{\zeta}_{k}\|)
\hspace*{-0.2em}+\hspace*{-0.2em}v/\mu,
\label{m27}
\end{eqnarray}
is satisfied with $\boldsymbol{\zeta}_{k}\hspace*{-0.2em}=\hspace*{-0.2em}\boldsymbol{\nu}_{k}
\hspace*{-0.2em}+\hspace*{-0.2em}\boldsymbol{\widetilde{\omega}}_{k}$. 
Following an argument similar to that in the proof of Theorem \ref{theo1}, we obtain
\begin{eqnarray}
\hspace*{-1.0em} \|\boldsymbol{x}_{k}\|
&\hspace*{-0.8em}\leq\hspace*{-0.8em}&e^{-\beta_{2}k/2}\mathcal{R}\|\boldsymbol{x}_{0}\|\hspace*{-0.2em}+\hspace*{-0.2em}\gamma_{2}(\|\boldsymbol{\zeta_{s}}\|_{[0,k-1]})
\hspace*{-0.2em}+\hspace*{-0.2em}c_{0}v, 
\label{m28}
\end{eqnarray}
where $\mathcal{R}\hspace*{-0.2em}=\hspace*{-0.2em}\mu_{\max}^{1/2}(\hspace*{-0.1em}P^{-1}\hspace*{-0.1em})/\mu_{\min}^{1/2}(\hspace*{-0.1em}P^{-1}\hspace*{-0.1em})$ denotes the overshootfactor, and $\gamma_{2}$ is a $\mathcal{K}$-function with the same structure as $\gamma_{1}$ in \eqref{m7}. Moreover,  
\begin{eqnarray}
\hspace*{-1.0em} c_{0}
&\hspace*{-0.8em}=\hspace*{-0.8em}&c_{1}/(\mu(1-e^{-\beta_{2}/2})),
\label{m228}
\end{eqnarray}
where $c_{1}\hspace*{-0.2em}=\hspace*{-0.2em}1/\hspace*{-0.2em}\sqrt{\mu_{\min}(P^{-1})}$. Hence, \eqref{m21} guarantees that the 
event-triggered LPV system \eqref{dd1} is practically exponentially ISS,  which completes the proof.
\end{IEEEproof}

Due to the parameter-varying term $F(\boldsymbol{p}_{k})$, 
directly solving \eqref{m21} via convex optimization is difficult. Therefore, following Subsection \ref{3b}, 
we reduce \eqref{m21} to a finite number of constraints, leading to 
the following theorem. 
\begin{theorem}
\label{theo4}
For the data matrices generated from $\mathbb{D}_{s}$ in \eqref{ee4}, suppose that there exist matrices 
$\hspace*{-0.1em}\Psi_{1},\Psi_{2}\hspace*{-0.3em}\in\hspace*{-0.3em}\mathbb{S}^{n}_{++}$, 
$\hspace*{-0.1em}\boldsymbol{\Xi}\hspace*{-0.2em}\in\hspace*{-0.2em}\mathbb{S}^{2\ell(2n+\mathcal{T})}$, and scalars 
$\mu,\varepsilon_{2}\hspace*{-0.2em}\in\hspace*{-0.2em}\mathds{R}_{+}$, such that 
\begin{subequations}
\label{m29}
\begin{align}
\hspace*{-0.6em}\left[\hspace*{-0.2em}
\begin{array}{cc}
\hspace*{-0.2em}\boldsymbol{M}_{11}\hspace*{-0.2em}  &
\hspace*{-0.2em}\boldsymbol{M}_{12}\hspace*{-0.2em} \\
\hspace*{-0.2em}I\hspace*{-0.2em}  &
\hspace*{-0.2em}\boldsymbol{0}\hspace*{-0.2em} \\
\hspace*{-0.2em}\boldsymbol{M}_{21}\hspace*{-0.2em}  &
\hspace*{-0.2em}\boldsymbol{M}_{22}\hspace*{-0.2em} 
\end{array} 
\hspace*{-0.2em}\right]^{\hspace*{-0.2em}\top}\hspace*{-0.3em}\left[\hspace*{-0.2em}
\begin{array}{cc}
\hspace*{-0.2em}\boldsymbol{\Xi}\hspace*{-0.2em}  &
\hspace*{-0.2em}\boldsymbol{0}\hspace*{-0.2em} \\
\hspace*{-0.2em}\boldsymbol{0}\hspace*{-0.2em}  &
\hspace*{-0.2em}-\boldsymbol{\Phi}\hspace*{-0.2em} 
\end{array} 
\hspace*{-0.2em}\right]\hspace*{-0.2em}
\left[\hspace*{-0.2em}
\begin{array}{cc}
\hspace*{-0.2em}\boldsymbol{M}_{11}\hspace*{-0.2em}  &
\hspace*{-0.2em}\boldsymbol{M}_{12}\hspace*{-0.2em} \\
\hspace*{-0.2em}I\hspace*{-0.2em}  &
\hspace*{-0.2em}\boldsymbol{0}\hspace*{-0.2em} \\
\hspace*{-0.2em}\boldsymbol{M}_{21}\hspace*{-0.2em}  &
\hspace*{-0.2em}\boldsymbol{M}_{22}\hspace*{-0.2em} 
\end{array} 
\hspace*{-0.2em}\right]
& \hspace*{-0.3em}\prec\hspace*{-0.2em} 
\boldsymbol{0}, \label{m29:1} \\
\hspace*{-0.6em}\left[\hspace*{-0.2em}
\begin{array}{c}
\hspace*{-0.2em}I\hspace*{-0.2em}  \\
\hspace*{-0.2em}\boldsymbol{\Upsilon}(\boldsymbol{p}_{k})\hspace*{-0.2em}  
\end{array} 
\hspace*{-0.2em}\right]^{\hspace*{-0.2em}\top} \hspace*{-0.2em}
\left[\hspace*{-0.2em}
\begin{array}{cc}
\hspace*{-0.2em}\boldsymbol{\Xi}_{11}\hspace*{-0.2em}  &
\hspace*{-0.2em}\boldsymbol{\Xi}_{12}\hspace*{-0.2em}   \\
\hspace*{-0.2em}\boldsymbol{\Xi}_{12}^{\top}\hspace*{-0.2em}  &
\hspace*{-0.2em}\boldsymbol{\Xi}_{22}\hspace*{-0.2em} 
\end{array} 
\hspace*{-0.2em}\right]\hspace*{-0.2em}
\left[\hspace*{-0.2em}
\begin{array}{c}
\hspace*{-0.2em}I\hspace*{-0.2em}  \\
\hspace*{-0.2em}\boldsymbol{\Upsilon}(\boldsymbol{p}_{k})\hspace*{-0.2em}  
\end{array} 
\hspace*{-0.2em}\right]
& \hspace*{-0.3em}\succeq\hspace*{-0.2em} 
\boldsymbol{0}, \label{m29:2} \\
\hspace*{-0.6em}\boldsymbol{\Xi}_{22}
& \hspace*{-0.3em}\prec\hspace*{-0.2em} 
\boldsymbol{0}, \label{m29:3} 
\end{align}
\end{subequations}
hold for all $\boldsymbol{p}_{k}\hspace*{-0.2em} \in\hspace*{-0.2em} \mathds{P}$, where 
\begin{subequations}
\label{m30}
\begin{align}
\hspace*{-0.6em}\boldsymbol{\Upsilon}(\boldsymbol{p}_{k})\hspace*{-0.2em}=
&  
\mathrm{Blkdiag}(\mathrm{Diag}(\boldsymbol{p}_{k})\hspace*{-0.2em} \otimes\hspace*{-0.2em} I_{n},
\mathrm{Diag}(\boldsymbol{p}_{k})\hspace*{-0.2em} \otimes\hspace*{-0.2em} I_{n}, \notag\\
\hspace*{-0.6em}
&  
\mathrm{Diag}(\boldsymbol{p}_{k})\hspace*{-0.2em} \otimes\hspace*{-0.2em} I_{\mathcal{T}}), \\
\hspace*{-0.6em}\boldsymbol{M}_{11}\hspace*{-0.2em}=&  
\boldsymbol{0}_{\ell(2n+\mathcal{T})\times \ell(2n+\mathcal{T})}, \label{m30:2}\\
\hspace*{-0.6em}\boldsymbol{M}_{12}\hspace*{-0.2em}=
&\mathrm{Blkdiag}(\boldsymbol{1}_{\ell}\hspace*{-0.3em}\otimes\hspace*{-0.2em}I_{n},
\hspace*{-0.1em}\boldsymbol{1}_{\ell}\hspace*{-0.3em}\otimes\hspace*{-0.2em}I_{n},
\hspace*{-0.1em}\boldsymbol{1}_{\ell}\hspace*{-0.3em}\otimes\hspace*{-0.2em}I_{\mathcal{T}}\hspace*{-0.1em}), \label{m30:3}\\
\hspace*{-0.6em}\boldsymbol{M}_{21}\hspace*{-0.2em}=&
\mathrm{Blkdiag}\left( \left[\hspace*{-0.2em}
\begin{array}{c}
\hspace*{-0.2em}\boldsymbol{0}_{n\times\ell n}\hspace*{-0.2em}  \\
\hspace*{-0.2em}I_{\ell n}\hspace*{-0.2em}  
\end{array} 
\hspace*{-0.2em}\right],\left[\hspace*{-0.2em}
\begin{array}{c}
\hspace*{-0.2em}\boldsymbol{0}_{n\times\ell n}\hspace*{-0.2em}  \\
\hspace*{-0.2em}I_{\ell n}\hspace*{-0.2em}  
\end{array} 
\hspace*{-0.2em}\right],\left[\hspace*{-0.2em}
\begin{array}{c}
\hspace*{-0.2em}\boldsymbol{0}_{\mathcal{T}\times\ell \mathcal{T}}\hspace*{-0.2em}  \\
\hspace*{-0.2em}I_{\ell \mathcal{T}}\hspace*{-0.2em}  
\end{array} 
\hspace*{-0.2em}\right] \right), \label{m30:4}\\
\hspace*{-0.6em}\boldsymbol{M}_{22}\hspace*{-0.2em}=&
\mathrm{Blkdiag}\left( \left[\hspace*{-0.2em}
\begin{array}{c}
\hspace*{-0.2em}I_{n}\hspace*{-0.2em}  \\
\hspace*{-0.2em}\boldsymbol{0}_{\ell n\times n}\hspace*{-0.2em}  
\end{array} 
\hspace*{-0.2em}\right],\left[\hspace*{-0.2em}
\begin{array}{c}
\hspace*{-0.2em}I_{n}\hspace*{-0.2em}  \\
\hspace*{-0.2em}\boldsymbol{0}_{\ell n\times n}\hspace*{-0.2em}  
\end{array} 
\hspace*{-0.2em}\right],\left[\hspace*{-0.2em}
\begin{array}{c}
\hspace*{-0.2em}I_{\mathcal{T}}\hspace*{-0.2em}  \\
\hspace*{-0.2em}\boldsymbol{0}_{\ell \mathcal{T}\times \mathcal{T}}\hspace*{-0.2em}  
\end{array} 
\hspace*{-0.2em}\right] \right), \label{m30:5} \\
\hspace*{-0.6em}\boldsymbol{\Phi}\hspace*{-0.2em}=&  
-\left[\hspace*{-0.2em}
\begin{array}{ccc}
\hspace*{-0.2em}\boldsymbol{\mathcal{Y}}  \hspace*{-0.2em} & \hspace*{-0.2em}\mu(\boldsymbol{\mathcal{X}}_{+}F_{Q})^{\top} 
\hspace*{-0.2em}&\hspace*{-0.2em} F_{Q}^{\top} \\
\hspace*{-0.2em}\mu(\boldsymbol{\mathcal{X}}_{+}F_{Q})  \hspace*{-0.2em} & \hspace*{-0.2em}\bar{\boldsymbol{\mathcal{Y}}} 
\hspace*{-0.2em}&\hspace*{-0.2em} \boldsymbol{0} \\ 
\hspace*{-0.2em}F_{Q}  \hspace*{-0.2em} & \hspace*{-0.2em}\boldsymbol{0}
\hspace*{-0.2em}&\hspace*{-0.2em} \boldsymbol{\mathcal{E}}
\end{array}
\hspace*{-0.2em}\right], \label{m30:f}
\end{align}
\end{subequations}
with $\boldsymbol{1}_{n}\hspace*{-0.2em}=\hspace*{-0.2em}[1,\cdots,1]\hspace*{-0.2em}\in\hspace*{-0.2em}\mathds{R}^{n}$ and 
\begin{subequations} 
\begin{align}
\hspace*{-0.6em}\boldsymbol{\mathcal{Y}}\hspace*{-0.2em}=&  
\mathrm{Blkdiag}(\Omega_{11},\boldsymbol{0}_{\ell n \times \ell n}),   \notag \\ 
\hspace*{-0.6em}\bar{\boldsymbol{\mathcal{Y}}}\hspace*{-0.2em}=&  
\mathrm{Blkdiag}(\Omega_{22},\boldsymbol{0}_{\ell n \times \ell n}), \notag \\
\hspace*{-0.6em}\boldsymbol{\mathcal{E}}\hspace*{-0.2em}=&  
\mathrm{Blkdiag}(\Omega_{33},\boldsymbol{0}_{\ell\mathcal{T}\times \ell\mathcal{T}}), \notag \\
\hspace*{-0.6em}\boldsymbol{\mathcal{X}}_{+}\hspace*{-0.2em}=&  
\mathrm{Blkdiag}(\boldsymbol{X}_{+},I_{\ell}\otimes\boldsymbol{X}_{+}). \notag
\end{align}
\end{subequations}
Here, $\Omega_{\imath\imath}$ with $\imath\hspace*{-0.2em}\in\hspace*{-0.2em}\mathds{N}_{[1,3]}$ are defined in \eqref{m22}. 
Then, a feasible solution to \eqref{m29}, together with the feedback controller \eqref{e10}, yields the data-based event-triggered scheme, 
ensuring that the event-triggered LPV system \eqref{dd1} is practically exponentially ISS.
\end{theorem}
\begin{IEEEproof}
The proof follows the same line reasoning as that of Theorem \ref{theo2}. Accordingly, $P$ and $F_{Q}$ 
obtained from Theorem \ref{theo2} are treated as known quantities, with $\beta_{2}\hspace*{-0.2em}=\beta_{1}/2$. 
We can rewrite \eqref{m21} as 
\begin{eqnarray}
\label{m31}
\hspace*{-0.6em}\boldsymbol{M}^{\top}(\boldsymbol{p}_{k})\boldsymbol{\Phi} 
\boldsymbol{M}(\boldsymbol{p}_{k})\hspace*{-0.2em}\succ\hspace*{-0.2em} \boldsymbol{0},
\end{eqnarray}
where $\boldsymbol{\Phi}$ is defined in \eqref{m30:f} and $\boldsymbol{M}(\boldsymbol{p}_{k})$ satisfies
\begin{eqnarray}
\label{m32}
\hspace*{-0.6em}\boldsymbol{M}(\boldsymbol{p}_{k})\hspace*{-0.2em}=\hspace*{-0.2em} 
\boldsymbol{M}_{22}\hspace*{-0.2em}+\hspace*{-0.2em}
\boldsymbol{M}_{21}\hspace*{-0.2em}\boldsymbol{\Upsilon}(\boldsymbol{p}_{k})(I\hspace*{-0.2em}-\hspace*{-0.2em}
\boldsymbol{M}_{11}
\hspace*{-0.2em}\boldsymbol{\Upsilon}(\boldsymbol{p}_{k}))^{-1}\hspace*{-0.2em}\boldsymbol{M}_{12},
\end{eqnarray}
where $\boldsymbol{M}_{\imath\jmath}, \imath,\jmath\in\mathds{N}_{[1,2]}$ are defined in \eqref{m30:2}-\eqref{m30:5}. 
Based on the full-block S-procedure, we can alternatively determine the event-triggered parameters by solving the convex program \eqref{m29} while ensuring the  exponential input-to-state practical stability of the 
closed-loop LPV system, which completes the proof.   
\end{IEEEproof}

The proposed event-triggered mechanism \eqref{e53} prevents events from being triggered at every instant. To show this, suppose by 
contradiction that an event is triggered at every instant, leading to $\hat{k}_{i+1} \hspace*{-0.2em}=\hspace*{-0.2em} \hat{k}_i 
\hspace*{-0.2em}+\hspace*{-0.2em} 1$. 
Then, the most recent triggering instant satisfies $\hat{k}_i \hspace*{-0.2em}=\hspace*{-0.2em} k$, 
and the sampling-induced error $\boldsymbol{e}(k) \hspace*{-0.2em}=\hspace*{-0.2em} \boldsymbol{x}_{\hat{k}_i} \hspace*{-0.2em}-\hspace*{-0.2em} \boldsymbol{x}_k$ is zero. 
Moreover, since $\boldsymbol{p}_{\hat{k}_i} \hspace*{-0.2em}=\hspace*{-0.2em} \boldsymbol{p}_k$ and $\boldsymbol{x}_{\hat{k}_i} \hspace*{-0.2em}=\hspace*{-0.2em} \boldsymbol{x}_k$, we obtain $\boldsymbol{\nu}_{k}\hspace*{-0.2em}=\hspace*{-0.2em}\widetilde{B}_{d}(\boldsymbol{p}_{k})K_{d}(\boldsymbol{p}_{k})\boldsymbol{e}(k)\hspace*{-0.2em}+\hspace*{-0.2em}\widetilde{B}_{d}(\boldsymbol{p}_{k})(K_{d}(\boldsymbol{p}_{\hat{k}_{i}})\hspace*{-0.2em}-\hspace*{-0.2em}K_{d}(\boldsymbol{p}_{k}))\boldsymbol{x}_{\hat{k}_{i}}\hspace*{-0.2em}=\hspace*{-0.2em} 0$.
However, the triggering condition is satisfied only when
$\boldsymbol{\nu}_k^\top \Psi_1 \boldsymbol{\nu}_k \hspace*{-0.2em}\ge\hspace*{-0.2em} \boldsymbol{x}_k^\top \Psi_2 \boldsymbol{x}_k \hspace*{-0.2em}+\hspace*{-0.2em} v,$
where $\Psi_2 \hspace*{-0.2em}\succeq\hspace*{-0.2em} 0$ and $v \hspace*{-0.2em}>\hspace*{-0.2em} 0$. 
Under $\boldsymbol{\nu}_k = 0$, the left-hand side of the inequality equals zero, 
while the right-hand side is strictly positive since $v>0$. 
Therefore, the triggering condition cannot be satisfied, which contradicts the assumption and implies that the proposed event-triggering mechanism inherently prevents triggering at all instants. 

Note that if $\boldsymbol{\Xi}_{22}\hspace*{-0.2em}\prec\hspace*{-0.2em}\boldsymbol{0}$, inequality \eqref{m29:2} is no longer multiconvex in $\boldsymbol{p}_{k}$, such that a finite set of
matrix inequalities can be specifically satisfied at the vertices
of the convex polytope $\mathds{P}$. 
Before concluding this subsection, we summarize the proposed data-driven event-triggered LPV control method for robust stabilization in the following procedure.  
\begin{procedure}
\label{proc1}
The design procedure of data-driven event-triggered LPV control for 
robust stabilization. \newline
\noindent 1) Construct the data-based description of the 
event-triggered LPV system \eqref{dd1} and determine the data 
length $\mathcal{T}$ according to the 
$\theta$-PoE criterion in  Lemma \ref{lemm1}. \newline
\noindent 2) Obtain a feasible solution to the conditions in Theorem \ref{theo2}, incorporate the additional LMI condition \eqref{LMIadd}, and 
determine the decision variables $P,F_{Q}$ and the constant $\beta_{1}$. \newline
\noindent 3) Obtain a feasible solution to the conditions in Theorem \ref{theo4} 
and verify the stability of the event-triggered LPV system   
\eqref{dd1}. \newline
\noindent 4) Deploy the data-based event-triggered LPV controller 
\eqref{et111}. 
\end{procedure}

For Procedure \ref{proc1}, the full state $\boldsymbol{x}_{k}$ is assumed to be measurable or available from the sensing devices.
If $\boldsymbol{x}_{k}$ is not directly accessible, an estimated state $\hat{\boldsymbol{x}}_{k}$ generated by a state observer can be employed for feedback control and trigger parameter design. In this case, the estimation error $\boldsymbol{x}_{k}-\hat{\boldsymbol{x}}_{k}$ enters both the 
closed-loop LPV dynamics and the triggering condition, and its interaction with the data uncertainty and scheduling dependence needs to be  
incorporated into the stability analysis. Such an observer-based extension is beyond the scope of the present paper and constitutes an interesting direction for future study.

\subsection{Application on Reference Tracking}

In this subsection, we generalize the previous theoretical developments to 
output reference tracking, that is, 
the output $\boldsymbol{y}_{k}$ converges to the desired reference $\boldsymbol{r}_{k}$. To this end, we define an auxiliary state $\boldsymbol{\chi}_{k}\hspace*{-0.2em}\in\hspace*{-0.2em}\mathds{R}^{r}$ as 
\begin{eqnarray}
\label{n01}
    \boldsymbol{\chi}_{k+1}\hspace*{-0.2em}=\hspace*{-0.2em}\boldsymbol{\chi}_{k}\hspace*{-0.2em}+\hspace*{-0.2em}
    (\boldsymbol{y}_{k}\hspace*{-0.2em}-\hspace*{-0.2em}\boldsymbol{r}_{k}),
\end{eqnarray}
where $\boldsymbol{r}_{k}\hspace*{-0.2em}\in\hspace*{-0.2em}\mathds{R}^{r}$ represents the reference signal and $\boldsymbol{y}_{k}$ is defined in \eqref{e1}. We denote the augmented state by $\boldsymbol{\psi}_{k}\hspace*{-0.2em}=\hspace*{-0.2em}\mathrm{Col}(\boldsymbol{x}_{k},\boldsymbol{\chi}_{k})$ $\in\hspace*{-0.2em}\mathds{R}^{\bar{n}}$ with $\bar{n}\hspace*{-0.2em}=\hspace*{-0.2em}n
\hspace*{-0.2em}+\hspace*{-0.2em}r$. Based on \eqref{e1}, 
we obtain the following augmented LPV system  
\begin{eqnarray}
\label{n02}
    \boldsymbol{\psi}_{k+1}\hspace*{-0.2em}=\hspace*{-0.2em}\widehat{A}_{d}(\boldsymbol{p}_{k})\boldsymbol{\psi}_{k}
    \hspace*{-0.2em}+\hspace*{-0.2em}\widehat{B}_{d}(\boldsymbol{p}_{k})\boldsymbol{u}_{k}
    \hspace*{-0.2em}+\hspace*{-0.2em}
    \underbrace{\widehat{E}_{d}
    (\boldsymbol{p}_{k})\boldsymbol{\varpi}_k}_{=\hat{\boldsymbol{\varpi}}_k}, 
\end{eqnarray}
where $\boldsymbol{\varpi}_k\hspace*{-0.2em}=\hspace*{-0.2em}\mathrm{Col}(\boldsymbol{\omega}_k,\boldsymbol{r}_{k})$, 
$\widehat{E}_{d}(\boldsymbol{p}_{k})=\mathrm{Blkdiag}(I,-I)$ and 
\begin{eqnarray}
\label{n03}
\hspace*{-0.6em}\widehat{A}_{d}(\boldsymbol{p}_{k})&\hspace*{-0.6em}=\hspace*{-0.6em}&\hspace*{-0.2em} \left[\hspace*{-0.2em}
\begin{array}{cc}
\hspace*{-0.2em}A_{d}(\boldsymbol{p}_{k})\hspace*{-0.2em}  & \hspace*{-0.2em}\boldsymbol{0} \\
\hspace*{-0.2em}C_{d}(\boldsymbol{p}_{k}) \hspace*{-0.2em}&\hspace*{-0.2em} I
\end{array} 
\hspace*{-0.2em}\right], B_{d}(\boldsymbol{p}_{k})\hspace*{-0.2em}=\hspace*{-0.2em}
\left[\hspace*{-0.2em}
\begin{array}{c}
\hspace*{-0.2em}B_{d}(\boldsymbol{p}_{k})\\
\hspace*{-0.2em}D_{d}(\boldsymbol{p}_{k}) 
\end{array} 
\hspace*{-0.2em}\right]. 
\end{eqnarray}
Note that $\widehat{A}_{d}(\boldsymbol{p}_{k})$ and $\widehat{B}_{d}(\boldsymbol{p}_{k})$ depend affinely 
on $\boldsymbol{p}_{k}$. In view of \eqref{e4}, taking $\widehat{A}_{d}(\boldsymbol{p}_{k})$ as an example yields 
\begin{eqnarray}
\label{n04}
\hspace*{-0.6em}\widehat{A}_{d}(\boldsymbol{p}_{k})&\hspace*{-0.6em}=\hspace*{-0.6em}&\hspace*{-0.2em} 
\widehat{A}_{d0}\hspace*{-0.2em}+\hspace*{-0.3em}\sum\limits_{\imath=1}^{\ell}\boldsymbol{p}_{k}^{[\imath]}\widehat{A}_{d\imath},
\end{eqnarray}
where $\widehat{A}_{d0}\hspace*{-0.2em}=\hspace*{-0.2em}
\left[\hspace*{-0.2em}
\begin{array}{cc}
\hspace*{-0.2em}A_{d0}\hspace*{-0.2em}\hspace*{-0.2em} &\hspace*{-0.2em} \boldsymbol{0}\hspace*{-0.2em}\hspace*{-0.2em}\\
\hspace*{-0.2em}C_{d0} \hspace*{-0.2em}&\hspace*{-0.2em} I\hspace*{-0.2em}
\end{array} 
\hspace*{-0.2em}\right]$ and 
$\widehat{A}_{d\imath}\hspace*{-0.2em}=\hspace*{-0.2em}
\left[\hspace*{-0.2em}
\begin{array}{cc}
\hspace*{-0.2em}A_{d\imath}\hspace*{-0.2em}\hspace*{-0.2em} &\hspace*{-0.2em} \boldsymbol{0}\hspace*{-0.2em}\hspace*{-0.2em}\\
\hspace*{-0.2em}C_{d\imath} \hspace*{-0.2em}&\hspace*{-0.2em} \boldsymbol{0}\hspace*{-0.2em}
\end{array} 
\hspace*{-0.2em}\right]$ with $\imath\hspace*{-0.2em}\in\hspace*{-0.2em}\mathds{N}_{[1,\ell]}$ are  matrices with appropriate 
dimensions. Similarly, $\widehat{B}_{d}(\boldsymbol{p}_{k})$ can be expressed in a form similar to \eqref{n04}. Accordingly, we rewrite 
the augmented LPV system as 
\begin{eqnarray}
\label{n05}
\boldsymbol{\psi}_{k+1}\hspace*{-0.2em}=\hspace*{-0.2em}\widehat{\mathcal{A}}_{d}
\left[\hspace*{-0.2em}
\begin{array}{c}
\hspace*{-0.2em}\boldsymbol{\psi}_{k}\hspace*{-0.2em}\\
\hspace*{-0.2em}\boldsymbol{p}_{k}\hspace*{-0.2em}\otimes\hspace*{-0.2em}\boldsymbol{\psi}_{k} \hspace*{-0.2em}
\end{array} 
\hspace*{-0.2em}\right]\hspace*{-0.2em}+\hspace*{-0.2em}
\widehat{\mathcal{B}}_{d}
\left[\hspace*{-0.2em}
\begin{array}{c}
\hspace*{-0.2em}\boldsymbol{u}_{k}\hspace*{-0.2em}\\
\hspace*{-0.2em}\boldsymbol{p}_{k}\hspace*{-0.2em}\otimes\hspace*{-0.2em}\boldsymbol{u}_{k} \hspace*{-0.2em}
\end{array} 
\hspace*{-0.2em}\right]\hspace*{-0.2em}+\hspace*{-0.2em}\hat{{\boldsymbol{\varpi}}}_{k}, 
\end{eqnarray}
where $\hspace*{-0.1em}\widehat{\mathcal{A}}_{d}\hspace*{-0.2em}=\hspace*{-0.2em}[\widehat{A}_{d0},\widehat{A}_{d1},
\cdots\hspace*{-0.1em},\widehat{A}_{d\ell}]\hspace*{-0.1em}$ and 
$\hspace*{-0.1em}\widehat{\mathcal{B}}_{d}\hspace*{-0.2em}=\hspace*{-0.2em}[\widehat{B}_{d0},\widehat{B}_{d1},
\cdots\hspace*{-0.1em},\widehat{B}_{d\ell}]$.  
We further obtain the following data-based representation of \eqref{n05}
\begin{eqnarray}
\label{n06}
\widehat{\boldsymbol{X}}_{+}\hspace*{-0.2em}=\hspace*{-0.2em}\widehat{\mathcal{A}}
\left[\hspace*{-0.2em}
\begin{array}{c}
\hspace*{-0.2em}\widehat{\boldsymbol{X}}\hspace*{-0.2em}\\
\hspace*{-0.2em}\widehat{\boldsymbol{X}}_{P} \hspace*{-0.2em}
\end{array} 
\hspace*{-0.2em}\right]\hspace*{-0.2em}+\hspace*{-0.2em}
\widehat{\mathcal{B}}
\left[\hspace*{-0.2em}
\begin{array}{c}
\hspace*{-0.2em}\boldsymbol{U}\hspace*{-0.2em}\\
\hspace*{-0.2em}\boldsymbol{U}_{P} \hspace*{-0.2em}
\end{array} 
\hspace*{-0.2em}\right]\hspace*{-0.2em}+\hspace*{-0.2em}\widehat{\boldsymbol{W}}, 
\end{eqnarray}
with the following data matrices constructed from the collected data 
\begin{subequations}
\label{n07}
\begin{align}
\hspace*{-0.6em}\widehat{\boldsymbol{X}}& \hspace*{-0.3em}=\hspace*{-0.3em} 
[\boldsymbol{\psi}_{0},\cdots,\boldsymbol{\psi}_{\widehat{\mathcal{T}}-1}]\hspace*{-0.2em}\in\hspace*{-0.2em}\mathds{R}^{\bar{n}\times\widehat{\mathcal{T}}}\hspace*{-0.2em}, \\
\hspace*{-0.6em}\widehat{\boldsymbol{X}}_{P}& \hspace*{-0.3em}=\hspace*{-0.3em} 
[\boldsymbol{p}_{0}\otimes\boldsymbol{\psi}_{0},\cdots,\boldsymbol{p}_{\widehat{\mathcal{T}}-1}
\otimes\boldsymbol{\psi}_{\widehat{\mathcal{T}}-1}]
\hspace*{-0.2em}\in\hspace*{-0.2em}\mathds{R}^{\ell\bar{n}\times\widehat{\mathcal{T}}}\hspace*{-0.2em}, \\
\hspace*{-0.6em}\widehat{\boldsymbol{W}}& \hspace*{-0.3em}=\hspace*{-0.3em} 
[\hat{\boldsymbol{\varpi}}_{0},\cdots,\hat{\boldsymbol{\varpi}}_{\widehat{\mathcal{T}}-1}]\hspace*{-0.2em}\in\hspace*{-0.2em}\mathds{R}^{\bar{n}\times\widehat{\mathcal{T}}}\hspace*{-0.2em}, \\
\hspace*{-0.6em}\widehat{\boldsymbol{X}}_{+}& \hspace*{-0.3em}=\hspace*{-0.3em} 
[\boldsymbol{\psi}_{1},\cdots,\boldsymbol{\psi}_{\widehat{\mathcal{T}}}]\hspace*{-0.2em}\in\hspace*{-0.2em}\mathds{R}^{\bar{n}\times\widehat{\mathcal{T}}}\hspace*{-0.2em}.
\end{align}
\end{subequations}
The signal sequences $\{\boldsymbol{\psi}_{k}\}_{k=0}^{\widehat{\mathcal{T}}}$ and $\{{\hat{\boldsymbol{\varpi}}_{k}}\}_{k=0}^{\widehat{\mathcal{T}}}$ constitute 
the data set $\widehat{\mathbb{D}}_{s}\hspace*{-0.3em}=\hspace*{-0.3em}\{\boldsymbol{u}_{k},\boldsymbol{p}_{k},\boldsymbol{\psi}_{k},
\hat{{\boldsymbol{\varpi}}}_{k}\}_{k=0}^{\widehat{\mathcal{T}}}$. The perturbation $\widehat{\boldsymbol{W}}$ is assumed to satisfy 
a bounded-energy constraint, similar to that in  Assumption \ref{ass1}. 
\begin{assumption}
\label{ass2}
For the external perturbation $\hspace*{-0.1em}\hat{\boldsymbol{\varpi}}_{k}\hspace*{-0.1em}$ with $k\hspace*{-0.2em}\in\hspace*{-0.2em}\mathds{N}$, it holds that
$\hspace*{-0.1em}\hat{\boldsymbol{\varpi}}_{\hspace*{-0.1em}k}\hspace*{-0.3em}\in\hspace*{-0.3em}\widehat{\mathds{B}}_{\delta}\hspace*{-0.3em}=\hspace*{-0.3em}\{\hspace*{-0.1em}\hat{\boldsymbol{\varpi}}|\|\hspace*{-0.1em}\hat{\boldsymbol{\varpi}}\hspace*{-0.1em}\|
\hspace*{-0.3em}\leq\hspace*{-0.3em}\hat{\delta}\}\hspace*{-0.1em}$, where $\hspace*{-0.1em}\hat{\delta}\hspace*{-0.2em}\in\hspace*{-0.2em}\mathds{R}_{+}$. 
This further 
implies $\widehat{\boldsymbol{W}}\hspace*{-0.1em}\widehat{\boldsymbol{W}}^{\top}\hspace*{-0.4em}\preceq\hspace*{-0.2em}
\widehat{\Delta}\widehat{\Delta}^{\hspace*{-0.1em}\top}\hspace*{-0.1em}$, where 
$\widehat{\Delta}\hspace*{-0.3em}=\hspace*{-0.3em}\sqrt{\widehat{\mathcal{T}}}\hat{\delta} I_{\bar{n}}$.
\end{assumption}

For Assumption \ref{ass2}, we note that the augmented perturbation $\hat{\boldsymbol{\varpi}}_{k}$ includes the reference signal 
$\boldsymbol{r}_{k}$, which is known a priori for the design of the robust control strategy.
Hence, the boundedness of $\|\hat{\boldsymbol{\varpi}}_{k}\|$ provides a natural extension of Assumption \ref{ass1} to the reference tracking setting. 
\begin{lemma}
\label{lemm2}
Suppose that the sequence $\{\mathcal{U}_{k}\}_{k=0}^{\widehat{\mathcal{T}}-1}$ of LPV system \eqref{n02} is 
$\hat{\theta}$-PoE of order $L_{\hat{o}}\hspace*{-0.2em}=\hspace*{-0.2em}(1\hspace*{-0.2em}+\hspace*{-0.2em}\ell)\bar{n}\hspace*{-0.2em}+\hspace*{-0.2em}1$, 
with $\widehat{\mathcal{T}}\hspace*{-0.2em}\geq\hspace*{-0.2em}\bar{n}
(1\hspace*{-0.2em}+\hspace*{-0.2em}\ell)(1\hspace*{-0.2em}+\hspace*{-0.2em}
m(1\hspace*{-0.2em}+\hspace*{-0.2em}\ell))\hspace*{-0.2em}-\hspace*{-0.2em}1$. If the excitation level satisfies $\hat{\theta}\hspace*{-0.2em}>\hspace*{-0.3em}\sqrt{L_{\hat{o}}}
c_{\omega}\hat{\delta}/\rho$, with $c_{\omega}$ and   
$\rho$ defined analogously to those in Lemma \ref{lemm1} for the augmented LPV system, 
then
\begin{eqnarray}
\hspace*{-0.6em} \mathrm{Rank}(\widehat{\Theta})\hspace*{-0.3em}& \hspace*{-0.6em}=\hspace*{-0.6em} &\hspace*{-0.3em}
(1\hspace*{-0.2em}+\hspace*{-0.2em}\ell)(\bar{n}\hspace*{-0.2em}+\hspace*{-0.2em}m), \label{n08}
\end{eqnarray}
where $\widehat{\Theta}\hspace*{-0.2em}=\hspace*{-0.2em}\mathrm{Col}(\boldsymbol{U},\boldsymbol{U}_{\hspace*{-0.2em}P},\widehat{\boldsymbol{X}},\widehat{\boldsymbol{X}}_{\hspace*{-0.2em}P})$ 
and $\mathcal{U}_{k}=\mathrm{Col}(\boldsymbol{u}_{k},\boldsymbol{p}_{k}\hspace*{-0.2em}\otimes\hspace*{-0.2em}\boldsymbol{u}_{k})$.
\end{lemma}
\begin{IEEEproof}
The proof follows the same line of reasoning as that of Lemma \ref{lemm1} and is omitted for brevity.
\end{IEEEproof}

Then, we construct the control strategy $\boldsymbol{u}_{k}\hspace*{-0.2em}=\hspace*{-0.2em}\widehat{K}_{d}(\boldsymbol{p}_{k})\boldsymbol{\psi}_{k}$ with a structure similar to that in \eqref{f1}, where  
$\widehat{K}_{d}(\boldsymbol{p}_{k})$ depends affinely on $\boldsymbol{p}_{k}$. Accordingly, $\boldsymbol{u}_{k}$ can be expressed 
in terms of $\boldsymbol{p}_{k}\hspace*{-0.2em}\otimes\hspace*{-0.2em}\boldsymbol{\psi}_{k}$ as in Subsection \ref{sub303}. 
We substitute $\boldsymbol{u}_{k}$ into \eqref{n02} and represent the unknown system matrices $\widehat{\mathcal{A}}_{d}$ and $\widehat{\mathcal{B}}_{d}$ by the relationship $[\widehat{\mathcal{A}}_{d},\widehat{\mathcal{B}}_{d}]\hspace*{-0.2em}=\hspace*{-0.2em}\hspace*{-0.2em}(\widehat{\boldsymbol{X}}_{+}\hspace*{-0.2em}-\hspace*{-0.2em}\widehat{\boldsymbol{W}})\widehat{\boldsymbol{G}}^{\dagger}$ with
$\widehat{\boldsymbol{G}}\hspace*{-0.2em}=\hspace*{-0.2em}[\widehat{\boldsymbol{X}},\widehat{\boldsymbol{X}}_{\hspace*{-0.1em}P},\widehat{\boldsymbol{U}},\widehat{\boldsymbol{U}}_{\hspace*{-0.1em}P}]$. 
The data-based 
representation of the LPV reference-tracking system can be written as 
\begin{eqnarray}
\hspace*{-1.0em} \boldsymbol{\psi}_{k+1}\hspace*{-0.3em}& \hspace*{-0.8em}=\hspace*{-0.8em} &\hspace*{-0.3em}
(\widehat{\boldsymbol{X}}_{\hspace*{-0.1em}+}\hspace*{-0.2em}-\hspace*{-0.2em}\widehat{\boldsymbol{W}})\widehat{\mathcal{V}}\mathrm{Col}(\boldsymbol{\psi}_{k},\boldsymbol{p}_{k}
\hspace*{-0.2em}\otimes\hspace*{-0.2em}\boldsymbol{\psi}_{k},
\boldsymbol{p}_{k}\hspace*{-0.2em}\otimes\hspace*{-0.2em}\boldsymbol{p}_{k}
\hspace*{-0.2em}\otimes\hspace*{-0.2em}\boldsymbol{\psi}_{k})\hspace*{-0.2em}+\hspace*{-0.2em}\hat{\boldsymbol{\varpi}}_{k},
\label{n09}
\end{eqnarray}
where $\widehat{\mathcal{V}}\hspace*{-0.2em}\in\hspace*{-0.2em}\mathds{R}^{\widehat{\mathcal{T}}\times \bar{n}(1+\ell+\ell^{2})}$ is any matrix satisfying 
\begin{eqnarray}
\widehat{\mathcal{N}}_{cl}\hspace*{-0.2em}=\hspace*{-0.2em}\left[\hspace*{-0.2em}
\begin{array}{ccc}
  I_{\bar{n}}  \hspace*{-0.2em} & \hspace*{-0.2em}\boldsymbol{0} \hspace*{-0.2em}&\hspace*{-0.2em}\boldsymbol{0} \\
    \boldsymbol{0} \hspace*{-0.2em}&\hspace*{-0.2em} I_{\ell}\otimes I_{\bar{n}}  \hspace*{-0.2em}&\hspace*{-0.2em} \boldsymbol{0} \\
    \widehat{K}_{d0}\hspace*{-0.2em}&\hspace*{-0.2em} \widehat{\bar{K}}_{d}\hspace*{-0.2em}&\hspace*{-0.2em} \boldsymbol{0} \\
    \boldsymbol{0} \hspace*{-0.2em}&\hspace*{-0.2em} I_{\ell}\otimes \widehat{K}_{d0} \hspace*{-0.2em}&\hspace*{-0.2em} I_{\ell}\otimes \widehat{\bar{K}}_{d}
\end{array}
\hspace*{-0.2em}\right]\hspace*{-0.4em}=\hspace*{-0.2em}\widehat{\boldsymbol{G}}\widehat{\mathcal{V}} \label{n10}
\end{eqnarray}
and accordingly $\widehat{\mathcal{N}}\hspace*{-0.2em}=\hspace*{-0.2em}[\widehat{\mathcal{A}}_{d},\widehat{\mathcal{B}}_{d}]\widehat{\mathcal{N}}_{cl}$. 
This result leaves room for designing a data-driven controller that stabilizes 
all LPV systems represented by \eqref{n09} for 
any $\hspace*{-0.1em}\hat{\boldsymbol{\varpi}}_{\hspace*{-0.1em}k}\hspace*{-0.3em}\in\hspace*{-0.3em}\widehat{\mathds{B}}_{\delta}\hspace*{-0.3em}=\hspace*{-0.3em}\{\hspace*{-0.1em}\hat{\boldsymbol{\varpi}}|\|\hspace*{-0.1em}\hat{\boldsymbol{\varpi}}\hspace*{-0.1em}\|
\hspace*{-0.3em}\leq\hspace*{-0.3em}\hat{\delta}\}\hspace*{-0.1em}$, thereby achieving robust reference tracking.  
\begin{theorem}
\label{theo5}
For the data matrices generated from $\hspace*{-0.1em}\widehat{\mathbb{D}}_{s}$ in \eqref{n07}, suppose that there exist a matrix 
$P\hspace*{-0.2em}\in\hspace*{-0.2em}\mathbb{S}^{\bar{n}}_{++}$ and scalars 
$\varepsilon_{3}\hspace*{-0.2em}\in\hspace*{-0.2em}\mathds{R}_{>0}$, $\hat{\sigma}\hspace*{-0.2em}\in\hspace*{-0.2em}\mathds{R}_{>1}$ and 
$\beta_{3}\hspace*{-0.2em}\in\hspace*{-0.2em}\mathds{R}_{(0,1)}$, such that the following condition holds for the known 
external perturbation bound $\widehat{\Delta}$
\begin{eqnarray}
\left[\hspace*{-0.2em}
\begin{array}{ccccc}
\hspace*{-0.2em}P  \hspace*{-0.2em} & \hspace*{-0.2em}\boldsymbol{0} \hspace*{-0.2em}&\hspace*{-0.2em}\widehat{F}^{\top}(\boldsymbol{p}_{k})\widehat{\boldsymbol{X}}^{\top}_{+} \hspace*{-0.2em}&\hspace*{-0.2em} P \hspace*{-0.2em}&\hspace*{-0.2em}\widehat{F}^{\top}(\boldsymbol{p}_{k})\hspace*{-0.2em}\\
\hspace*{-0.2em}\boldsymbol{0} \hspace*{-0.2em}&\hspace*{-0.2em} \hat{\sigma} P  \hspace*{-0.2em}&\hspace*{-0.2em} P \hspace*{-0.2em}&\hspace*{-0.2em} \boldsymbol{0} \hspace*{-0.2em}&\hspace*{-0.2em}\boldsymbol{0}\hspace*{-0.2em}\\
\hspace*{-0.2em}\widehat{\boldsymbol{X}}_{+}\widehat{F}(\boldsymbol{p}_{k})\hspace*{-0.2em}&\hspace*{-0.2em} P\hspace*{-0.2em}&\hspace*{-0.2em} P-\varepsilon_{3}\widehat{\Delta} \widehat{\Delta}^{\top} \hspace*{-0.2em}&\hspace*{-0.2em} \boldsymbol{0}\hspace*{-0.2em}&\hspace*{-0.2em}\boldsymbol{0}\hspace*{-0.2em}\\
\hspace*{-0.2em}P \hspace*{-0.2em}&\hspace*{-0.2em} \boldsymbol{0} \hspace*{-0.2em}&\hspace*{-0.2em} \boldsymbol{0} \hspace*{-0.2em}&\hspace*{-0.2em} \beta_{3}^{-1}P\hspace*{-0.2em}&\hspace*{-0.2em}\boldsymbol{0} \hspace*{-0.2em}\\
 \hspace*{-0.2em}\widehat{F}(\boldsymbol{p}_{k}) \hspace*{-0.2em}&\hspace*{-0.2em} \boldsymbol{0} \hspace*{-0.2em}&\hspace*{-0.2em} \boldsymbol{0} \hspace*{-0.2em}&\hspace*{-0.2em} \boldsymbol{0}\hspace*{-0.2em}&\hspace*{-0.2em} \varepsilon_{3}I_{\widehat{\mathcal{T}}}\hspace*{-0.2em}
\end{array}
\hspace*{-0.2em}\right]\hspace*{-0.4em}\succ\hspace*{-0.2em}\boldsymbol{0}, \label{n11}
\end{eqnarray}
where $\widehat{F}(\boldsymbol{p}_{k})\hspace*{-0.2em}=\hspace*{-0.2em}\widehat{\mathcal{V}}\mathrm{Col}(I_{\bar{n}},
\boldsymbol{p}_{k}\hspace*{-0.2em}\otimes\hspace*{-0.2em} I_{\bar{n}},
\boldsymbol{p}_{k}\hspace*{-0.2em}\otimes\hspace*{-0.2em} \boldsymbol{p}_{k}\hspace*{-0.2em}\otimes\hspace*{-0.2em} I_{\bar{n}})P$. 
Then, the closed-loop LPV system \eqref{n09}
is exponentially ISS. 
\end{theorem}
\begin{IEEEproof}
The proof follows the same line of reasoning as that of Theorem \ref{theo1}. The main difference lies in replacing 
$\boldsymbol{X}_{\hspace*{-0.1em}+}$ and $F(\boldsymbol{p}_{k})$ with 
$\widehat{\boldsymbol{X}}_{\hspace*{-0.1em}+}$ and $\widehat{F}(\boldsymbol{p}_{k})$, due to the dimensional change induced by 
the augmented state $\boldsymbol{\psi}_k$. he details are omitted for brevity. 
\end{IEEEproof}

According to Remark \ref{remarkx}, the augmented state $\boldsymbol{\psi}_k$ is also subject to physical constraints described by a convex set $\hat{\mathcal{X}}$. Hence, an LMI condition analogous to \eqref{LMIadd} can be imposed to ensure that the augmented state remains within the prescribed 
working region.

\begin{theorem}
\label{theo6}
For the data matrices generated from $\hspace*{-0.1em}\widehat{\mathbb{D}}_{s}$ in \eqref{n07}, suppose that there exist matrices 
$P\hspace*{-0.2em}\in\hspace*{-0.2em}\mathbb{S}^{\bar{n}}_{++}$, 
$Z_{0}\hspace*{-0.2em}\in\hspace*{-0.2em}\mathds{R}^{m\times \bar{n}}$, 
$\bar{Z}\hspace*{-0.2em}\in\hspace*{-0.2em}\mathds{R}^{m\times \ell \bar{n}}$, 
$\Xi\hspace*{-0.2em}\in\hspace*{-0.2em}\mathbb{S}^{2\ell(4\bar{n}+\widehat{\mathcal{T}})}$, 
$\mathcal{F}\hspace*{-0.2em}\in\hspace*{-0.2em}\mathds{R}^{\widehat{\mathcal{T}}\times \bar{n}(1+\ell+\ell^{2})}$, as well as $F_{Q}$ 
$\in\hspace*{-0.2em}\mathds{R}^{\widehat{\mathcal{T}}(1+\ell)\times\bar{n}(1+\ell)}$, such that the conditions analogous to \eqref{m10} hold for all $\boldsymbol{p}_{k}\in\mathds{P}$, where the matrices $\Upsilon(\boldsymbol{p}_{k})$ and 
$M_{\imath,\jmath},\imath,\jmath\in\mathds{N}_{[1,2]}$ are defined analogously to\eqref{m11}, with $n$ replaced by $\bar{n}$, 
and 
\begin{eqnarray}
\label{n12}
\hspace*{-0.6em}\Phi\hspace*{-0.2em}&\hspace*{-0.6em}=\hspace*{-0.6em}&  
\left[\hspace*{-0.2em}
\begin{array}{ccccc}
\hspace*{-0.2em}\mathcal{Y}  \hspace*{-0.2em} & \hspace*{-0.2em}\boldsymbol{0} \hspace*{-0.2em}&\hspace*{-0.2em}(\widehat{\mathcal{X}}_{+}F_{Q})^{\top} \hspace*{-0.2em}&\hspace*{-0.2em} \mathcal{Y} \hspace*{-0.2em}&\hspace*{-0.2em}F^{\top}_{Q}\hspace*{-0.2em} \\
\hspace*{-0.2em}\boldsymbol{0}  \hspace*{-0.2em} & \hspace*{-0.2em}\hat{\sigma}\mathcal{Y} \hspace*{-0.2em}&\hspace*{-0.2em}\mathcal{Y} \hspace*{-0.2em}&\hspace*{-0.2em} \boldsymbol{0} \hspace*{-0.2em}&\hspace*{-0.2em}\boldsymbol{0}\hspace*{-0.2em}\\
\hspace*{-0.2em}\widehat{\mathcal{X}}_{+}F_{Q}  \hspace*{-0.2em} & \hspace*{-0.2em}\mathcal{Y} \hspace*{-0.2em}&\hspace*{-0.2em}\bar{\mathcal{Y}} \hspace*{-0.2em}&\hspace*{-0.2em} \boldsymbol{0} \hspace*{-0.2em}&\hspace*{-0.2em}\boldsymbol{0}\hspace*{-0.2em}\\
\hspace*{-0.2em}\mathcal{Y}  \hspace*{-0.2em} & \hspace*{-0.2em}\boldsymbol{0} \hspace*{-0.2em}&\hspace*{-0.2em}\boldsymbol{0} \hspace*{-0.2em}&\hspace*{-0.2em} \beta_{3}^{-1}\mathcal{Y} \hspace*{-0.2em}&\hspace*{-0.2em}\boldsymbol{0}\hspace*{-0.2em}\\
\hspace*{-0.2em}F_{Q}  \hspace*{-0.2em} & \hspace*{-0.2em}\boldsymbol{0} \hspace*{-0.2em}&\hspace*{-0.2em}\boldsymbol{0} \hspace*{-0.2em}&\hspace*{-0.2em} \boldsymbol{0} \hspace*{-0.2em}&\hspace*{-0.2em}\mathcal{E}\hspace*{-0.2em}
\end{array}
\hspace*{-0.2em}\right], 
\end{eqnarray}
where 
\begin{subequations} 
\begin{align}
\hspace*{-0.6em}\mathcal{Y}\hspace*{-0.2em}=&  
\mathrm{Blkdiag}(P,\boldsymbol{0}_{\ell \bar{n} \times \ell \bar{n}}), \notag  \\ 
\hspace*{-0.6em}\bar{\mathcal{Y}}\hspace*{-0.2em}=&  
\mathrm{Blkdiag}(P\hspace*{-0.2em}-\hspace*{-0.2em}\varepsilon_{3}\widehat{\Delta}\widehat{\Delta}^{\top},\boldsymbol{0}_{\ell \bar{n} \times \ell \bar{n}}), \notag \\
\hspace*{-0.6em}\mathcal{E}\hspace*{-0.2em}=&  
\mathrm{Blkdiag}(\varepsilon_{3}I_{\widehat{\mathcal{T}}},\boldsymbol{0}_{\ell\widehat{\mathcal{T}}\times \ell\widehat{\mathcal{T}}}), \notag  \\
\hspace*{-0.6em}\widehat{\mathcal{X}}_{+}\hspace*{-0.2em}=&  
\mathrm{Blkdiag}(\widehat{\boldsymbol{X}}_{+},I_{\ell}\otimes\widehat{\boldsymbol{X}}_{+}). \notag
\end{align}
\end{subequations}
Then, the state-feedback control strategy $\boldsymbol{u}_{k}\hspace*{-0.3em}=\hspace*{-0.3em}\widehat{K}_{\hspace*{-0.1em}d}(\boldsymbol{p}_{k})\boldsymbol{\psi}_{k}$ can be constructed as  
\begin{eqnarray}
\label{n15}
    \widehat{K}_{d0}\hspace*{-0.2em}=\hspace*{-0.2em}Z_{0}P^{-1}, \widehat{\bar{K}}_{d}\hspace*{-0.2em}=\hspace*{-0.2em}\bar{Z}(I_{\ell}\otimes P)^{-1}, 
\end{eqnarray}
which ensures the exponential ISS of the closed-loop system \eqref{n02}.
\end{theorem}
\begin{IEEEproof}
    The proof follows the same line of reasoning as that of Theorem \ref{theo2} and is omitted for brevity. 
\end{IEEEproof}

Motivated by Subsection \ref{subsec303}, in what follows, we next introduce an event-triggered communication scheme while ensuring the closed-loop convergence of \eqref{n02}. Using the same sequence of transmission instants $\{\hat{\boldsymbol{k}}_{i}\}_{i\in\mathds{N}}$ and the ZOH implementation as in Subsection \ref{subsec303}, event-triggered control input is given by 
\begin{eqnarray}
\label{n116}
\boldsymbol{u}_{k}\hspace*{-0.2em}=\hspace*{-0.2em}\widehat{K}_{d}(\boldsymbol{p}_{\hat{k}_{i}})\boldsymbol{\psi}_{\hat{k}_{i}}, 
k\hspace*{-0.2em}\in\hspace*{-0.2em}\mathds{N}_{[\hat{k}_{i},\hat{k}_{i+1})}. 
\end{eqnarray}
Then, the closed-loop LPV system is given by   
\begin{eqnarray}
\hspace*{-0.6em} \boldsymbol{\psi}_{k+1}\hspace*{-0.3em}& \hspace*{-0.6em}=\hspace*{-0.6em} &\hspace*{-0.3em}
\widehat{A}_{d}(\boldsymbol{p}_{k})\boldsymbol{\psi}_{k}\hspace*{-0.2em}+\hspace*{-0.2em}\widehat{B}_{d}(\boldsymbol{p}_{k})\boldsymbol{u}_{\hat{k}_{i}}
\hspace*{-0.2em}+\hspace*{-0.2em}\hat{\widetilde{{\boldsymbol{\varpi}}}}_{k}, \notag \\
\hspace*{-0.6em} \hspace*{-0.3em}& \hspace*{-0.6em}=\hspace*{-0.6em} &\hspace*{-0.3em}
(\widehat{A}_{d}(\boldsymbol{p}_{k})+\widehat{B}_{d}\widehat{K}_{d}(\boldsymbol{p}_{k}))\boldsymbol{\psi}_{k}\hspace*{-0.2em}+\hspace*{-0.2em}\hat{\boldsymbol{\nu}}_{k}
\hspace*{-0.2em}+\hspace*{-0.2em}\hat{\widetilde{{\boldsymbol{\varpi}}}}_{k}, \label{n117}
\end{eqnarray}
where $\hat{\widetilde{{\boldsymbol{\varpi}}}}_{k}$ has a structure analogous to that of $\widetilde{{\boldsymbol{\varpi}}}_k$ in Subsection \ref{subsec303}, and  $\hat{\boldsymbol{\nu}}_{k}$ is given by  
\begin{eqnarray}
\hspace*{-0.6em} \hat{\boldsymbol{\nu}}_{k}\hspace*{-0.3em}& \hspace*{-0.6em}=\hspace*{-0.6em} &\hspace*{-0.3em}
\widehat{B}_{d}(\boldsymbol{p}_{k})\widehat{K}_{d}(\boldsymbol{p}_{k})\hat{\boldsymbol{e}}(k)+ \notag \\
\hspace*{-0.6em} \hspace*{-0.3em}& \hspace*{-0.6em}\hspace*{-0.6em} &\hspace*{-0.3em}
\widehat{B}_{d}(\boldsymbol{p}_{k})(\widehat{K}_{d}(\boldsymbol{p}_{\hat{k}_{i}})-\widehat{K}_{d}(\boldsymbol{p}_{k}))
\boldsymbol{\psi}_{\hat{k}_{i}}, \label{n118}
\end{eqnarray}
where $\hat{\boldsymbol{e}}_{k}\hspace*{-0.2em}=\hspace*{-0.2em}\boldsymbol{\psi}_{\hat{k}_{i}}\hspace*{-0.2em}-\hspace*{-0.2em}\boldsymbol{\psi}_{k}$ 
denotes the sampling-induced error, capturing the discrepancy between the most recently transmitted state and the current sampled state.  
Hence, we rewrite 
the event-triggered LPV system \eqref{n117} in the following data-based representation 
\begin{eqnarray}
\hspace*{-0.6em} \boldsymbol{\psi}_{k+1}\hspace*{-0.3em}& \hspace*{-0.6em}=\hspace*{-0.6em} &\hspace*{-0.3em}
(\widehat{\boldsymbol{X}}_{\hspace*{-0.1em}+}\hspace*{-0.2em}-\hspace*{-0.2em}\widehat{\boldsymbol{W}})\widehat{\mathcal{V}}
\left[\hspace*{-0.2em}
\begin{array}{c}
\hspace*{-0.2em}\boldsymbol{\psi}_{k}\hspace*{-0.2em}\\
\hspace*{-0.2em}\boldsymbol{p}_{k}\hspace*{-0.2em}\otimes\hspace*{-0.2em}\boldsymbol{\psi}_{k} \hspace*{-0.2em} \\
\hspace*{-0.2em}\boldsymbol{p}_{k}\hspace*{-0.2em}\otimes\hspace*{-0.2em}\boldsymbol{p}_{k}
\hspace*{-0.2em}\otimes\hspace*{-0.2em}\boldsymbol{\psi}_{k} \hspace*{-0.2em}
\end{array} 
\hspace*{-0.2em}\right]
\hspace*{-0.2em}+\hspace*{-0.2em}\hat{\boldsymbol{\nu}}_{k}
\hspace*{-0.2em}+\hspace*{-0.2em}\hat{\widetilde{{\boldsymbol{\varpi}}}}_{k},
\label{n119}
\end{eqnarray}
Note that $\hat{\boldsymbol{\nu}}_{k}$ in the data-based representation \eqref{n119} indirectly
depends on the data matrices, and $\widehat{\mathcal{B}}_{d}$ can be represented using the offline data matrices, as in Subsection \ref{subsec303}. Given the sequence of triggering instants $\hspace*{-0.1em}\{\hat{k}_{i}\},i\hspace*{-0.2em}\in\hspace*{-0.2em}\mathds{N}\hspace*{-0.1em}$ with $\hat{k}_{0}\hspace*{-0.2em}=\hspace*{-0.2em}0$, 
the triggering condition is defined as 
\begin{eqnarray}
\hspace*{-0.6em} \hat{k}_{i+1}\hspace*{-0.3em}& \hspace*{-0.6em}=\hspace*{-0.6em} &\hspace*{-0.3em}
\inf\left\{ k>\hat{k}_{i}| \hat{\boldsymbol{\nu}}^{\top}_{k}\widehat{\Psi}_{1}\hat{\boldsymbol{\nu}}_{k}\geq 
\boldsymbol{\psi}_{k}^{\top}\widehat{\Psi}_{2}\boldsymbol{\psi}_{k}+
\hat{v} \right\},  \label{n120}
\end{eqnarray}
where $\hspace*{-0.1em}\hat{v}\hspace*{-0.2em}\in\hspace*{-0.2em}\mathds{R}_{>0}\hspace*{-0.1em}$ 
is an arbitrary constant, and $\hspace*{-0.1em}\widehat{\Psi}_{1}\hspace*{-0.1em}$ and $\hspace*{-0.1em}
\widehat{\Psi}_{2}$ are positive 
definite matrices. 
\begin{theorem}
\label{theo7}
For the event-triggered LPV system \eqref{n117} and the data set $\widehat{\mathbb{D}}_{s}$, 
suppose that the conditions of Lemma \ref{lemm2} and Theorem \ref{theo6} hold, yielding the controller in \eqref{n15}. 
Then, system \eqref{n117} 
is practically exponentially ISS, if there exist $\hat{\mu},\varepsilon_{4}\in\mathds{R}_{+}$
and $\widehat{\Psi}_{1}, \widehat{\Psi}_{2}\hspace*{-0.2em}\in\hspace*{-0.2em}\mathbb{S}_{++}^{\bar{n}}$, such that  
\begin{eqnarray}
\hspace*{-1.4em}\left[\hspace*{-0.2em}
\begin{array}{ccc}
\hspace*{-0.3em}\widehat{\Omega}_{11}  
\hspace*{-0.2em} & \hspace*{-0.2em} 
\widehat{\Omega}_{12} \hspace*{-0.3em}&\hspace*{-0.3em}
\widehat{\Omega}_{13} \\
\hspace*{-0.3em}
\widehat{\Omega}_{12}^{\top}  \hspace*{-0.2em} & \hspace*{-0.2em} 
\widehat{\Omega}_{22} \hspace*{-0.3em}&\hspace*{-0.3em} 
\widehat{\Omega}_{23} \\
\hspace*{-0.3em}
\widehat{\Omega}_{13}^{\top}  \hspace*{-0.2em} & \hspace*{-0.2em} 
\widehat{\Omega}_{23}^{\top} 
\hspace*{-0.3em}&\hspace*{-0.3em}
\widehat{\Omega}_{33}
\end{array} 
\hspace*{-0.2em}\right]\hspace*{-0.3em}&\hspace*{-0.8em}\prec\hspace*{-0.8em}&\boldsymbol{0}  \label{n121}
\end{eqnarray}
holds for  
\begin{eqnarray}
\hspace*{-0.6em}\widehat{\Omega}_{11}&\hspace*{-0.8em}=\hspace*{-0.8em}& 
-\hat{\mu}\beta_{4}P\hspace*{-0.2em}+\hspace*{-0.2em}P\widehat{\Psi}_{2}P, 
\widehat{\Omega}_{12}\hspace*{-0.2em}=\hspace*{-0.2em}
\hat{\mu} \hat{F}^{\hspace*{-0.1em}\top}
\hspace*{-0.2em}(\boldsymbol{p}_{k})\widehat{\boldsymbol{X}}_{+}, 
\widehat{\Omega}_{13}\hspace*{-0.2em}=\hspace*{-0.2em}\hat{F}^{\hspace*{-0.1em}\top}\hspace*{-0.2em}(\boldsymbol{p}_{k}), \notag \\
\hspace*{-0.6em}
\widehat{\Omega}_{22}&\hspace*{-0.8em}=\hspace*{-0.8em}& 
-P\widehat{\Psi}_{1}P+
\varepsilon_{4}\hat{\mu}^{2}\widehat{\Delta}\widehat{\Delta}^{\top}, 
\Omega_{23}\hspace*{-0.2em}=\hspace*{-0.2em}\boldsymbol{0}_{\bar{n}\times \widehat{\mathcal{T}}}, 
\Omega_{33}\hspace*{-0.2em}=\hspace*{-0.2em}-\varepsilon_{4}I_{\widehat{\mathcal{T}}}, \notag \label{7788}
\end{eqnarray}
where $\beta_{4}=\beta_{3}/2$, and $P$ and $\beta_{3}$ are obtained from Theorem \ref{theo6}.
\end{theorem}
\begin{IEEEproof}
The proof follows the same line of reasoning as that of Theorem \ref{theo3}. The main difference lies in 
replacing 
$\boldsymbol{X}_{\hspace*{-0.1em}+}$ and $F(\boldsymbol{p}_{k})$ with 
$\widehat{\boldsymbol{X}}_{\hspace*{-0.1em}+}$ and $\widehat{F}(\boldsymbol{p}_{k})$, respectively, due to the dimensional change induced by 
the augmented state $\boldsymbol{\psi}$. The details are omitted for brevity.  
\end{IEEEproof}

The parameter-varying term $\hat{F}(\boldsymbol{p}_{k})$ makes it difficult to solve \eqref{n121} directly via  convex optimization. 
Therefore, following Subsection \ref{3b}, we reduce \eqref{n121} to a finite number of constraints, leading to 
the following theorem. 
\begin{theorem}
\label{theo8}
For the data matrices generated from $\widehat{\mathbb{D}}_{s}$ in \eqref{n07}, suppose that there exist matrices 
$\hspace*{-0.1em}\widehat{\Psi}_{1},\widehat{\Psi}_{2}\hspace*{-0.3em}\in\hspace*{-0.3em}\mathbb{S}^{\bar{n}}_{++}$, 
$\hspace*{-0.1em}\boldsymbol{\Xi}\hspace*{-0.2em}\in\hspace*{-0.2em}\mathbb{S}^{2\ell(2{\color{revisionpurple}\bar{n}+\hat{\mathcal{T}}})}$, and the scalars 
$\hat{\mu},\varepsilon_{4}\hspace*{-0.2em}\in\hspace*{-0.2em}\mathds{R}_{+}$, such that the conditions 
analogous to \eqref{m29} hold for 
all $\boldsymbol{p}_{k}\hspace*{-0.2em}\in\hspace*{-0.2em}\mathds{P}$, 
in which the matrices $\boldsymbol{\Upsilon}(\boldsymbol{p}_{k})$ 
and $\boldsymbol{M}_{\imath,\jmath},\imath,\jmath\in\mathds{N}_{[1,2]}$ 
are defined analogously to \eqref{m30}, with $n$ and $\mathcal{T}$ replaced by $\bar{n}$ and $\hat{\mathcal{T}}$, respectively,
and 
\begin{eqnarray}
\hspace*{-0.6em}\boldsymbol{\Phi}\hspace*{-0.2em}=&  
-\left[\hspace*{-0.2em}
\begin{array}{ccc}
\hspace*{-0.2em}\boldsymbol{\mathcal{Y}}  
\hspace*{-0.2em} & \hspace*{-0.2em}\hat{\mu}(\widehat{\boldsymbol{\mathcal{X}}}_{+}F_{Q})^{\top} 
\hspace*{-0.2em}&\hspace*{-0.2em} F_{Q}^{\top} \\
\hspace*{-0.2em}\hat{\mu}(\widehat{\boldsymbol{\mathcal{X}}}_{+}F_{Q})  
\hspace*{-0.2em} & \hspace*{-0.2em}\bar{\boldsymbol{\mathcal{Y}}} 
\hspace*{-0.2em}&\hspace*{-0.2em} \boldsymbol{0} \\ 
\hspace*{-0.2em}F_{Q}  \hspace*{-0.2em} & \hspace*{-0.2em}\boldsymbol{0}
\hspace*{-0.2em}&\hspace*{-0.2em} \boldsymbol{\mathcal{E}}
\end{array}
\hspace*{-0.2em}\right], \label{aaa}
\end{eqnarray}
with $\boldsymbol{1}_{\bar{n}}\hspace*{-0.2em}=\hspace*{-0.2em}[1,\cdots,1]
\hspace*{-0.2em}\in\hspace*{-0.2em}
\mathds{R}^{\bar{n}}$ and 
\begin{subequations} 
\begin{align}
\hspace*{-0.6em}\boldsymbol{\mathcal{Y}}\hspace*{-0.2em}=&  
\mathrm{Blkdiag}(\widehat{\Omega}_{11},\boldsymbol{0}_{\ell \bar{n} \times \ell \bar{n}}),  \notag \\ 
\hspace*{-0.6em}\bar{\boldsymbol{\mathcal{Y}}}\hspace*{-0.2em}=&  
\mathrm{Blkdiag}(\widehat{\Omega}_{22},\boldsymbol{0}_{\ell \bar{n} \times \ell \bar{n}}), \notag \\
\hspace*{-0.6em}\boldsymbol{\mathcal{E}}\hspace*{-0.2em}=&  
\mathrm{Blkdiag}(\widehat{\Omega}_{33},\boldsymbol{0}_{\ell\widehat{\mathcal{T}}\times \ell
\widehat{\mathcal{T}}}), \notag \\
\hspace*{-0.6em}\widehat{\boldsymbol{\mathcal{X}}}_{+}\hspace*{-0.2em}=&  
\mathrm{Blkdiag}(\widehat{\boldsymbol{X}}_{+},I_{\ell}\otimes\widehat{\boldsymbol{X}}_{+}). \notag
\end{align}
\end{subequations}
Here, $\widehat{\Omega}_{\imath\imath}$ with $\imath\hspace*{-0.2em}\in\hspace*{-0.2em}\mathds{N}_{[1,3]}$ 
are defined in \eqref{7788}. Hence, a feasible solution to the conditions of this theorem yields the data-based event-triggered scheme and 
the feedback controller \eqref{n116}, ensuring that the event-triggered LPV system 
\eqref{n117} is practically exponentially ISS.
\end{theorem}
\begin{IEEEproof}
The proof follows the same line of reasoning as that of Theorem \ref{theo4} and is omitted for brevity.  
\end{IEEEproof}

The augmented state increases the dimensions of the data matrices and optimization variables, and reducing the resulting computational complexity warrants further study. The proposed data-driven event-triggered LPV control procedure for robust reference tracking is summarized as follows. 
\begin{procedure}
\label{proc2}
The design procedure for data-driven event-triggered LPV control for 
robust reference tracking. \newline
\noindent 1) Construct the auxiliary integral compensator \eqref{n01} 
and obtain the augmented LPV systems \eqref{n02}. Determine 
the data length $\widehat{\mathcal{T}}$ according to the  
$\hat{\theta}$-PoE criterion in Lemma  
\ref{lemm2}. \newline
\noindent 2) Obtain a feasible solution to the conditions in Theorem \ref{theo6}, incorporate the additional LMI condition \eqref{LMIadd}, and  
determine the decision variables $P,F_{Q}$ and the 
constant $\beta_{3}$. \newline
\noindent 3) Obtain a feasible solution to the conditions in Theorem \ref{theo8} 
and verify the stability of the augmented event-triggered 
LPV system 
\eqref{n02}. \newline
\noindent 4) Deploy the data-based event-triggered LPV controller \eqref{n116}
for reference tracking.  
\end{procedure}

\section{Simulation Examples}
\label{section4}
This section verifies the effectiveness of the proposed theoretical developments in Section \ref{section3} through three 
numerical examples and one practical example, implemented in \texttt{Matlab R2023a} using 
\texttt{YALMIP} 
\cite{Lofberg2004Yalmip} and the \texttt{Mosek} solver \cite{Aps2019Mosek}. 

\begin{example}
We consider the LPV systems described by \eqref{e1} with dimensions $n\hspace*{-0.2em}=\hspace*{-0.2em}2$, 
$m\hspace*{-0.2em}=\hspace*{-0.2em}1$, and 
$\ell\hspace*{-0.2em}=\hspace*{-0.2em}2$. The system matrices  
in \eqref{e4} are specified as  
\begin{eqnarray*}
\hspace*{-0.6em}A_{d0}&\hspace*{-0.8em}=\hspace*{-0.8em}& \hspace*{-0.2em} 
\left[\hspace*{-0.2em}
\begin{array}{cc}
\hspace*{-0.2em}0.2485 
\hspace*{-0.2em} & \hspace*{-0.2em} -1.0355  \\
\hspace*{-0.2em}0.8910   
\hspace*{-0.2em} & \hspace*{-0.2em}0.4065
\end{array}
\hspace*{-0.2em}\right], 
A_{d1}\hspace*{-0.2em}=\hspace*{-0.3em}  
\left[\hspace*{-0.2em}
\begin{array}{cc}
\hspace*{-0.2em}-0.0063
\hspace*{-0.2em} & \hspace*{-0.2em} -0.0938  \\
\hspace*{-0.2em}0.0000  
\hspace*{-0.2em} & \hspace*{-0.2em}0.0188
\end{array}
\hspace*{-0.2em}\right], \notag \\
\hspace*{-0.6em}A_{d2}&\hspace*{-0.8em}=\hspace*{-0.8em}& \hspace*{-0.2em} 
\left[\hspace*{-0.2em}
\begin{array}{cc}
\hspace*{-0.2em}-0.0063 
\hspace*{-0.2em} & \hspace*{-0.2em} -0.0938  \\
\hspace*{-0.2em}0.0000   
\hspace*{-0.2em} & \hspace*{-0.2em}0.0188
\end{array}
\hspace*{-0.2em}\right], 
B_{d0}^{\top}\hspace*{-0.2em}=\hspace*{-0.3em}  
\left[\hspace*{-0.2em}
\begin{array}{cc}
\hspace*{-0.2em}0.3190
\hspace*{-0.2em} & \hspace*{-0.2em} -1.3080 
\end{array}
\hspace*{-0.2em}\right], \notag \\
\hspace*{-0.6em}B_{d1}^{\top}&\hspace*{-0.8em}=\hspace*{-0.8em}& 
\hspace*{-0.2em} 
\left[\hspace*{-0.2em}
\begin{array}{cc}
\hspace*{-0.2em}0.3000 
\hspace*{-0.2em} & \hspace*{-0.2em} 1.4000 
\end{array}
\hspace*{-0.2em}\right], 
B_{d2}^{\top}\hspace*{-0.2em}=\hspace*{-0.3em}  
\left[\hspace*{-0.2em}
\begin{array}{cc}
\hspace*{-0.2em}0.0000
\hspace*{-0.2em} & \hspace*{-0.2em} 0.0000
\end{array}
\hspace*{-0.2em}\right].
\label{ss1}
\end{eqnarray*}
The scheduling signal set is given by  
$\mathds{P}\hspace*{-0.2em}=\hspace*{-0.2em}[-1,1]\hspace*{-0.2em}
\times\hspace*{-0.2em}[-1,1]$. Following  
Procedure \ref{proc1}, the required data length satisfies $\mathcal{T}\hspace*{-0.2em}\geq
\hspace*{-0.2em}23$, and we set $\mathcal{T}\hspace*{-0.2em}=\hspace*{-0.2em}23$. Accordingly, 
$23$ data samples are generated from the uniform 
distribution $\mathrm{D}_{\mathrm{ru}}(-1,1)$, yielding 
 $\mathrm{Rank}(\Theta)=9$, as required by Lemma \ref{lemm1}. We set  
$\varepsilon_{1}\hspace*{-0.2em}=\hspace*{-0.2em}0.01$, 
$\sigma\hspace*{-0.2em}=\hspace*{-0.2em}4$ and 
$\beta_{1}\hspace*{-0.2em}=\hspace*{-0.2em}0.2$ in Theorem \ref{theo1}.  
Solving the semidefinite program in Theorem \ref{theo2} then yields 
the feasible solution $P=\mathrm{Col}([1.4082,0.1336], 
[0.1336,0.7021])$. 
Note that to avoid numerical conditioning issues associated with the inversion of  
$P$, we additionally impose the constraint $0.1\hspace*{-0.2em}\leq\hspace*{-0.2em}
\mathrm{Trace}(P)\hspace*{-0.2em}\leq\hspace*{-0.2em}10$ 
when implementing Procedure \ref{proc1}. 
For the event-triggered 
control synthesis, we set $\mu=40$, $\varepsilon_{2}=0.001$, 
$\beta_{2}=\beta_{1}/2=0.1$ in Theorem \ref{theo3} and 
$v\hspace*{-0.2em}=\hspace*{-0.2em}0.01$ in 
the triggering logic \eqref{e53}. Based on the feasible $P$ and 
$F_{Q}$, we can obtain the triggering matrices  
\begin{eqnarray*}
    \hspace*{-0.6em}\Psi_{1}&\hspace*{-0.8em}=\hspace*{-0.8em}& 
    \hspace*{-0.2em} 
    \left[\hspace*{-0.2em}
    \begin{array}{cc}
    \hspace*{-0.2em}166.7528 
    \hspace*{-0.2em} & \hspace*{-0.2em} -14.2105  \\
    \hspace*{-0.2em}-14.2105   
    \hspace*{-0.2em} & \hspace*{-0.2em}36.4492
    \end{array}
    \hspace*{-0.2em}\right]\hspace*{-0.2em}, 
    \Psi_{2}\hspace*{-0.2em}=\hspace*{-0.3em}  
    \left[\hspace*{-0.2em}
    \begin{array}{cc}
    \hspace*{-0.2em}0.0710
    \hspace*{-0.2em} & \hspace*{-0.2em} 0.0266  \\
    \hspace*{-0.2em}0.0266  
    \hspace*{-0.2em} & \hspace*{-0.2em}0.0174
    \end{array}
    \hspace*{-0.2em}\right]\hspace*{-0.3em}.  
    \label{ss2}
\end{eqnarray*}
We set the simulation horizon to $T_{I}=2s$, the discretization step to $k_{\mathrm{step}}=0.01s$, and the 
initial state to $x_{0}=\mathrm{Col}(2,-2)$. The perturbation signal $\boldsymbol{\omega_{k}}$ is randomly generated subject to 
$\|\boldsymbol{\omega_{k}}\|\leq\delta=0.1$. The resulting 
state 
responses of the closed-loop event-triggered LPV system and the inter-event intervals of the event-triggered LPV controller are shown in 
Fig. \ref{Fig1}. This example validates the effectiveness of  
Procedure \ref{proc1} for data-driven event-triggered robust stabilization of LPV systems. 
\begin{figure}[htbp]
\centerline{
\includegraphics[width=8.5cm]{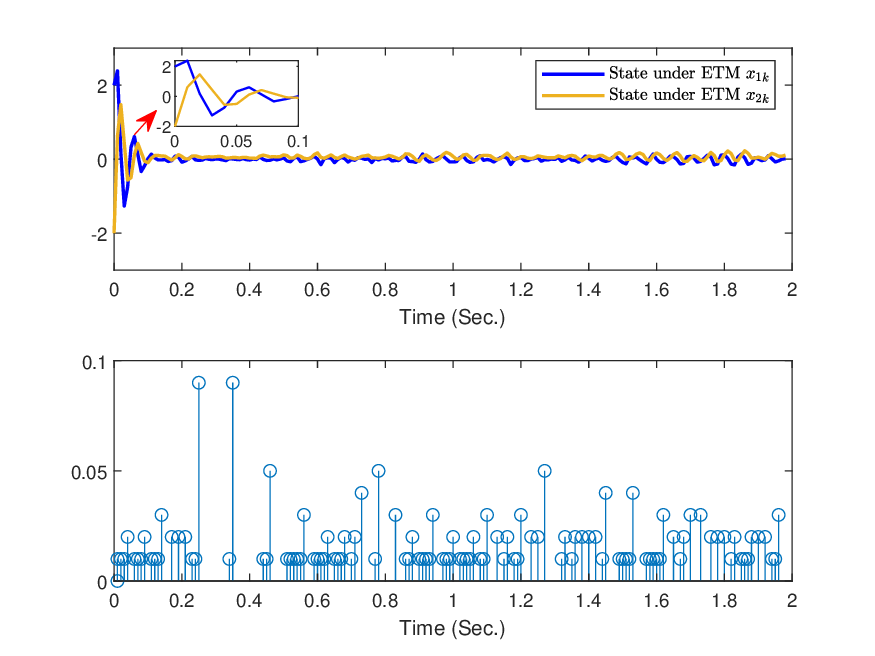} \vspace{-1ex}}
\caption{State responses and inter-event intervals of the 
event-triggered LPV system \eqref{dd1} under the data-driven 
controller \eqref{et111}.}
\label{Fig1}
\end{figure}
\end{example}

In what follows, we present two examples to validate the effectiveness of Procedure \ref{proc2}. 
Example \ref{example2} considers one-dimensional output tracking, whereas Example \ref{example4} considers output tracking for a wheeled vehicle, with potential applications to robotic motion control. 

\begin{example}
    \label{example2}
    We consider the augmented LPV system \eqref{n02} 
    with dimensions $\bar{n}\hspace*{-0.2em}=\hspace*{-0.2em}2$, 
    $m\hspace*{-0.2em}=\hspace*{-0.2em}1$, and 
    $\ell\hspace*{-0.2em}=\hspace*{-0.2em}1$. Then, the system parameters 
    embedded 
    in \eqref{n03} are prescribed by 
    \begin{eqnarray*}
    \hspace*{-0.6em}\widehat{A}_{d0}&\hspace*{-0.8em}=\hspace*{-0.8em}& \hspace*{-0.2em} 
    \left[\hspace*{-0.2em}
    \begin{array}{cc}
    \hspace*{-0.2em}0.3023 
    \hspace*{-0.2em} & \hspace*{-0.2em} 0.0000  \\
    \hspace*{-0.2em}0.1885  
    \hspace*{-0.2em} & \hspace*{-0.2em}1.0000
    \end{array}
    \hspace*{-0.2em}\right], 
    \widehat{A}_{d1}\hspace*{-0.2em}=\hspace*{-0.3em}  
    \left[\hspace*{-0.2em}
    \begin{array}{cc}
    \hspace*{-0.2em}0.5469
    \hspace*{-0.2em} & \hspace*{-0.2em} 0.0000  \\
    \hspace*{-0.2em}0.0997  
    \hspace*{-0.2em} & \hspace*{-0.2em}0.0000
    \end{array}
    \hspace*{-0.2em}\right],  \notag \\
    \hspace*{-0.6em}\widehat{B}_{d0}^{\top}&\hspace*{-0.8em}=\hspace*{-0.8em}& 
    \hspace*{-0.2em} 
    \left[\hspace*{-0.2em}
    \begin{array}{cc}
    \hspace*{-0.2em}0.9902 
    \hspace*{-0.2em} & \hspace*{-0.2em} 0.9672 
    \end{array}
    \hspace*{-0.2em}\right], 
    \widehat{B}_{d1}^{\top}\hspace*{-0.2em}=\hspace*{-0.3em}  
    \left[\hspace*{-0.2em}
    \begin{array}{cc}
    \hspace*{-0.2em}0.6914
    \hspace*{-0.2em} & \hspace*{-0.2em} 0.0470
    \end{array}
    \hspace*{-0.2em}\right].
    \label{ss3}
\end{eqnarray*}
The scheduling set is given by 
$\mathds{P}\hspace*{-0.2em}=\hspace*{-0.2em}[-1,1]$. Following 
Procedure \ref{proc2}, the required data length satisfies $\widehat{\mathcal{T}}\hspace*{-0.2em}\geq
\hspace*{-0.2em}11$, and we set $\widehat{\mathcal{T}}\hspace*{-0.2em}=
\hspace*{-0.2em}17$. This yields $\mathrm{Rank}(\widehat{\Theta})=6$, as required by 
 Lemma \ref{lemm2}. 
 We choose 
$\varepsilon_{3}\hspace*{-0.2em}=\hspace*{-0.2em}0.01$, 
$\hat{\sigma}\hspace*{-0.2em}=\hspace*{-0.2em}4$ and 
$\beta_{3}\hspace*{-0.2em}=\hspace*{-0.2em}0.2$ in Theorem \ref{theo5}. 
Solving the semidefinite program in Theorem \ref{theo6}
yields $P=\mathrm{Col}([0.0692,-0.0210], 
[-0.0210,0.0559])$. 
To avoid numerical conditioning issues associated with the inversion of 
$P$, we additionally impose $0.1\hspace*{-0.2em}\leq\hspace*{-0.2em}
\mathrm{Trace}(P)\hspace*{-0.2em}\leq\hspace*{-0.2em}10$ 
when implementing Procedure \ref{proc2}. 
We define 
$\hat{\mu}=9$, $\varepsilon_{4}=0.001$, 
$\beta_{4}=\beta_{3}/2=0.1$ in Theorem \ref{theo7} and 
$\hat{v}\hspace*{-0.2em}=\hspace*{-0.2em}20$ in 
the triggering logic \eqref{n120}. Based on the feasible $P$ and 
$F_{Q}$, we can obtain the triggering matrices 
\begin{eqnarray*}
    \hspace*{-0.6em}\widehat{\Psi}_{1}&\hspace*{-0.8em}=\hspace*{-0.8em}& 
    \hspace*{-0.2em} 
    \left[\hspace*{-0.2em}
    \begin{array}{cc}
    \hspace*{-0.2em}550.3198 
    \hspace*{-0.2em} & \hspace*{-0.2em} 260.2914  \\
    \hspace*{-0.2em}260.2914   
    \hspace*{-0.2em} & \hspace*{-0.2em}991.6524
    \end{array}
    \hspace*{-0.2em}\right]\hspace*{-0.2em}, 
    \widehat{\Psi}_{2}\hspace*{-0.2em}=\hspace*{-0.3em}  
    \left[\hspace*{-0.2em}
    \begin{array}{cc}
    \hspace*{-0.2em}0.2990
    \hspace*{-0.2em} & \hspace*{-0.2em} 0.0534  \\
    \hspace*{-0.2em}0.0534  
    \hspace*{-0.2em} & \hspace*{-0.2em}0.2027
    \end{array}
    \hspace*{-0.2em}\right]\hspace*{-0.3em}.  
    \label{ss4}
\end{eqnarray*}
We set the simulation horizon to $T_{I}=6s$, and the initial state to $\boldsymbol{\psi}_{0}=\mathrm{Col}(1,1)$. 
The perturbation signal $\hat{\boldsymbol{\varpi}}_k$ is randomly generated subject to 
$\|\hat{\boldsymbol{\varpi}}_k\|\hspace*{-0.2em}\leq
\hspace*{-0.2em}\hat{\delta}=1.01$. 
In this example, we consider two representative one-dimensional reference signals to evaluate the robust tracking performance. 
First, we consider 
the sinusoidal signal $\hspace*{-0.1em}\boldsymbol{r}_{k}\hspace*{-0.2em}=
\hspace*{-0.2em}\sin(4hk)$, where $ 
h\hspace*{-0.2em}=\hspace*{-0.2em}2\pi T_{I}/k_{\mathrm{step}}$. 
The tracking responses of the closed-loop LPV system and the inter-event intervals of the event-triggered LPV controller are shown in 
Fig. \ref{Fig2}. 
Next, we consider the square-wave 
signal $\hspace*{-0.1em}\boldsymbol{r}_{k}$, whose minimum and maximum 
values are $-1$ and $1$, respectively. 
The corresponding tracking responses and inter-event intervals are shown in 
Fig. \ref{Fig3}. 
Accordingly, this example validates the effectiveness of Procedure \ref{proc2} for data-driven event-triggered robust tracking of scalar reference signals in LPV systems.  
\begin{figure}[htbp]
    \centerline{
    \includegraphics[width=8.5cm]{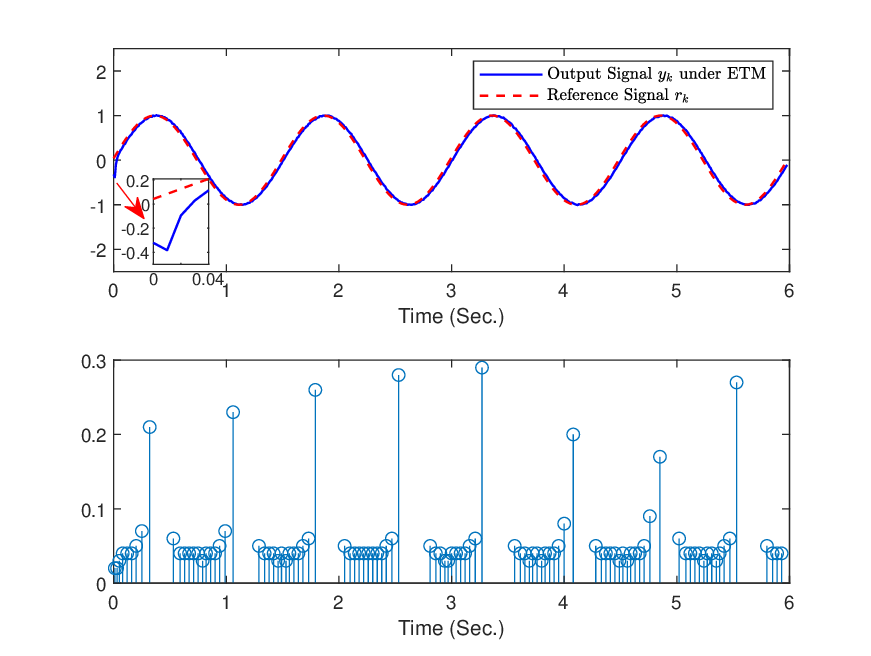} \vspace{-1ex}}
    \caption{Robust tracking responses and inter-event intervals of the 
    event-triggered LPV system \eqref{n02} under the data-driven 
    controller \eqref{n116}-Case 1: sinusoidal reference signal.}
    \label{Fig2}
\end{figure}
\begin{figure}[htbp]
\centerline{
\includegraphics[width=8.5cm]{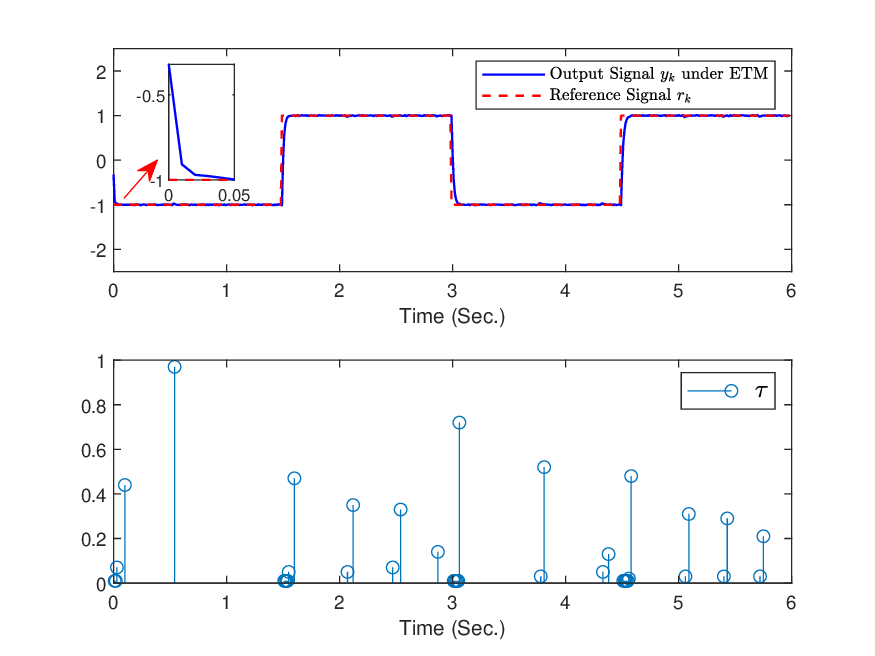} \vspace{-1ex}}
\caption{Robust tracking responses and inter-event intervals of the 
event-triggered LPV system \eqref{n02} under the data-driven 
controller \eqref{n116}-Case 2: square-wave reference signal.
}
\label{Fig3}
\end{figure}
\end{example}

\begin{example}
\label{example4}
We consider a wheeled vehicle described by a unicycle kinematic model, whose error kinematics can be approximated as $\dot{e}_y = v e_{\psi} - a e_y$, and $\dot{e}_{\psi} = -b e_{\psi} + \omega - \kappa(s) v$, where $e_y (\mathrm{m})$ denotes the lateral deviation, $e_{\psi} (\mathrm{rad})$ denotes the heading error, $v (\mathrm{m/s})$ denotes the forward velocity, $\omega (\mathrm{rad/s})$ denotes the command yaw rate. 
$\kappa(s)(\mathrm{m^{-1}})$ denotes the curvature of the reference path at the arc-length coordinate $s (\mathrm{m})$. 
$a, b \hspace{-0.2em}\in\hspace{-0.2em} \mathbb{R}_+$ denote coefficients characterizing the decay rates of the lateral and heading errors induced by hysteresis and tire–ground friction \cite{Morro2011TRO}. 
$\kappa(s)$ is known for a specified reference trajectory. 
We define $u = \omega - \kappa(s) v$ as the equivalent yaw-rate control input and $y = e_y$ as the measured output \cite{Xu2020TITS,Kalaria2022TITS}. In this section, we set 
$ a \in [0.5, 3.0] \, s^{-1} $ and $b \in [0.2, 2.0] \, s^{-1}$.

With a sampling period of $T_s = 0.02\,\mathrm{s}$, applying ZOH discretization to the continuous-time LPV model yields the following discrete-time LPV system 
\begin{equation}
\begin{aligned}
\boldsymbol{x}_{k+1} &= (A_{d0}+\boldsymbol{p}_k A_{d1})\boldsymbol{x}_k + (B_{d0}+\boldsymbol{p}_k B_{d1})\boldsymbol{u}_k + \boldsymbol{\omega}_k,\\
\boldsymbol{y}_k     &= (C_{d0}+\boldsymbol{p}_k C_{d1})\boldsymbol{x}_k,
\end{aligned}
\end{equation}
where $\boldsymbol{x}_k=\mathrm{Col(}e_y(k), e_\psi(k))$, $u_k$ denotes the equivalent yaw-rate control input, $\boldsymbol{y}_k=e_y(k)$ is the measured output, and $\boldsymbol{\omega}_k$ is an additive disturbance aggregating modeling errors and exogenous perturbations. 
The system matrices are specified as
\[
A_{d0} \hspace*{-0.2em}=\hspace*{-0.2em}
\begin{bmatrix}
0.8617 \hspace*{-0.2em}&\hspace*{-0.2em} 0.0378\\
0      \hspace*{-0.2em}&\hspace*{-0.2em} 0.9754
\end{bmatrix}, 
A_{d1} \hspace*{-0.2em}=\hspace*{-0.2em}
\begin{bmatrix}
0 \hspace*{-0.2em}&\hspace*{-0.2em} 0.0199\\
0 \hspace*{-0.2em}&\hspace*{-0.2em} 0
\end{bmatrix}, B_{d0} \hspace*{-0.2em}=\hspace*{-0.2em}
\begin{bmatrix}
0\\
0.0500
\end{bmatrix},
\] 
$B_{d1}\hspace*{-0.2em}=\hspace*{-0.2em}\boldsymbol{0}$, $C_{d0}\hspace*{-0.2em}=\hspace*{-0.2em}[1,0]$, 
 $C_{d1}\hspace*{-0.2em}=\hspace*{-0.2em}\boldsymbol{0}$, and 
$D_{d0}\hspace*{-0.2em}=D_{d1}\hspace*{-0.2em}=\hspace*{-0.2em}\boldsymbol{0}$. 
The scheduling signal is defined as the normalized parameter associated with the vehicle speed, namely, 
$\boldsymbol{p}_k\hspace*{-0.2em}=\hspace*{-0.2em}\frac{2v_k-(v_{\max}+v_{\min})}{v_{\max}-v_{\min}}$, which 
is constrained to $\mathds{P}\hspace*{-0.2em}=\hspace*{-0.2em}[-1,1]$. 

Following  
Procedure \ref{proc2}, the required data length satisfies $\widehat{\mathcal{T}}\hspace*{-0.2em}\geq
\hspace*{-0.2em}17$, and we set $\widehat{\mathcal{T}}\hspace*{-0.2em}=
\hspace*{-0.2em}17$. This yields 
$\mathrm{Rank}(\widehat{\Theta})=8$, as required by  
 Lemma \ref{lemm2} for the augmented LPV system \eqref{n02}. We set 
$\varepsilon_{3}\hspace*{-0.2em}=\hspace*{-0.2em}0.001$, 
$\hat{\sigma}\hspace*{-0.2em}=\hspace*{-0.2em}1.2$ and 
$\beta_{3}\hspace*{-0.2em}=\hspace*{-0.2em}0.3$ in Theorem \ref{theo5}. 
Solving the semidefinite program in Theorem \ref{theo6}
then yields the feasible solution
\begin{eqnarray*}
\hspace*{-0.6em}P&\hspace*{-0.9em}=\hspace*{-0.9em}& \hspace*{-0.2em} 
\left[\hspace*{-0.2em}
\begin{array}{ccc}
\hspace*{-0.2em}0.6445 
\hspace*{-0.2em} & \hspace*{-0.2em}  0.0014  
\hspace*{-0.2em} & \hspace*{-0.2em} -0.0076\\
\hspace*{-0.2em} 0.0014  
\hspace*{-0.2em} & \hspace*{-0.2em}0.6489 
\hspace*{-0.2em} & \hspace*{-0.2em} 0.0045 \\
\hspace*{-0.2em}-0.0076  
\hspace*{-0.2em} & \hspace*{-0.2em}0.0045 
\hspace*{-0.2em} & \hspace*{-0.2em} 0.6946
\end{array}
\hspace*{-0.2em}\right]\hspace*{-0.3em}.
\end{eqnarray*}
To avoid  numerical conditioning issues associated with the inversion of 
$P$, we additionally impose $0.1\hspace*{-0.2em}\leq\hspace*{-0.2em}
\mathrm{Trace}(P)\hspace*{-0.2em}\leq\hspace*{-0.2em}10$ when  
solving the semidefinite program in Theorem \ref{theo6}. 
We set 
$\hat{\mu}\hspace*{-0.2em}=\hspace*{-0.2em}5\mathrm{e}4$, $\varepsilon_{4}\hspace*{-0.2em}=\hspace*{-0.2em}0.01$, 
$\beta_{4}\hspace*{-0.2em}=\hspace*{-0.2em}\beta_{3}/2\hspace*{-0.2em}=\hspace*{-0.2em}0.15$ in Theorem \ref{theo7} and 
$\hat{v}\hspace*{-0.2em}=\hspace*{-0.2em}500$ in 
the triggering logic \eqref{n120}. Using the feasible $P$ and 
$F_{Q}$, we solve the conditions in Theorem \ref{theo8} to 
obtain the triggering matrices 
\begin{eqnarray*}
    \hspace*{-0.6em}\widehat{\Psi}_{1}&\hspace*{-0.8em}=\hspace*{-0.8em}& 
    \hspace*{-0.2em} 
    \left[\hspace*{-0.2em}
    \begin{array}{ccc}
    \hspace*{-0.2em}7.0206e9
    \hspace*{-0.2em} & \hspace*{-0.2em} -2.9398e7  
    \hspace*{-0.2em} & \hspace*{-0.2em} 1.4707e8\\
    \hspace*{-0.2em}-2.9398e7   
    \hspace*{-0.2em} & \hspace*{-0.2em}6.9294e9
    \hspace*{-0.2em} & \hspace*{-0.2em} -8.5972e7\\
    \hspace*{-0.2em}1.4707e8  
    \hspace*{-0.2em} & \hspace*{-0.2em}-8.5972e7
    \hspace*{-0.2em} & \hspace*{-0.2em} 6.0536e9 
    \end{array}
    \hspace*{-0.2em}\right]\hspace*{-0.2em}, \notag \\
    \hspace*{-0.6em}\widehat{\Psi}_{2}
    &\hspace*{-0.8em}=\hspace*{-0.8em}& 
    \hspace*{-0.2em}   
    \left[\hspace*{-0.2em}
    \begin{array}{ccc}
    \hspace*{-0.2em}2102.8 
    \hspace*{-0.2em} & \hspace*{-0.2em} -93.3637  
    \hspace*{-0.2em} & \hspace*{-0.2em} 291.8638\\
    \hspace*{-0.2em}-93.3637   
    \hspace*{-0.2em} & \hspace*{-0.2em}1836.3 
    \hspace*{-0.2em} & \hspace*{-0.2em} -166.0117\\
    \hspace*{-0.2em}291.8638   
    \hspace*{-0.2em} & \hspace*{-0.2em}-166.0117
    \hspace*{-0.2em} & \hspace*{-0.2em} 205.1941
    \end{array}
    \hspace*{-0.2em}\right]\hspace*{-0.2em}. 
    \label{ss6}
\end{eqnarray*}
The simulation horizon is set to $T_I \hspace*{-0.2em}=\hspace*{-0.2em} 45\mathrm{s}$. 
We model the disturbance $\boldsymbol{\omega}_k$ as a bounded random signal, 
resulting in the 
augmented disturbance $\hat{\boldsymbol{\varpi}}_k$ satisfying
$\|\hat{\boldsymbol{\varpi}}_k\| \hspace*{-0.2em}\le\hspace*{-0.2em} \hat{\delta} = 0.1$. 

For reference tracking, $\boldsymbol{r}_k$ denotes the desired lateral offset from the 
geometric centerline of reference path. 
We set $\boldsymbol{r}_k\hspace*{-0.2em}=\hspace*{-0.2em}0$ to enforce offset-free tracking and consider two reference trajectories. 
A circular path is described by $x_r(k)\hspace*{-0.2em}=\hspace*{-0.2em}R\sin\theta_k$ and $y_r(k)\hspace*{-0.2em}=\hspace*{-0.2em}R\cos\theta_k$, whereas 
a lemniscate path is described by $x_r(k)\hspace*{-0.2em}=\hspace*{-0.2em}R\sin\theta_k$ and $y_r(k)\hspace*{-0.2em}=\hspace*{-0.2em}A_l\sin(2\theta_k)$, 
where $\theta_k \hspace*{-0.2em}=\hspace*{-0.2em} 2\pi k/(N\hspace*{-0.2em}-\hspace*{-0.2em}1)$, and $N$ denotes the total number of sampling instants. 
These two trajectories represent typical vehicle motion patterns. In the simulations, We set $R\hspace*{-0.2em}=\hspace*{-0.2em}10$ and 
$A_l\hspace*{-0.2em}=\hspace*{-0.2em}R/2$.
\begin{figure*}[!htb]
\centering
\subfigure[The circle trajectory (T1)]{\includegraphics[width=0.24\hsize,height=0.24\hsize]{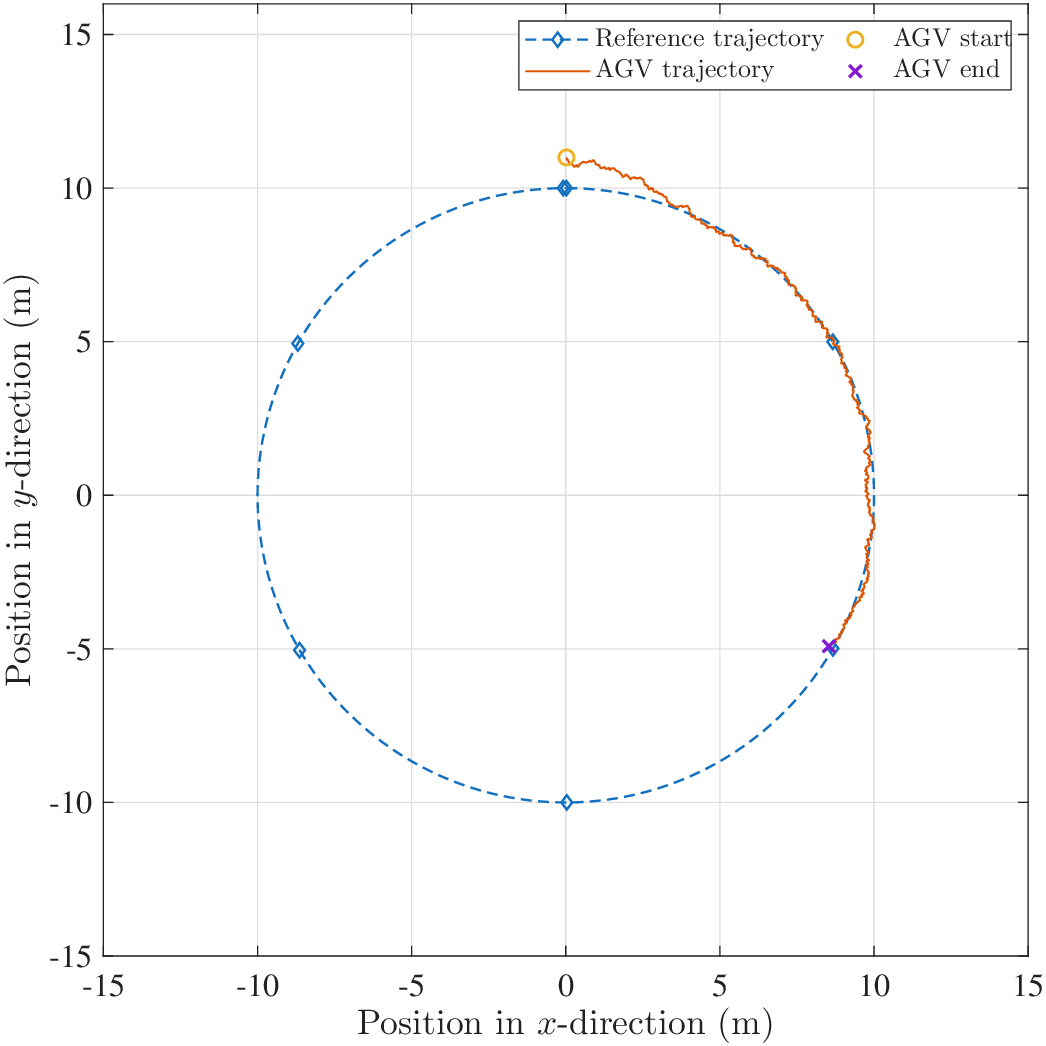}\label{F8: sub_figure1}}
\subfigure[The circle trajectory (T2)]{\includegraphics[width=0.24\hsize,height=0.24\hsize]{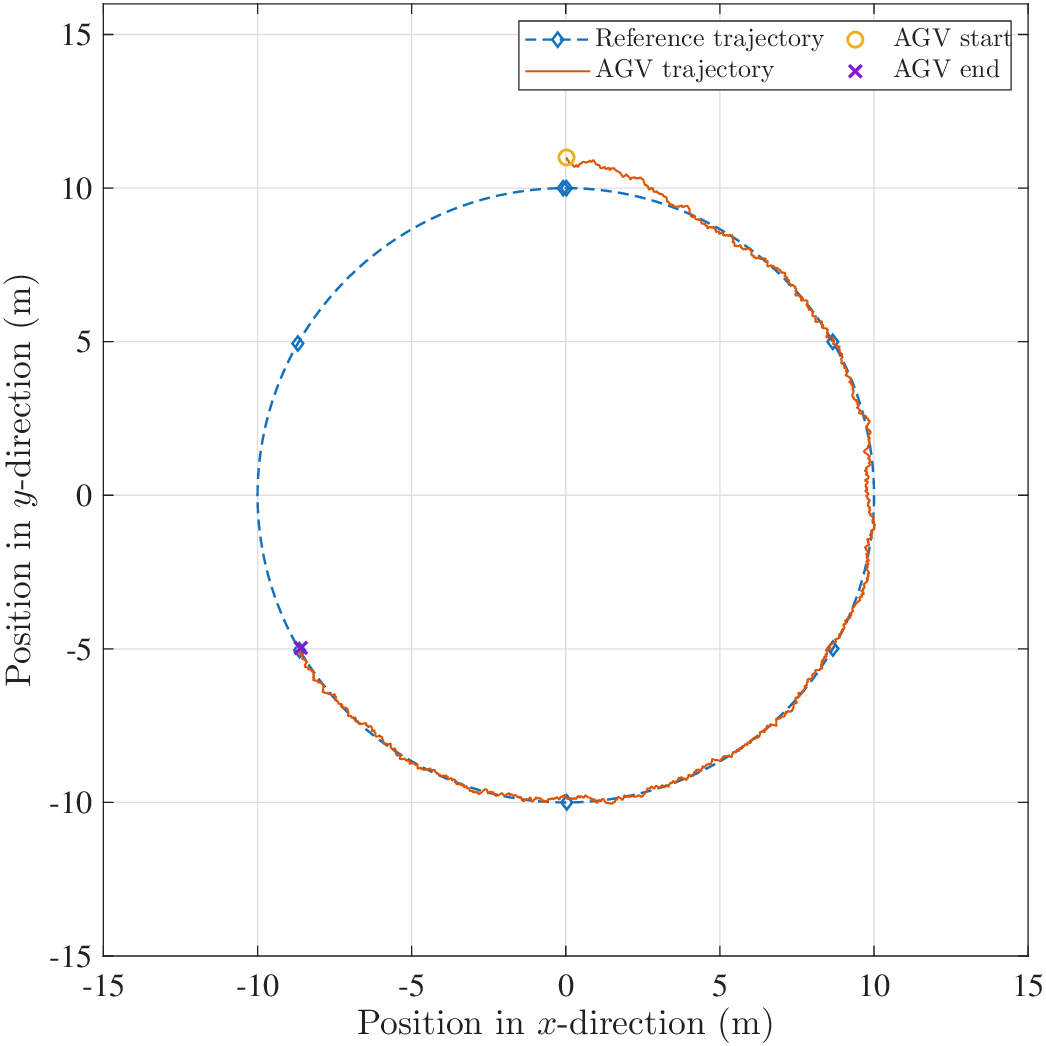}\label{F8: sub_figure2}}
\subfigure[The circle trajectory (T3)]{\includegraphics[width=0.24\hsize,height=0.24\hsize]{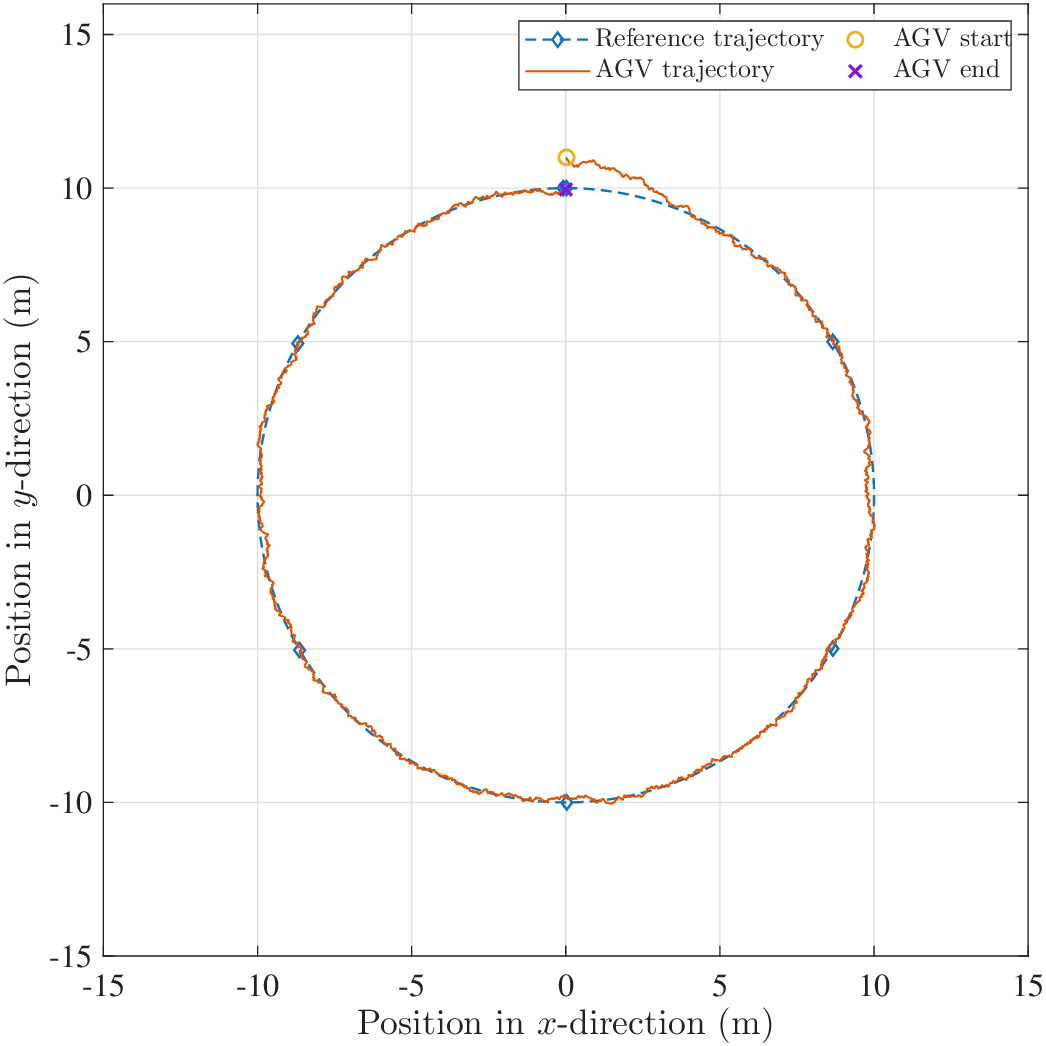}\label{F8: sub_figure3}}
\subfigure[The event-trigger instants]{\includegraphics[width=0.24\hsize,height=0.24\hsize]{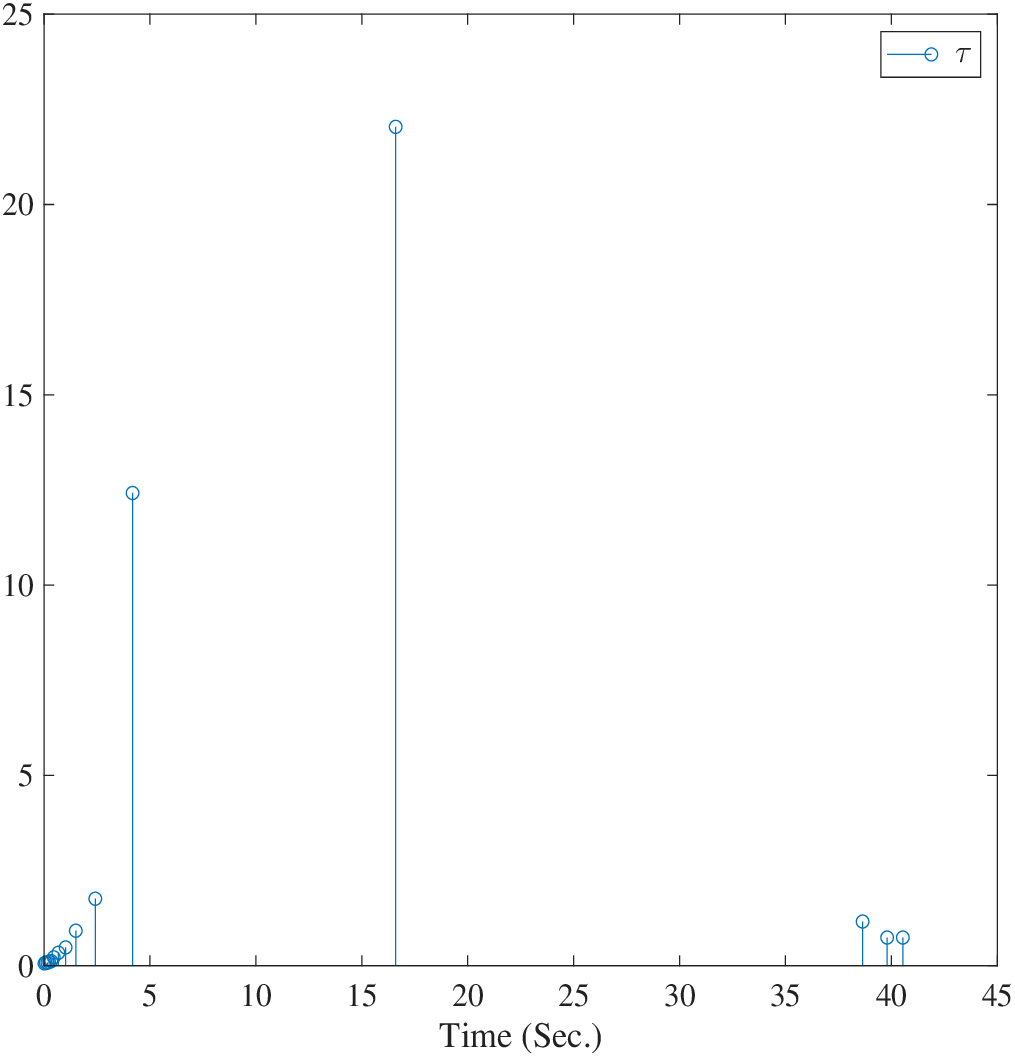}\label{F8: sub_figure4}}
    
\vspace{0.1cm}
\subfigure[The lemniscate trajectory (T1)]{\includegraphics[width=0.24\hsize,height=0.24\hsize]{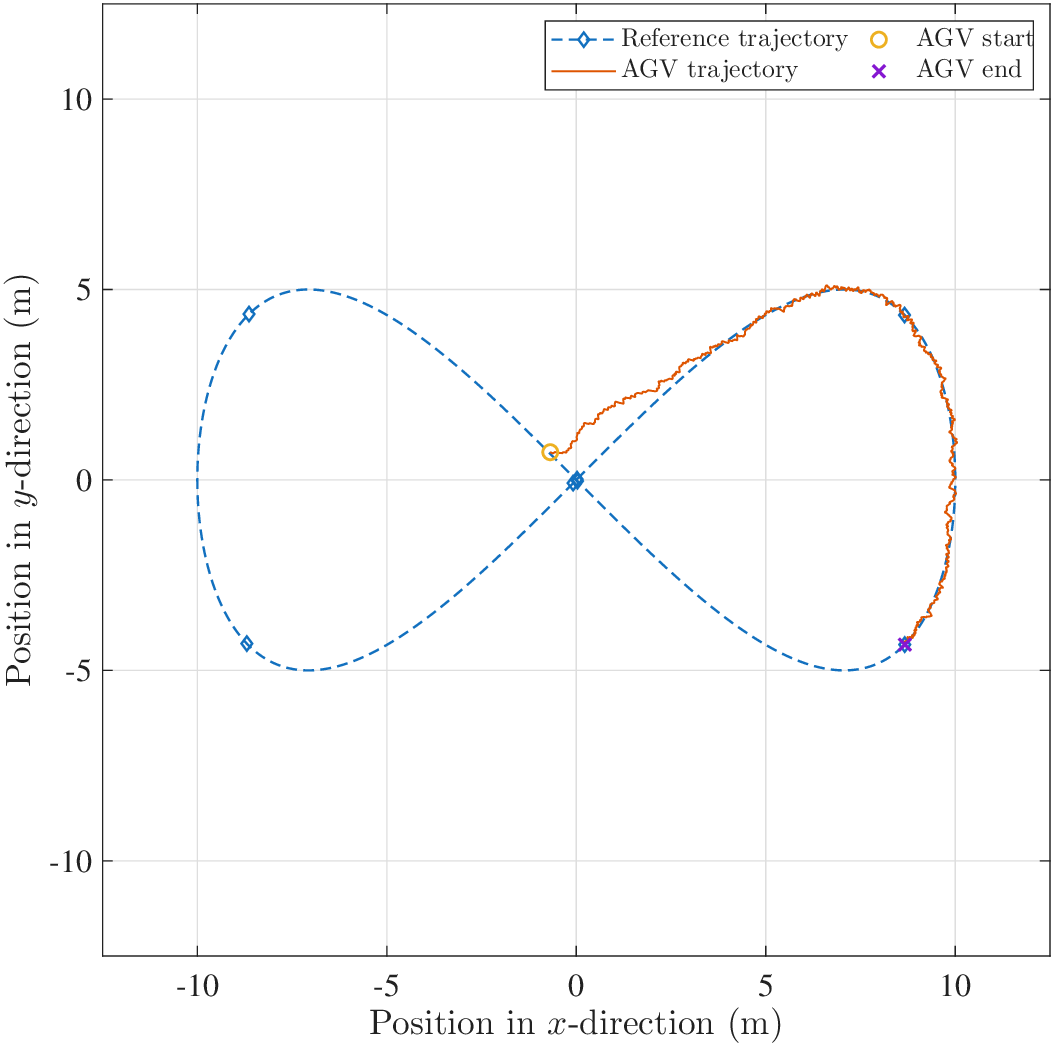}\label{F8: sub_figure5}}
\subfigure[The lemniscate trajectory (T2)]{\includegraphics[width=0.24\hsize,height=0.24\hsize]{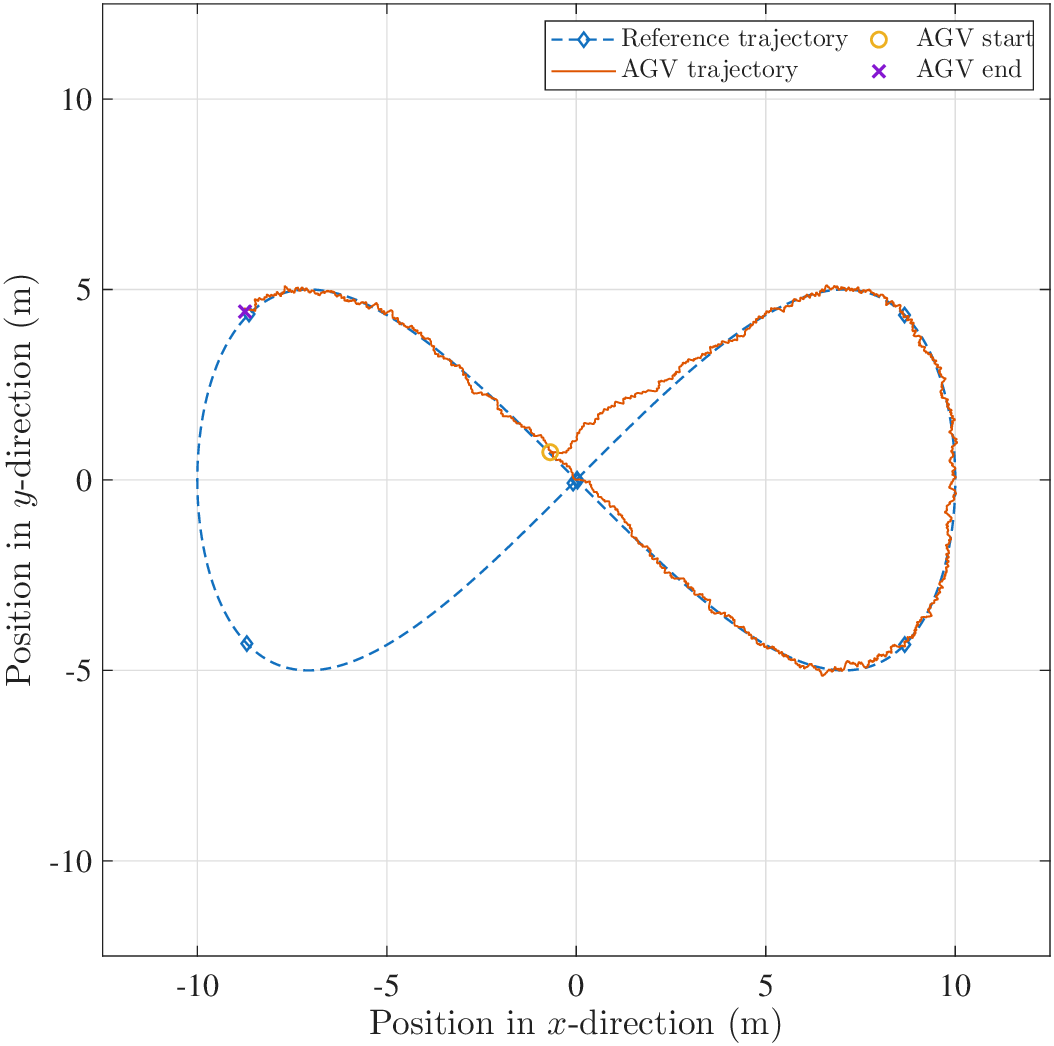}\label{F8: sub_figure6}}
\subfigure[The lemniscate trajectory (T3)]{\includegraphics[width=0.24\hsize,height=0.24\hsize]{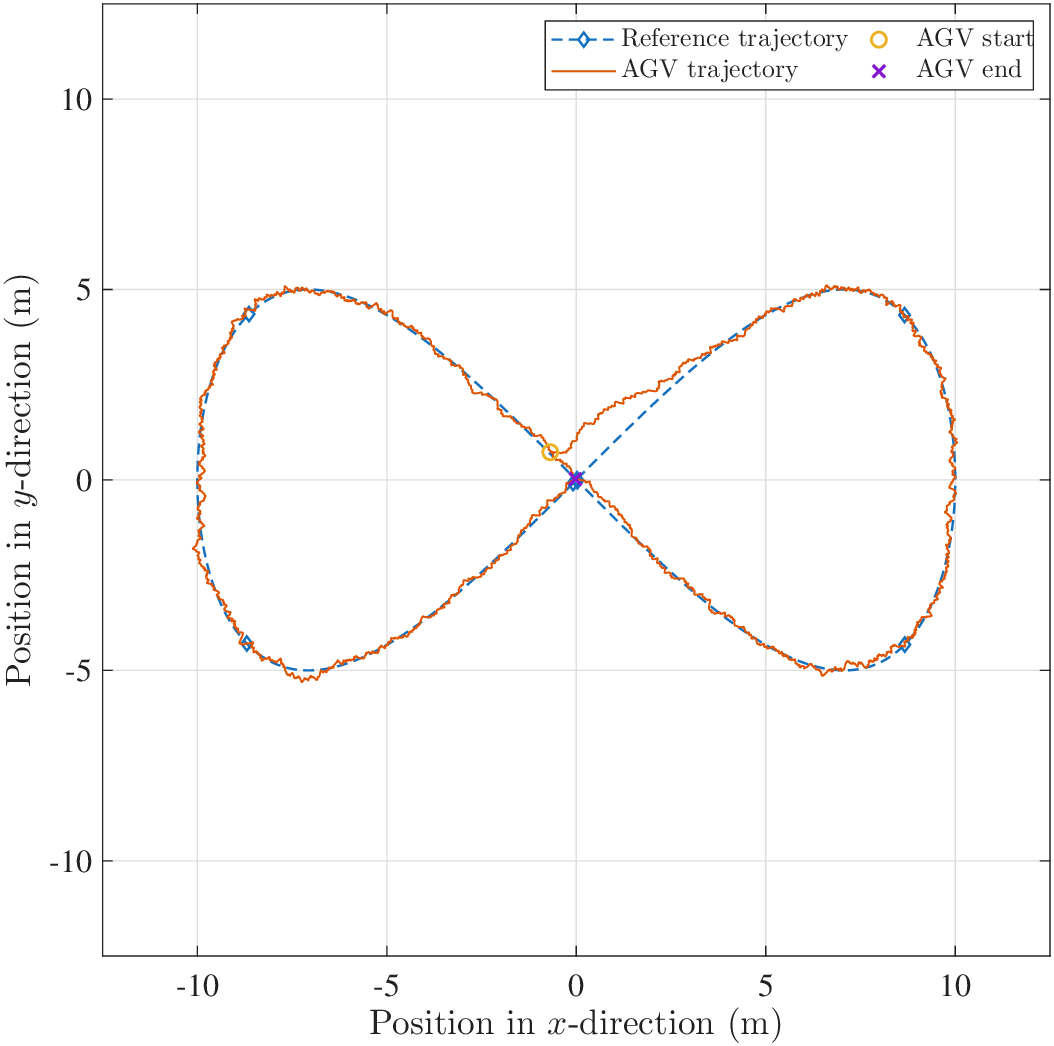}\label{F8: sub_figure7}}
\subfigure[The event-trigger instants]{\includegraphics[width=0.24\hsize,height=0.24\hsize]{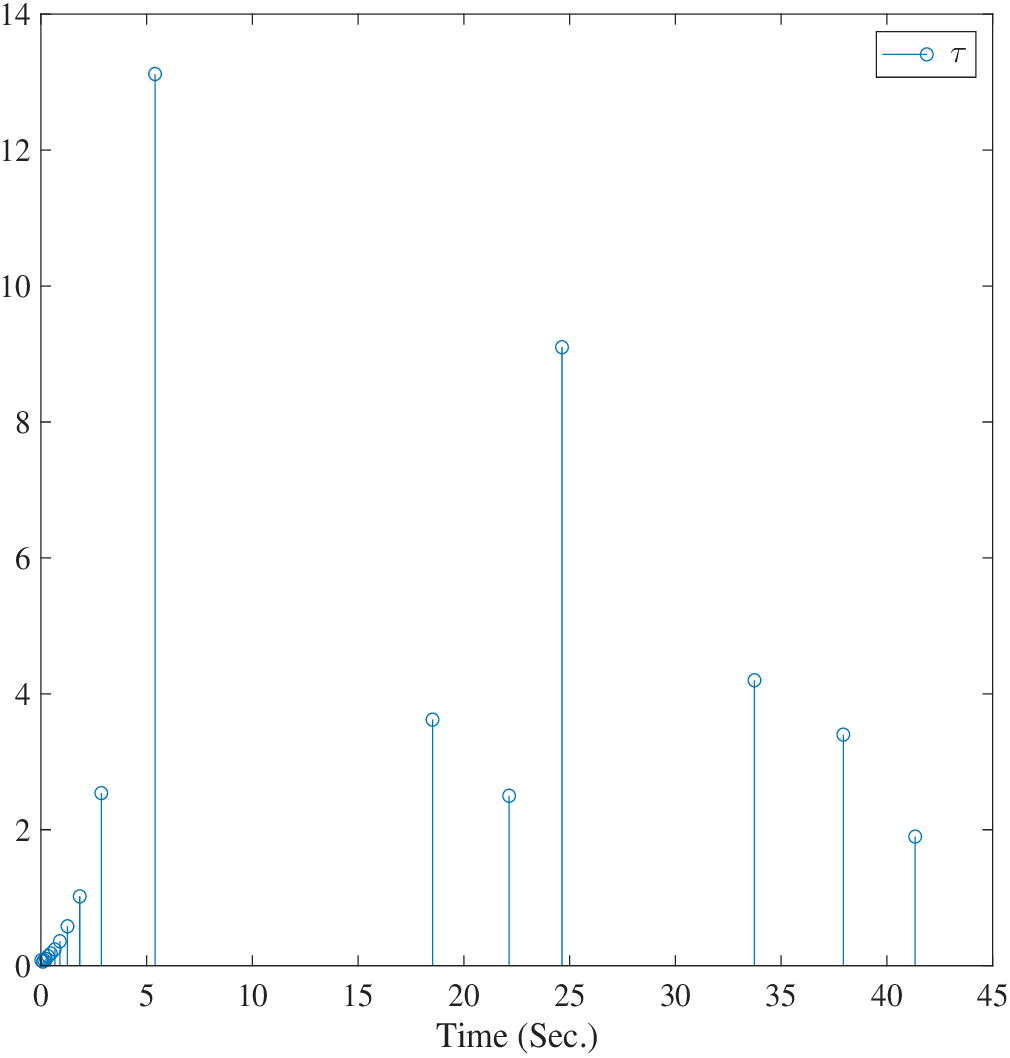}\label{F8: sub_figure8}}
\caption{The robust tracking responses of a wheeled vehicle with the LPV model \eqref{n04}, and the inter-event intervals of the data-driven event-triggered mechanism \eqref{n120}.
Scenario I (circular reference, corresponding to Case 1 in Table \ref{tab1}) and Scenario II (lemniscate reference, corresponding to Case 4 in Table \ref{tab1}).}
\label{F8}
\end{figure*}

Fig. \ref{F8} illustrates the robust LPV tracking responses for the two reference trajectories and the inter-event intervals of the proposed event-triggered mechanism
under different parameter settings. 
For clarity in Fig. \ref{F8}, we divide the simulation horizon into three intervals, that is, 
$\mathrm{T}1[0\hspace*{-0.2em}-\hspace*{-0.2em}15\mathrm{s})$, $\mathrm{T}2[15\hspace*{-0.2em}-\hspace*{-0.2em}30\mathrm{s})$, and 
$\mathrm{T}3[30\hspace*{-0.2em}-\hspace*{-0.2em}45\mathrm{s})$ 
to clearly illustrate the planar tracking responses of the wheeled vehicle.. 

To evaluate the communication efficiency of the event-triggered scheme, we denote the number of triggering instants over
the $N$ sampling instants by $N_{\mathrm{tr}}$. The average inter-event interval and the communication enhancement ratio are defined as
\[
\overline{\Delta k} \triangleq \frac{1}{N_{\mathrm{tr}}-1}
\sum_{i=1}^{N_{\mathrm{tr}}-1} (\hat{k}_{i+1}-\hat{k}_i),
\qquad
\eta \triangleq 1 - \frac{N_{\mathrm{tr}}}{N},
\]
respectively. The communication efficiencies of the event-triggered LPV controller under different parameter settings are compared in 
Table \ref{tab1}. 
The periodic sampling scheme, used as a baseline, satisfies $N_{\mathrm{tr}}=N=2250$ and $\overline{\Delta k}=1$, indicating that the
control input is updated at every sampling instant.
The data-driven ETC reduces the number of triggering instants to 
fewer than $25$ in all tested cases, yielding a communication enhancement ratio of at least $98.98\%$ while ensuring closed-loop 
tracking stability.

\begin{table}[t]
\centering
\caption{Communication Efficiency Comparison of Data-Driven \\ ETC Strategy for the LPV System \eqref{n02}.}
\label{tab1}
\begin{tabular}{lccc}
\toprule
Parameter Settings 
& $N_{\mathrm{tr}}$ 
& $\overline{\Delta k}$ 
& $\eta$ \\
\midrule
\rowcolor{mygray}
Periodic Sampling & $2250$ & $1$ & $0\%$ \\
Case~1($\hat{\mu}=5e4$, $\varepsilon_{4}=0.01$) & $18$ & $121.5294$ & $99.20\%$ \\
\rowcolor{mygray}
Case~2($\hat{\mu}=6e4$, $\varepsilon_{4}=0.01$) & $15$ & $144.7143$ & $99.33\%$ \\
Case~3($\hat{\mu}=7e4$, $\varepsilon_{4}=0.01$) & $16$ & $133.4000$ & $99.29\%$ \\
\rowcolor{mygray}
Case~4($\hat{\mu}=5e4$, $\varepsilon_{4}=0.02$) & $23$ & $89.9545$ & $98.98\%$ \\
Case~5($\hat{\mu}=5e4$, $\varepsilon_{4}=0.03$) & $12$ & $79.8182$ & $99.47\%$ \\
\bottomrule
\end{tabular}
\end{table}
\end{example}

\section{Conclusion}
\label{section5}
This paper has investigated data-driven event-triggered robust stabilization and reference tracking for LPV systems.
First, we established the 
$\theta$-PoE condition for LPV systems 
as a prerequisite for direct data-driven LPV control synthesis. 
 Then, we explored data-based stability certificates for event-triggered LPV systems, 
and established computationally tractable semidefinite programs that can be verified 
at the vertices of the polytopic scheduling set, 
enabling the computation of feedback gains and triggering parameters. Besides, we extended the stabilization method to  
robust LPV reference 
tracking by introducing an auxiliary integral compensator. 
Finally, numerical examples, including a wheeled-vehicle tracking example, were presented to validate the proposed theoretical results.  
Future works will focus on extending the ISS-based tracking framework to
fixed-time stability \cite{WangH2025TAES} and safety preference \cite{WangH2024TCSI}, as well as developing 
a self-triggered communication scheme \cite{HuZ2025Auto} with dissipativity
guarantees \cite{Koch2021Determining}.


\begin{thebibliography}{99}

\bibitem{Mohammadpour2012Control}
J. Mohammadpour and C. W. Scherer,
\emph{Control of Linear Parameter Varying Systems with Applications}. 
Springer Science \& Business Media, 2012.

\bibitem{Martin2023Guarantees}
T. Martin, T. B. Sch$\ddot{\text{o}}$n, and F. Allg$\ddot{\text{o}}$wer, 
"Guarantees for data-driven control of nonlinear systems 
using semidefinite programming: A survey", 
\emph{Annual Reviews in Control}, vol. 56, num. 100911, 2023. 

\bibitem{Ping2022Tube}
X. Ping, J. Yao, B. Ding, and Z. Li,
"Tube-based output feedback robust MPC for LPV systems with scaled terminal constraint sets",
\emph{IEEE Transactions on Cybernetics}, vol. 52, no. 8, pp. 7563-7576, 2022.

\bibitem{Pandey2017A}
A. P. Pandey and M. C. de Oliveira,
"A new discrete-time stabilizability condition for linear parameter-varying systems",
\emph{Automatica}, vol. 79, pp. 214-217, 2017.


\bibitem{Vargas2022Robust}
A. N. Vargas, C. M. Agulhari, R. C. L. F. Oliveira, and V. M. Preciado,
"Robust stability analysis of linear parameter-varying systems with Markov jumps",
\emph{IEEE Transactions on Automatic Control}, vol. 67, no. 11, pp. 6234-6239, 2022.

\bibitem{Cox2018Affine}
P. B. Cox, S. Weiland, and R. T$\acute{\text{o}}$th,
"Affine parameter-dependent Lyapunov functions for LPV systems with affine dependence",
\emph{IEEE Transactions on Automatic Control}, vol. 63, no. 11, pp. 3865-3872, 2018.

\bibitem{Gallegos2024Set}
J. A. Gallegos and K. A. Barbosa,
"Set-theoretical stability analysis of LPV systems via Minkowski-Lyapunov functions",
\emph{IEEE Transactions on Automatic Control}, vol. 70, no. 1, pp. 573-578, 2025.


\bibitem{Mulagaleti2024Parameter}
S. K. Mulagaleti, M. Mejari, and A. Bemporad,
"Parameter-dependent robust control invariant sets for LPV systems with bounded parameter-variation rate",
\emph{IEEE Transactions on Automatic Control}, vol. 70, no. 2, pp. 1259-1266, 2025.

\bibitem{Cheng2021Modeling}
J. Cheng, M. Wu, F. Wu, C. Lu, X. Chen, and W. Cao, 
"Modeling and control of drill-string system with stick-slip 
vibrations using LPV technique", 
\emph{IEEE Transactions on Control Systems Technology}, 
vol. 29, no. 2, pp. 718-730, 2021. 

\bibitem{Polcz2019Passivity}
P. Polcz, B. Kulcsar, T. Peni, and G. Szederkenyi, 
"Passivity analysis of rational LPV systems using Finsler's lemma", 
in \emph{IEEE 58th Conference on Decision and Control}, 
pp. 3793-3798, 2019. 

\bibitem{Souza2006Robust}
C. E. de Souza, K. A. Barbosa, and A. T. Neto, 
"Robust $\mathcal{H}_{\infty}$ filtering for discrete-time 
linear systems with uncertain time-varying parameters", 
\emph{IEEE Transactions on Signal Processing}, 
vol. 54, no. 6, pp. 2110-2118, 2006. 

\bibitem{Meijer2024Certificates}
T. J. Meijer, V. Dolk, and W. P. M. H. Heemels,
"Certificates of nonexistence for analyzing stability, stabilizability and detectability of LPV systems",
\emph{Automatica}, vol. 170, num. 111841, 2024.

\bibitem{Han2025Event}
Z. Han, T. Pan, and D. Xu, 
"Event-triggered model-free adaptive integral sliding mode constrained maximum power point tracking control 
strategy for photovoltaic system", 
\emph{International Journal of Innovative Computing, Information and Control}, vol. 21, no. 3, pp. 649-665, 2025. 

\bibitem{Cai2025Dynamic}
G. Cai, T. Wu, M. Hao, H. Liu, and B. Zhou, 
"Dynamic event-triggered gain-scheduled $H_{\infty}$ control 
for a polytopic LPV model of morphing aircraft", 
\emph{IEEE Transactions on Aerospace and Electronic Systems}, 
vol. 61, no. 1, pp. 93-106, 2025. 







\bibitem{Willems2005A}
J. C. Willems, P. Rapisarda, I. Markovsky, and B. L. M. De Moor, 
"A note on persistency of excitation", \emph{Systems \& Control Letters}, vol. 54, pp. 325-329, 2005.

\bibitem{He2025From}
K. He, S. Shi, T. van den Boom, and B. De Schutter, 
"From learning to safety: A direct data-driven framework 
for constrained control", 
\emph{arXiv Preprint}, arXiv: 2505.15515v1. 

\bibitem{Liu2024Learning}
W. Liu, G. Wang, J. Sun, F. Bullo, and J. Chen, 
"Learning robust data-based LQG controllers from noisy data", 
\emph{IEEE Transactions on Automatic Control}, 
vol. 69, no. 12, pp. 8526-8538, 2024. 


\bibitem{Persis2020Formulas}
C. De Persis and P. Tesi,
"Formulas for data-driven control: Stabilization, optimality, and robustness",
\emph{IEEE Transactions on Automatic Control}, vol. 65, no. 3, pp. 909-924, 2020.

\bibitem{Dorfler2023On}
F. D$\ddot{\text{o}}$rfler, P. Tesi, and C. De Persis,
"On the certainty-equivalence approach to direct data-driven LQR design",
\emph{IEEE Transactions on Automatic Control},
vol. 68, no. 12, pp. 7989-7996, 2023. 

\bibitem{Berberich2021Data}
J. Berberich, J. K$\ddot{\text{o}}$hler, 
M. A. M$\ddot{\text{u}}$ller, and F. Allg$\ddot{\text{o}}$wer, 
"Data-driven model predictive control with stability and 
robustness guarantees", 
\emph{IEEE Transactions on Automatic Control}, 
vol. 66, no. 4, pp. 1702-1717, 2021. 

\bibitem{Ren2024Event}
Z. Ren, H. Liu, G. Wen, and J. L$\ddot{\text{u}}$, 
"Event-triggered data-driven security formation control 
for quadrotors under denial-of-service attacks and 
communication faults", 
\emph{IEEE Transactions on Cybernetics}, 
vol. 55, no. 11, pp. 5226-5236, 2025. 

\bibitem{Wang2024Periodic}
J. Wang, J. Sun, J. Yang, and S. Li, 
"Periodic event-triggered model predictive control for 
networked nonlinear uncertain systems with disturbances", 
\emph{IEEE Transactions on Cybernetics}, 
vol. 54, no. 12, pp. 7501-7513, 2024. 

\bibitem{Ma2021Sparse}
R. Ma, P. Shi, and L. Wu, 
"Sparse false injection attacks reconstruction via 
descriptor sliding mode observers", 
\emph{IEEE Transactions on Automatic Control}, 
vol. 66, no. 11, pp. 5369-5376, 2021. 

\bibitem{An2024Dynamic}
T. An, B. Dong, H. Yan, L. Liu, and B. Ma, 
"Dynamic event-triggered strategy-based optimal control of 
modular robot manipulator: 
A multiplayer nonzero-sum game perspective", 
\emph{IEEE Transactions on Cybernetics}, 
vol. 54, no. 12, pp. 7514-7526, 2024. 

\bibitem{Sen2024On}
M. De la Sen, "On some useful links between external positivity, Kalman-Yakubovich-Popov lemma and positive real lemma for linear time-invariant systems with constant point delays", 
\emph{International Journal of Innovative Computing, Information and Control}, vol. 20, no. 5, pp. 1415-1438, 2024. 

\bibitem{Rotulo2022Online}
M. Rotulo, C. De Persis, and P. Tesi,
"Online learning of data-driven controllers for unknown switched linear systems",
\emph{Automatica}, vol. 145, num. 110519, 2022.

\bibitem{Berberich2023Combining}
J.Berberich, C. W. Scherer, and F. Allg$\ddot{\text{o}}$wer, 
"Combining prior knowledge and data for robust 
controller design", 
\emph{IEEE Transactions on Automatic Control}, 
vol. 68, no. 8, pp. 4618-4633, 2023. 

\bibitem{Waarde2022From}
H. J. van Waarde, M. K. Camlibel, and M. Mesbahi,
"From noisy data to feedback controllers: 
Nonconservative design via a matrix S-lemma",
\emph{IEEE Transactions on Automatic Control},
vol. 67, no. 1, pp. 162-175, 2022.

\bibitem{Bisoffi2022Data}
A. Bisoffi, C. De Persis, and P. Tesi, 
"Data-driven control via Petersen's lemma", \emph{Automatica}, vol. 145, num. 110537, 2022. 

\bibitem{Persis2023Learning}
C. De Persis, M. Rotulo, and P. Tesi,
"Learning controllers from data via approximate nonlinearity cancellation",
\emph{IEEE Transactions on Automatic Control}, vol. 68, no. 10, pp. 6082-6097, 2023.

\bibitem{Dai2021A}
T. Dai and M. Sznaier, 
"A semi-algebraic optimization approach to 
data-driven control of continuous-time 
nonlinear systems", 
\emph{IEEE Control Systems Letters}, vol. 5, no. 2, 
pp. 487-492, 2021. 

\bibitem{Koch2022Provably}
A. Koch, J. Berberich, and F. Allg$\ddot{\text{o}}$wer,
"Provably robust verification of dissipativity properties from data",
\emph{IEEE Transactions on Automatic Control}, vol. 67, no. 8, pp. 4248-4255, 2022.


\bibitem{Bisoffi2022Learning}
A. Bisoffi, C. De Persis, and P. Tesi,
"Learning controllers for performance through LMI regions",
\emph{IEEE Transactions on Automatic Control}, 
vol. 68, no. 7, pp. 4351-4358, 2023. 

\bibitem{Li2024Data}
L. Li, C. De Persis, P. Tesi, and N. Monshizadeh,
"Data-based transfer stabilization in linear systems",
\emph{IEEE Transactions on Automatic Control},
vol. 69, no. 3, pp. 1866-1873, 2024. 

\bibitem{Strasser2024Koopman}
R. Str$\ddot{\text{a}}$sser, M. Schaller, K. Worthmann, J. Berberich, and F. Allg$\ddot{\text{o}}$wer,
"Koopman-based feedback design with stability guarantees",
\emph{IEEE Transactions on Automatic Control}, 
vol. 70, no. 1, pp. 355-370, 2025. 

\bibitem{Koch2021Determining}
A. Koch, J. Berberich, J. K$\ddot{\text{o}}$hler, and F. Allg$\ddot{\text{o}}$wer,
"Determining optimal input-output properties: A data-driven approach",
\emph{Automatica}, vol. 134, num. 109906, 2021. 

\bibitem{Bianchi2025Data}
M. Bianchi, S. Grammatico, and J. Cort$\acute{\text{e}}$s,
"Data-driven stabilization of switched and constrained linear systems",
\emph{Automatica}, vol. 171, num. 111974, 2025.

\bibitem{Disaro2024On}
G. Disar$\grave{\text{o}}$ and M. E. Valcher,
"On the equivalence of model-based and data-driven 
approaches to the design of unknown-input observers",
\emph{IEEE Transactions on Automatic Control}, 
vol. 70, no. 3, pp. 2074-2081, 2025. 

\bibitem{Persis2024Event}
C. De Persis, R. Postoyan, and P. Tesi, 
"Event-triggered control from data", 
\emph{IEEE Transactions on Automatic Control}, 
vol. 69, no. 6, pp. 3780-3795, 2024.

\bibitem{Zhou2024Fuzzy}
H. Zhou, Y. Zuo, and S. Tong, 
"Fuzzy adaptive event-triggered consensus control for 
nonlinear multiagent systems under jointly connected 
switching networks", 
\emph{IEEE Transactions on Cybernetics}, 
vol. 54, no. 12, pp. 7163-7172, 2024. 

\bibitem{Liu2024Decentralized}
C. Liu, Z. Chu, Z. Duan, H. Zhang, and Z. Ma, 
"Decentralized event-triggered tracking control for 
unmanned interconnected systems via particle swarm optimization-based 
adaptive dynamic programming", 
\emph{IEEE Transactions on Cybernetics}, 
vol. 54, no. 11, pp. 6895-6909, 2024. 

\bibitem{Wildhagen2022Data}
S. Wildhagen, J. Berberich, M. Hertneck, and F. Allg$\ddot{\text{o}}$wer,
"Data-driven analysis and controller design for 
discrete-time systems under aperiodic sampling",
\emph{IEEE Transactions on Automatic Control}, 
vol. 68, no. 6, pp. 3210-3225, 2023. 


\bibitem{Verhoek2023Data}
C. Verhoek, J. Berberich, S. Haesaert, F. Allg$\ddot{\text{o}}$wer, and R. T$\acute{\text{o}}$th,
"Data-driven dissipativity analysis of 
linear parameter-varying systems",
\emph{IEEE Transactions on Automatic Control}, 
vol. 69, no. 12, pp. 8603-8616, 2024. 

\bibitem{Verhoek2024Direct}
C. Verhoek, R. T$\acute{\text{o}}$th, and H. S. Abbas, 
"Direct data-driven state-feedback control of linear parameter-varying systems", 
\emph{International Journal of Robust and Nonlinear Control}, vol. 35, no. 16, pp. 6955-6977, 2025.  

\bibitem{Verhoek2024The}
C. Verhoek, I. Markovsky, S. Haesaert, and 
R. T$\acute{\text{o}}$th, 
"A behavioral approach for LPV data-driven representations", 
\emph{IEEE Transactions on Automatic Control}, vol. 71, no. 3, pp. 1616-1629, 2026.

\bibitem{Abbas2016A}
H. S. Abbas, R. T$\acute{\text{o}}$th, 
N. Meskin, J. Mohanmadpour, and J. Hanema, 
"A robust MPC for input-output LPV models", 
\emph{IEEE Transactions on Automatic Control}, 
vol. 61, no. 12, pp. 4183-4188, 2016. 

\bibitem{Golabi2016Event}
A. Golabi, N. Meskin, R. T$\acute{\text{o}}$th, 
J. Mohammadpour, and T. Donkers, 
"Event-triggered control for discrete-time linear 
parameter-varying systems", 
in \emph{2016 American Control Conference}, 
pp. 3680-3685, 2016. 



\bibitem{Clouson2022Robust}
J. Clouson, H. van Waarde, and F. D$\ddot{\text{o}}$rfler, 
"Robust fundamental lemma for data-driven control", 
\emph{arXiv Preprint},
arXiv: 2205.06636. 

\bibitem{Digge2022Data}
V. Digge and R. Pasumarthy, 
"Data-driven event-triggered control for discrete-time 
LTI systems", 
in \emph{2022 European Control Conference}, pp. 1355-1360,

\bibitem{Scherer2001LPV}
C. W. Scherer, 
"LPV control and full block multipliers", \emph{Automatica}, vol. 37, pp. 361-375, 2001. 


\bibitem{Lofberg2004Yalmip}
J. L$\ddot{\text{o}}$fberg, "YALMIP: A toolbox 
for modeling and optimization in MATLAB", 
in \emph{2004 IEEE International Confernece on Robotics 
and Automation}, pp. 284-289, 2004. 

\bibitem{Aps2019Mosek}
M. ApS, Mosek Optimization Toolbox for Matlab, 
\emph{User's Guide and Reference Manual, Version 4}, 
[Online] https://www.mosek.com/, 2019.

\bibitem{WangH2025TAES}
H. Wang, M. Li, "Adaptive fuzzy nonsingular fixed-time safety flight containment control for multi-UAVs", \emph{IEEE Transactions on Aerospace and Electronic Systems}, vol. 61, no. 1, pp. 802-812, 2025.
\bibitem{WangH2024TCSI}
H. Wang, Z. Meng, J. Qiao and S. Liu, "Command filtered-based adaptive predefined-time control for uncertain nonlinear systems with applications to RLC circuit", \emph{IEEE Transactions on Circuits and Systems I: Regular Papers}, vol.71, no.10, pp.4792-4801, 2024.
\bibitem{HuZ2025Auto}
R. Ma, Z. Hu, R. Yang, and L. Wu, "Aperiodic-sampled neural network controllers with closed-loop stability verification", 
\emph{Automatica}, vol. 183, num. 112573, 2026.

\bibitem{Morro2011TRO}
A. Morro, A. Sgorbissa, and R. Zaccaria,
"Path following for unicycle robots with an arbitrary path curvature",
\emph{IEEE Transactions on Robotics},
vol. 27, no. 5, pp. 1016-1023, 2011.

\bibitem{Xu2020TITS}
S. Xu, H. Peng, and Y. Tang,
"Preview path tracking control with delay compensation for autonomous vehicles",
\emph{IEEE Transactions on Intelligent Transportation Systems},
vol. 22, no. 5, pp. 2979-2989, 2021.

\bibitem{Kalaria2022TITS}
D. Kalaria, Q. Lin, and J. M. Dolan,
"Delay-aware robust control for safe autonomous driving and racing",
\emph{IEEE Transactions on Intelligent Transportation Systems}, vol. 25, no. 7, pp. 7140-7150, 2024.

\bibitem{Zeilinger2016Auto}
C. Conte, C. N. Jones, M. Morari and M. N. Zeilinger, "Distributed synthesis and stability of cooperative distributed model
predictive control for linear systems", \emph{Automatica}, vol. 69, pp. 117-125, 2016.

\bibitem{He2021CSL}
H. Yin, P. Seiler, M. Jin and M. Arcak, "Imitation learning with stability and safety guarantees", \emph{IEEE Control Systems Letters}, Vol. 6, pp. 409-414, 2021.
\end{thebibliography}
\end{document}